%% file: main.tex
\documentclass[submission,copyright,creativecommons]{eptcs}
\providecommand{\event}{23rd International Conference on Quantum Physics and Logic (QPL 2026)}
\usepackage{amsfonts}
\usepackage{amsthm}
\usepackage{amsmath}
\allowdisplaybreaks
\usepackage{amssymb}
\usepackage{bm}
\usepackage{blkarray}
\usepackage[dvipsnames]{xcolor}
\usepackage{graphicx} 
\usepackage{subcaption}
\usepackage{mathtools}
\usepackage{underscore} 
\usepackage{tikz} 
\usepackage{tikzit}
\usepackage{nicematrix}
\NiceMatrixOptions{
code-for-first-row = \color{blue}\scriptstyle ,
code-for-last-row = \color{blue}\scriptstyle ,
code-for-first-col = \color{blue}\scriptstyle ,
code-for-last-col = \color{blue}\scriptstyle
}

\NiceMatrixOptions{
custom-line = {
    letter = I ,
    command = hddashedline ,
    ccommand = cddashedline ,
    tikz = dashed
}
}

\usepackage{newtxtext,newtxmath}

\usepackage[shortlabels, inline]{enumitem}

\hypersetup{
    colorlinks=true,
    linkcolor=blue,
    filecolor=blue,
    urlcolor=blue,
    citecolor=blue
    }
\usepackage{crossreftools}
\usepackage{footnotebackref} 
\usetikzlibrary{quantikz2} 
\usepackage{algorithm}
\usepackage{algpseudocodex}
\usepackage{stmaryrd}
\usepackage{anyfontsize}
\usepackage{dsfont} 
\usepackage{multicol}
\usepackage{xspace}
\usepackage{lipsum}
\usepackage{booktabs}
\usepackage{pifont}
\usepackage{placeins}

\usepackage[T1]{fontenc} 

\usepackage{orcidlink} 

\usepackage[colorinlistoftodos,prependcaption,textsize=tiny]{todonotes}
\usepackage{xspace} 

\usepackage[final,inline]{showlabels}

\input{thesisstyles.tikzstyles}
\input{zx.tikzstyles}

\graphicspath{ {./img/} }

\makeatletter
\newcommand{\optionaldesc}[2]{%
  \phantomsection
  #1\protected@edef\@currentlabel{#1}\label{#2}%
}
\makeatother

\theoremstyle{definition}
\newtheorem{theorem}{Theorem}[section]
\newtheorem{corollary}[theorem]{Corollary}
\newtheorem{lemma}[theorem]{Lemma}
\newtheorem*{lemma*}{Lemma}
\newtheorem*{proposition*}{Proposition}
\newtheorem*{remark*}{Remark}
\newtheorem{proposition}[theorem]{Proposition}

\newtheorem{definition}[theorem]{Definition}

\newtheorem{remark}[theorem]{Remark}
\newtheorem{observation}[theorem]{Observation}

\newtheorem*{rep@theorem}{\rep@title}
\newcommand{\newreptheorem}[2]{%
	\newenvironment{rep#1}[1]{%
    \def\rep@title{#2 \ref{##1} (restated)}%
		\begin{rep@theorem}}%
		{\end{rep@theorem}}}
\newreptheorem{thm}{Theorem}

  {\list{}{\leftmargin=#1\rightmargin=#1}\item[]}%
  {\endlist}

\newcommand{\problem}[1]{\mathbf{#1}}
\newcommand{\laproblem}{\problem{Inverse+DAG}}
\newcommand{\bigO}{\mathcal{O}} 
\newcommand{\rowbasechange}[2]{{#1}_{#2}}
\newcommand{\colbasechange}[2]{{#1}^{#2}}
\newcommand{\starpart}{\mathcal{P}}
\newcommand{\solved}{S}
\newcommand{\tosolve}{L}
\newcommand{\LS}{K_{LS}}
\newcommand{\ILS}{K_{ILS}}

\DeclareMathOperator{\rank}{rank}

\newcommand\scalemath[2]{\scalebox{#1}{\mbox{\ensuremath{\displaystyle #2}}}} 
\newcommand{\scalediag}[2][1.25]{\scalebox{#1}{\input{#2.tikz}}} 
\DeclareMathOperator{\support}{supp} 
\DeclareMathOperator{\vspan}{span} 
\DeclareMathOperator{\poly}{poly} 
\newcommand{\supp}[1]{\support\left(#1\right)} 
\newcommand{\trl}{\triangleleft} 
\newcommand{\comp}[1]{\bar{#1}} 
\newcommand{\ld}{\lambda}
\newcommand{\sse}{\subseteq}
\newcommand{\abs}[1]{\left| #1 \right|} 
\newcommand{\FF}{\mathbb{F}}
\newcommand{\ZZ}{\mathbb{Z}}
\newcommand{\CC}{\mathbb{C}}
\newcommand{\cM}{\mathcal{M}}

\newcommand{\somevector}{\vec{\mathscr{a}}}

\newcommand{\Xlike}{\mathcal{X}}
\newcommand{\Zlike}{\mathcal{Z}}

\newcommand{\Fd}{\FF_d}
\newcommand{\flow}{$\Fd$-flow\xspace}
\newcommand{\flowintitles}{\texorpdfstring{\flow}{F\_d-flow\xspace}}
\newcommand{\LOG}{labelled open graph\xspace}
\newcommand{\LOFG}{labelled open $\Fd$-graph\xspace}
\newcommand{\OFG}{open $\Fd$-graph\xspace}

\DeclareMathOperator{\Id}{Id}
\newcommand{\id}[2]{\Id_{#1}^{#2}} 
\newcommand{\ldX}{\ld_X}
\newcommand{\ldZ}{\ld_Z}
\DeclareMathOperator{\Diag}{Diag}
\newcommand{\diag}[1]{\Diag\left(#1\right)}

\newcommand{\conlab}[1]{({\crtcrossreflabel{#1}[#1]})} 
\newcommand{\conref}[1]{\eqref{#1}} 

\newcommand{\problemstatement}[3]{
\begin{quote}
$#1$\\
\textbf{Input:} #2\\
\textbf{Output:} #3
\end{quote}
}

\newcommand{\multiplier}[1]{\mathcal{R}^{(#1)}}
\newcommand{\phasegate}{R}
\newcommand{\ketbra}[2]{\ket{#1}\!\bra{#2}}
\newcommand{\intf}[1]{\left\llbracket #1 \right\rrbracket} 

\renewcommand{\t}[1]{\ensuremath{^{\otimes #1}}}

\renewcommand{\conlab}[1]{{\crtcrossreflabel{(#1)}[#1]}} 
\renewcommand{\conref}[1]{\ref{#1}} 

\makeatletter
\newcommand\etc{etc\@ifnextchar.{}{.\@}\xspace}
\newcommand\ie{i.e.\@\xspace}
\newcommand\eg{e.g.\@\xspace}
\newcommand\cf{cf.\@\xspace}
\makeatother

\usepackage{cleveref}
\makeatletter

\renewcommand\theHALG@line{\thealgorithm.\arabic{ALG@line}}
\makeatother

\makeatletter
\newenvironment{breakablealgorithm}
  {
   \begin{center}
     \refstepcounter{algorithm}
     \hrule height.8pt depth0pt \kern2pt
     \renewcommand{\caption}[2][\relax]{
       {\raggedright\textbf{\fname@algorithm~\thealgorithm} ##2\par}%
       \ifx\relax##1\relax 
         \addcontentsline{loa}{algorithm}{\protect\numberline{\thealgorithm}##2}%
       \else 
         \addcontentsline{loa}{algorithm}{\protect\numberline{\thealgorithm}##1}%
       \fi
       \kern2pt\hrule\kern2pt
     }
  }{
     \kern2pt\hrule\relax
   \end{center}
  }
\makeatother

\usepackage{calc}

\newcommand{\fixedaligneqq}[1]{%
  \mathrel{%
    \mathop{=}%
    \limits^{%
      \makebox[\widthof{$\scriptstyle\ref{lem:H-box-push}$}][c]{$\scriptstyle #1$}%
    }%
  }%
}

\title{Working With Measurement-Based Computations on Qudits}
\author{Piotr Mitosek \orcidlink{0009-0000-9214-5714}
\institute{Institute of Theoretical Physics\\ Leibniz University Hannover\\ Germany}
\email{piotr.mitosek(at)itp.uni-hannover.de}
\and
Miriam Backens \orcidlink{0000-0002-5418-1084}
\institute{Universit\'e de Lorraine, CNRS, Inria, LORIA \\ F-54000 Nancy\\ France}
}
\def\titlerunning{Working With Measurement-Based Computations on Qudits}
\def\authorrunning{P. Mitosek \& M. Backens}

\begin{document}

\maketitle

\begin{abstract}
    Measurement-based quantum computing is a universal model of quantum computation in which successive product measurements of an entangled resource state drive the computation. The non-deterministic nature of measurements necessitates adaptivity to ensure an overall deterministic computation. Flow structures characterise cases in which such an adaptive correction procedure is possible. Recently, flow has been defined in a setting where the resource states are prime-dimensional qudit graph states rather than the usual qubit graph states. Yet, this qudit flow definition is more burdensome to work with than analogous definitions for qubits.

    Here, we give a simpler definition of qudit flow and consider various useful properties of this flow, drawing on results for the qubit case. In particular, we show how to focus qudit flow and argue that focused flow is canonical. We improve the previous algebraic formulation to capture focused flow and use it to obtain an $\bigO(n^3)$ flow-finding algorithm (where $n$ is the number of qudits), matching the best known complexity for qubit flows and improving on the previous $\bigO(n^4)$ result for qudits. Furthermore, we explore multiple flow-preserving transformations, thus opening a pathway to using flow for optimisation. These transformations include pivoting, removal and insertion of certain types of vertices, and reversibility of flow. Lastly, we propose an algorithmic approach to generating large qudit computations with flow, for testing or machine learning.
\end{abstract}

\section{Introduction}

In computer science, it is standard to encode information in bits, i.e.\ using binary values 0 or~1.
This has carried across to quantum computing, where most research focuses on qubits, aka `quantum bits', that have two classical states (and superpositions thereof).
True qubit systems exist in nature,
nevertheless, many of the most promising candidates for physical qubits naturally have more than two basic states,
and there is increasing interest in quantum computation using $d$-dimensional systems, called qudits.

Much of the theoretical research about qudit quantum computation focuses on the quantum circuit model \cite{gheorghiuStandardFormQudit2014,heyfronQuantum2019,yangQuantum2025}.
This means a computation is driven entirely by unitary gates; state preparation and measurements can be restricted to the computational basis and the latter are needed only for reading out information at the end.
The qudit Clifford fragment for prime dimensions $d$ is core to qudit quantum computing analogous to its role with respect to qubits \cite{hostensStabilizer2005,gheorghiuStandardFormQudit2014}.
There are definitions of qudit stabiliser states and also qudit graph states \cite{bahramgiriGraph2006}, which in dimension $d$ correspond to $\FF_d$-weighted graphs: where in the binary case, edges are either present or not, in an $\FF_d$-weighted graph, each pair of vertices is associated with a weight $w\in\FF_d$ such that $w=0$ corresponds to there being no edge between the two vertices.

Qudit graph states then permit the generalisation of the one-way model of measurement-based quantum computation (MBQC) \cite{raussendorfOneWayQuantumComputer2001} to prime-dimensional systems \cite{paesaniScheme2021,boothMeasurementbasedQuantumComputation2022,boothOutcomeDeterminismMeasurementbased2023,romanovaMeasurement-based2026}.
In the one-way model, an entangled resource graph state is prepared at the beginning; this resource can be independent of the computation (as long as it is large enough).
The computation proceeds by successive adaptive single qudit measurements in different \emph{measurement spaces} (called measurement planes in the qubit case).
The individual measurements change the underlying state and can thereby drive a computation, yet they are not generally deterministic.
This is why the `successive adaptive' part matters: If the combination of underlying graph state and measurement spaces for each qudit (formalised as a `\LOFG') is chosen well, then it is possible to modify later measurements depending on the outcomes of earlier measurements so as to make the overall computation deterministic.
Whether overall determinism is possible is characterised by a property called \flow\footnote{This property was previously called $\mathbb{Z}_d$-flow, yet it actually relies on properties of a field, not just a ring. We change the notation and terminology to emphasise this.} \cite{boothOutcomeDeterminismMeasurementbased2023}, the qudit equivalent of `extended gflow' on qubits \cite{browneGeneralizedFlowDeterminism2007}.
Yet the current definition of \flow is somewhat unwieldy in practice as one validity condition requires considering any totalisation of the partial order that is part of the \flow specification.

In this work, we develop a different perspective on \flow on qudits, which adapts the `algebraic Pauli flow' formulation \cite{mitosekAlgebraicInterpretationPauli2026} to qudits and is more straightforward to work with.
This then allows several results from qubit gflow to be reproduced for qudits.
In particular, we adapt the qubit gflow concept of `focusing' \cite{backensThereBackAgain2021} to \flow and show that the existence of \flow is equivalent to the existence of focused \flow.
Next, we directly generalise previous algebraic formulations of qubit flow \cite{mhallaWhichGraphStates2014a,mitosekPauliFlowOpen2024,mitosekAlgebraicInterpretationPauli2026}: given a labelled open graph, we show simple linear algebra conditions allowing verification of the existence of \flow.
We also show how to lift the definition of and core results about qubit `focused sets' \cite{simmonsRelatingMeasurementPatterns2021,mitosekAlgebraicInterpretationPauli2026} to qudit \emph{focused vectors} -- these structures classify `do nothing' stabilizers, \ie stabilizers that do not impact the correction procedure.
We use the new algebraic interpretation to devise new fast algorithms for finding \flow, improving previous $\bigO(n^4)$ routine \cite{boothOutcomeDeterminismMeasurementbased2023} to $\bigO(n^3)$, \ie the complexity matching our qubit result \cite{mitosekAlgebraicInterpretationPauli2026}.
While \flow only works for prime $d$, our algorithms solve the underlying linear algebra problem over any finite field $\Fd$, not necessarily a prime field.
Lastly, we adapt some flow-preserving rules previously known for qubit flows to also work for qudit flow.
In addition to the local complementation and local scaling rules that were already known for the qudit case, this includes pivoting \cite{duncanGraphtheoreticSimplificationQuantum2020,backensThereBackAgain2021,simmonsRelatingMeasurementPatterns2021}, reversibility \cite{mitosekAlgebraicInterpretationPauli2026}, and the removal or insertion of a family of measurements \cite{duncanGraphtheoreticSimplificationQuantum2020,backensThereBackAgain2021, mcelvanneyCompleteFlowPreservingRewrite2023,perezBackens2025}.
We also show how these flow-preserving rules allow the generation of arbitrary \LOFG with \flow, generalising a recent result of one of the authors in the qubit case \cite{backensGenerating2026}.

The generalisation from qubits to qudits comes with several intricacies as the field $\mathbb{F}_2$ (having characteristic 2) behaves quite differently than $\Fd$ for any other prime $d$.
In general, working over qudits means there are many more possible measurement spaces, correction operators, and graph symmetries to be considered and new choices had to be made to define \eg the concept of focusing or a useful pivot operation in the qudit context.

In Section~\ref{Sec:background}, we provide the necessary background.
Next, in Section~\ref{Sec:Focusing flow}, we adapt the notion of `focusing' to the qudit case and show that focused \flow exists whenever any \flow exists.
In Section~\ref{Sec:Algebraic formulation of focused flow}, we give an algebraic formulation of focused \flow.
We use this formulation in Section~\ref{Sec:Flow-finding} to derive a faster flow-finding algorithm.
In Section~\ref{Sec:Flow-preserving rewriting}, we consider \flow-preserving rewriting.
Lastly, in Section~\ref{Sec:Conclusions and future directions}, we give conclusions and discuss potential further work.
Various technical parts of the paper are included in the appendices.

\section{Background}\label{Sec:background}

We work with finite fields $\Fd$ of order $d = p^k$ for prime $p$ and a positive integer $k$.
Unless specified otherwise, we assume throughout that $k = 1$ and thus $d=p$ is prime.
We use $\Fd^*$ for $\Fd \setminus \{ 0 \}$.
We also write $*$ for an arbitrary element of $\Fd^*$, \ie an arbitrary non-zero element of $\Fd$.
All matrices are always over the field $\Fd$.
Let $K$ be a matrix with row labels $\mathcal{R}$ and column labels $\mathcal{C}$.
We refer to it as $\mathcal{R} \times \mathcal{C}$ matrix.
We write $K_{v,*}$ where $v \in \mathcal{R}$ for the $v$-labelled row of $K$ and, similarly, $K_{*,u}$ with $u \in \mathcal{C}$ for the $u$-labelled column of $K$.
Given $\mathcal{R}' \subseteq \mathcal{R}$ and $\mathcal{C}' \subseteq \mathcal{C}$, we write $K_{\mathcal{R}'}^{\mathcal{C}'}$ for the $\mathcal{R}' \times \mathcal{C}'$ submatrix of $K$.

A comparison of these concepts and other notions throughout this paper to what is known for qubits (\ie the case $d=2$) can be found in Appendix~\ref{Sec:QubitComp}.

\begin{definition}\label{Def:FG}
    An \emph{$\Fd$-graph} $G=(V,E)$ is a simple loop-free $\Fd^*$-weighted graph.
    We write $n := |V|$ for the number of vertices of $G$.
    We associate the graph with its \emph{adjacency matrix}: the $V \times V$ matrix $G$, where $G_{u,v}$ stands for the weight of edge $(u,v)$, or $0$ if there is no such edge.
\end{definition}

\begin{definition}\label{Def:OFG}
    An open $\Fd$-graph $(G,I,O)$ is a triple where $G$ is an $\Fd$-graph, and $I, O \subseteq V$ are sets of input and output vertices.
    We write $\comp{I} := V \setminus I$ and $\comp{O} := V \setminus O$ for the sets of non-input and non-output vertices respectively.
    Finally, we write $B := V \setminus (I \cup O)$ for the set of internal vertices (B for in-Between).
\end{definition}

\begin{definition}\label{Def:LOFG}
    A \LOFG $(G,I,O,\ld)$ is a quadruple where $(G,I,O)$ is an open $\Fd$-graph and $\ld \colon \comp{O} \to \FF_d^2 \setminus \{(0, 0)\}$ is a \emph{measurement labelling} such that for all $i \in I \setminus O$: $\ld(i) = (0,*)$.
\end{definition}

\subsection{Graph states and the one-way model of measurement-based quantum computation}

Like graphs in the qubit case, $\FF_d$-graphs can be used to define a class of quantum states on qudits called \emph{graph states} \cite[Section~C]{bahramgiriGraph2006}.
These states are stabiliser states, meaning they are eigenstates with eigenvalue~1 of certain tensor products of the qudit Pauli matrices.

We first introduce some of the formalism of qudit stabiliser quantum mechanics.
Let $\{\ket{k}\}_{k\in\FF_d}$ be the computational basis of a qudit, then the Pauli-$Z$ and $X$ matrices are defined as $Z\ket{k} = \omega^k\ket{k}$ and $X\ket{k} = \ket{k+1}$, where $\omega = e^{\frac{2i\pi}{d}}$ and the addition inside the ket is modulo $d$.
The qudit Hadamard operator performs the discrete Fourier transform of dimension $d$, it is defined as $H\ket{k} = \sum_{j\in\FF_d} \omega^{j k} \ket{j}$.
Unlike for qubits, the qudit Hadamard is not self inverse; instead it satisfies $H^4 = I$.
It is straightforward to show that $Z^d = I = X^d$ and the gates satisfy $Z X = \omega X Z$, $H Z H^{-1} = X^{-1}$, as well as $H X H^{-1} = Z$.
The \emph{Pauli group} on one qudit is generated by $Z$ and $X$, its elements are of the form $\omega^a X^b Z^c$ for some $a,b,c\in\FF_d$.
A tensor product of single-qudit Pauli group elements is called a \emph{Pauli product}.
In the following, for any single-qudit unitary $U$ whose eigenvalues are $\{\omega^k\}_{k\in\FF_d}$, we will use $\ket{k:U}$ to denote the eigenstate of $U$ with eigenvalue $\omega^k$.
In particular, $\ket{k:Z} = \ket{k}$ and $\ket{k:X} = H^{-1} \ket{k}$ for all $k\in\FF_d$.
We will also need an entangling two-qudit gate, the controlled-$Z$ gate defined as $CZ \ket{j k} = \omega^{j k} \ket{j k}$.
As for qubits, $CZ$ gates applied to any (even overlapping) pairs of qudits commute.

In the following, we use subscripts to denote the qudit (or qudits) to which a gate is applied, e.g.\ $CZ_{u,v}$ acts on qudits $u$ and $v$.

\begin{definition}\label{def:graph-state}
	Suppose $G = (V,E)$ is an $\FF_d$-graph, then the associated \emph{graph state} is given by $\ket{G} := \prod_{u,w\in V} CZ_{u,w}^{G_{u,w}} \bigotimes_{v\in V} \ket{0:X}$.
\end{definition}

\begin{lemma}[\cite{bahramgiriGraph2006}]\label{lem:graph-state-stabilisers}
	Suppose $G = (V,E)$ is an $\FF_d$-graph, then the stabiliser group for the graph state $\ket{G}$ is generated by the Pauli products of the form $X_v \prod_{w\in V} Z_w^{G_{v,w}}$ for each $v\in V$.
\end{lemma}

In the one-way model of MBQC, a computation is performed via successive adaptive single-qudit measurements on a resource ($\FF_d$)-graph state.
Each measurement has a distinguished \emph{desired outcome}, which is the one driving the computation in the desired direction.
While individual measurements are non-deterministic, by restricting measurements to specific well-behaved measurement spaces, it is possible to compensate for a specific measurement having resulted in an undesired outcome through changing later measurements.
In this way one can use measurements to perform a computation that is overall deterministic.
The following presentation of the one-way model on qudits follows \cite{boothOutcomeDeterminismMeasurementbased2023}.

A projective measurement on a qudit consisting of measurement projectors $\{\pi_k\}_{k\in\FF_d}$ can be identified with a unitary matrix $U = \sum_{k\in\FF_d} \mu_k \pi_k$ for some set of eigenvalues $\mu_k$.
We will assume that the desired outcome is associated with eigenvalue 1, \ie we consider only unitary representations of measurements that satisfy $U\ket{\psi} = \ket{\psi}$ for some $\ket{\psi}\in\CC^d$ called the \emph{fixpoint}.

Each single-qudit Pauli $X^a Z^b$ for some $a,b\in\FF_d$ such that $a,b$ are not both 0 determines a \emph{measurement space}, which is denoted $\cM(a,b)$ and which groups together measurements with shared characteristics.
The measurement labelling $\ld$ of a \LOFG (Definition~\ref{Def:LOFG}) refers exactly to these measurement spaces.
Then a measurement unitary $U$ is in $\cM(a,b)$ if and only if it satisfies $X^a Z^b U = \omega U X^a Z^b$, implying the following useful properties.

\begin{proposition}[{\cite[Proposition~2]{boothOutcomeDeterminismMeasurementbased2023}}]
	If $U\in\cM(a, b)$, then the spectrum of $U$ is $\{\omega^m \mid m\in\FF_d\}$, each eigenvalue has multiplicity 1, and $U$ is special unitary.
	Denoting $\ket{0:U}$ the fixpoint of $U$, then $\ket{m:U} = (X^a Z^b)^{-m} \ket{0:U}$ is an eigenvector of $U$ associated with eigenvalue $\omega^m$.
\end{proposition}

If the measurement $U$ yields the desired outcome, then the projection applied to the resource state is $\proj{0:U}$.
If instead the measurement yields an undesired outcome, the projection is $\proj{m:U}$ for some $m\in\FF_d^*$.
By the above proposition, the undesired projectors are linked to the desired projector via powers of $X^a Z^b$: $\proj{0:U} = (X^a Z^b)^m \proj{m:U} (X^a Z^b)^{-m}$.

This is the basis for the adaptive procedure that can correct undesired measurement outcomes: Suppose the measurement on qudit $v$ has resulted in the undesired outcome $m\in\FF_d^*$.
If there is a stabiliser of the underlying graph that acts on $v$ as $(X^a Z^b)^m$ and otherwise acts non-trivially only on qudits that have not yet been measured, then this stabiliser can be used to correct the measurement on $v$.
The Pauli operators of the correcting stabiliser that act on qudits other than $v$ itself are called \emph{by-products}.
The assumption that non-trivial by-products arise only on unmeasured qudits implies that they can later be incorporated into the measurement on that qudit.
Generalising from a single measurement to an entire computation, we need a partial order on the measurements that indicates which measurements happen before or after others, as well as a correcting stabiliser for each measured qubit; the partial order and the stabilisers need to be compatible to ensure that by-products affect only qudits that come later in the partial order.
Some care also needs to be taken around treating qudits that act as inputs for the entire computation, meaning they are not prepared in the state $\ket{0:X}$ but arrive in some arbitrary state and are entangled into the graph; as well as qudits that are not measured but remain as outputs at the end.
These conditions were formalised in \cite{boothOutcomeDeterminismMeasurementbased2023}, we will introduce them now.

\subsection{\flowintitles}

We use \cite[Lemma 4.3]{boothMeasurementbasedQuantumComputation2022} as a definition of \flow instead of the original version of the definition in \cite[Definition 9]{boothOutcomeDeterminismMeasurementbased2023}, changing the name to `\flow' with $\FF_d$ instead of $\ZZ_d$ to emphasise that the underlying algebraic structure needs to be a finite field.

\begin{definition}[{\flow \cite[Lemma 4.3]{boothMeasurementbasedQuantumComputation2022}}]\label{Def:full flow}
    A \LOFG $(G,I,O,\ld)$ has \emph{\flow} if there exists a tuple $(C,\prec)$ where $C \in \Fd^{V \times V}$ is called the \emph{correction matrix} and $\prec$ is a strict partial order on $V$, such that:\begin{itemize}
        \item[\conlab{CF}] $\forall u \in \comp{O} . \ld(u) = (C_{u,u}, (GC)_{u,u})$,
        \item[\conlab{ZF}] $C_{u,v} = 0$ whenever $u \in I$ or $v \in O$, and
        \item[\conlab{OF}] when the columns and rows of $G$ and $C$ are ordered according to any totalisation of $\prec$, then $C$ and $GC$ are lower triangular.
    \end{itemize}
\end{definition}

The \flow definition implicitly specifies the Pauli product $\bigotimes_{u\in V} X^{C_{u,v}} Z^{(GC)_{u,v}}$ for correcting the measurement on qudit $v$.
It is straightforward to check that this is indeed a stabiliser of the underlying graph state.
Indeed, the flow correction condition \conref{CF} ensures that the Pauli by-products and the desired correction together form a stabiliser for each measured qudit.
The zeroes condition \conref{ZF} states that inputs cannot correct other vertices and that outputs do not need a correction.
Finally, the order condition \conref{OF} states that it is possible to correct measurements in a causal way.
Out of the three conditions, the third one is clearly the most complicated.
Not only is it the only condition in which the strict partial order $\prec$ is involved, the condition also considerers any totalisation of this strict partial order.

A worked example presenting \flow is given in Figure~\ref{Fig:running example}.
The unweighted graph underlying this example is a modification of the graph in \cite[Example 2.43]{backensThereBackAgain2021}.
Verifying that the relevant matrices are lower triangular with respect to some totalisation is a mildly tedious step that requires reordering rows and columns.

\begin{figure}
    \centering
    \begin{subfigure}[b]{0.23\textwidth}
        \centering
        $\scalediag[1]{ex1}$
        \caption{}
        \label{Fig:LOFG}
    \end{subfigure}
    \begin{subfigure}[b]{0.23\textwidth}
        \vspace{0.5cm}
        \centering
        $\scalemath{0.85}{\begin{pNiceArray}{c:ccc:cc}[first-row, first-col]
            \bm{G} & i & a & b & d & o_1 & o_2 \\
            i & 0 & 2 & 0 & 0 & 0 & 0 \\
            \hdottedline
            a & 2 & 0 & 4 & 1 & 4 & 2 \\
            b & 0 & 4 & 0 & 3 & 3 & 0 \\
            d & 0 & 1 & 3 & 0 & 0 & 4 \\
            \hdottedline
            o_1 & 0 & 4 & 3 & 0 & 0 & 0 \\
            o_2 & 0 & 2 & 0 & 4 & 0 & 0 \\
        \end{pNiceArray}}$
        \caption{}
        \label{Fig:G}
    \end{subfigure}
    \begin{subfigure}[b]{0.23\textwidth}
        \centering
        $\scalemath{0.85}{\begin{pNiceArray}{c:ccc:cc}[first-row, first-col]
            \bm{C} & i & a & b & d & o_1 & o_2 \\
            i & 0 & 0 & 0 & 0 & 0 & 0 \\
            \hdottedline
            a & 3 & 0 & 0 & 0 & 0 & 0 \\
            b & 0 & 0 & 2 & 0 & 0 & 0 \\
            d & 0 & 0 & 0 & 4 & 0 & 0 \\
            \hdottedline
            o_1 & 0 & 2 & 0 & 0 & 0 & 0 \\
            o_2 & 0 & 0 & 1 & 3 & 0 & 0 \\
        \end{pNiceArray}}$
        \caption{}
        \label{Fig:C}
    \end{subfigure}
    \begin{subfigure}[b]{0.23\textwidth}
        \centering
        $\scalemath{0.85}{\begin{pNiceArray}{c:ccc:cc}[first-row, first-col]
            \bm{GC} & i & a & b & d & o_1 & o_2 \\
            i & 1 & 0 & 0 & 0 & 0 & 0 \\
            \hdottedline
            a & 0 & 3 & 0 & 0 & 0 & 0 \\
            b & 2 & 1 & 0 & 2 & 0 & 0 \\
            d & 3 & 0 & 0 & 2 & 0 & 0 \\
            \hdottedline
            o_1 & 2 & 0 & 1 & 0 & 0 & 0 \\
            o_2 & 1 & 0 & 0 & 1 & 0 & 0 \\
        \end{pNiceArray}}$
        \caption{}
        \label{Fig:GC}
    \end{subfigure}
    \caption{Running example of a \LOFG with $d=5$. \ref{Fig:LOFG}: a \LOFG with one input $i$ and two outputs $o_1$ and $o_2$. The measurement labels of non-outputs are given near the names of the vertices. Edge weights are denoted by numbers near the middles of the edges. \ref{Fig:G}: the corresponding adjacency matrix. The dotted lines subdivide input, internal, and output vertices for improved readability but otherwise play no role. \ref{Fig:C}: the corresponding correction matrix. This matrix is lower triangular. \ref{Fig:GC}: the product of adjacency and correction matrices from \ref{Fig:G} and \ref{Fig:C} respectively. It is not lower triangular, however it becomes lower triangular when the vertices are ordered $i \prec a \prec d \prec b \prec o_1 \prec o_2$. Under this ordering, the correction matrix from \ref{Fig:C} remains lower triangular, hence the \LOFG has \flow.}
    \label{Fig:running example}
\end{figure}

\section{Focusing \flowintitles}\label{Sec:Focusing flow}

If a \LOFG has \flow, this \flow is not generally unique.
For qubits, the notion of \emph{focused gflow} (or its generalisation, focused Pauli flow) has been very useful: for example, focused flow is unique if there are equal numbers of inputs and outputs \cite{mhallaWhichGraphStates2014a,backensThereBackAgain2021,simmonsRelatingMeasurementPatterns2021}, it can be easier to work with for flow-preserving rewriting \cite{duncanGraphtheoreticSimplificationQuantum2020,backensThereBackAgain2021,perezBackens2025}, and it enables a linear-algebraic formulation of the flow conditions, leading to more efficient flow-finding algorithms \cite{mhallaWhichGraphStates2014a,mitosekAlgebraicInterpretationPauli2026}.
In this context, `focusing' means restricting the correction by-products that a qubit can receive, depending on its measurement label.
We now generalise this idea to qudit \flow. We use notation $\left(\ldX(v), \ldZ(v)\right) := \ld(v)$.

\begin{definition}\label{Def:focused}
    An \flow $(C,\prec)$ on a \LOFG $(G,I,O,\ld)$ is \emph{focused} when:\begin{itemize}
        \item[\conlab{FX}] for all $v \in \comp{O}$ such that $\ldX(v) = 0$, we have $(GC)_{v,u} = 0$ for any $u \in \comp{O}\setminus\{v\}$ and
        \item[\conlab{FZ}] for all $v \in \comp{O}$ such that $\ldX(v) \ne 0$, we have $C_{v,u} = 0$ for any $u \in \comp{O}\setminus\{v\}$.
    \end{itemize}
\end{definition}

For example, the \flow from Figure~\ref{Fig:running example} is focused.

The first condition~\conref{FX} ($F$ for focused) ensures that $(0,*)$-measured vertices receive only by-products of the form $X^a$ (with trivial $Z$ component).
All other vertices receive only by-products of the form $Z^b$ (with trivial $X$-component), as captured by condition~\conref{FZ}.
This corresponds to focusing conditions for the qubit case: $XY$-measured vertices (now generalised to $(0,*)$) can only get an $X$ by-product, while $YZ$ and $XZ$ (now generalised to $(*,0)$ and $(*,*)$ respectively) can only get a $Z$ by-product.

We show, that existence of \flow is equivalent to existence of focused \flow.
The proof may be found in Appendix~\ref{Sec:ProvingFocusing}.

\begin{theorem}\label{Th:flow focusing}
    Let $(G,I,O,\ld)$ be a \LOFG which has \flow.
    Then there exists a focused \flow on $(G,I,O,\ld)$.
\end{theorem}

Later on, in Corollary \ref{Cor:uniqueness}, we will explain that focused \flow can be viewed as canonical, in the sense that when $|I|=|O|$, then there exists a unique focused \flow, up to extensions of the partial order.

\section{Algebraic formulation of focused \flowintitles}\label{Sec:Algebraic formulation of focused flow}

The definition of \flow gives linear algebraic conditions that a correction procedure must satisfy to enable deterministic MBQC.
Yet, condition \conref{O} is awkward as it asks for a specific property to hold for any totalisation of the partial order.
We now refine the algebraic formulation, tailoring it specifically towards focused \flow.
This loses some generality -- the new definition can no longer express non-focused \flow{} -- but we already proved that any \LOFG with \flow also has focused \flow.
The small restriction is therefore outweighed in many cases by the associated gains in simplicity and efficiency of avoiding the arbitrary totalisations of the partial order.

First, we reduce some redundancy.
In Definition~\ref{Def:full flow}, the correction matrix $C$ is a $V \times V$ matrix, but the condition~\conref{ZF} forces $I$-labelled rows and $O$-labelled columns of $C$ to be identically zero.
Thus, instead of considering the full $V \times V$ matrix $C$, it suffices to instead consider its $\comp{I} \times \comp{O}$ submatrix.
Correspondingly, it suffices to consider a submatrix of the adjacency matrix; because we want to multiply the two submatrices, the shape of this second submatrix will be $\comp{O}\times\comp{I}$:

\begin{definition}\label{Def:ram}
    Given an open $\Fd$ graph $(G,I,O)$, the matrix $A_{(G,I,O)} := G_{\comp{O}}^{\comp{I}}$ is its \emph{reduced adjacency matrix}.
\end{definition}

The reduced adjacency matrix allows the following restatement of the \flow definition, with a proof of equivalence of two definitions given in Appendix~\ref{subsec:prelims}.

\begin{definition}[\flow]\label{Def:flow}
    A \LOFG $(G,I,O,\ld)$ has an \flow $(C, \prec)$, where $C \in \Fd^{\comp{I} \times \comp{O}}$ is called a \emph{reduced correction matrix} and $\prec$ is a strict partial order on $\comp{O}$ such that:\begin{itemize}
        \item[\conlab{C}] $\forall u \in \comp{O} \setminus I . \ld(u) = (C_{u,u}, (AC)_{u,u})$ and $\forall u \in I \cap \comp{O} . \ldZ(u) = (AC)_{u,u}$
        \item[\conlab{O}] when the columns and rows of $A$ and $C$ are ordered according to any totalisation of $\prec$, then $C$ and $AC$ are lower triangular.
    \end{itemize}
\end{definition}

The reduced adjacency matrix and reduced correction matrix of the running example from Figure~\ref{Fig:running example} are given in Figures~\ref{Fig:A} and~\ref{Fig:Cred}, respectively.

For the new formulation, instead of working directly with the (reduced) adjacency matrix of the \LOFG, we will define two matrices (both of shape $\comp{O}\times\comp{I}$) that encode the measurement labels as well as relevant properties of the adjacency matrix.
First, we introduce some relevant notation.
We write $\delta_{a,b}$ for the Kronecker delta, i.e.\ $\delta_{a,b} = 1$ when $a=b$ and $0$ otherwise.
We write $\id{\mathcal{R}}{\mathcal{C}}$ for a $\mathcal{R} \times \mathcal{C}$ matrix where $(\id{\mathcal{R}}{\mathcal{C}})_{u,v} = \delta_{u,v}$ for all $u \in \mathcal{R}$ and $v \in \mathcal{C}$.
When $\mathcal{R} = \mathcal{C}$, $\id{\mathcal{R}}{\mathcal{C}}$ is the identity matrix.
When $|\mathcal{R}|=1$ or $|\mathcal{C}|=1$, then $\id{\mathcal{R}}{\mathcal{C}}$ is a row or column vector, respectively, which has zeroes everywhere except for the at most one entry in $\mathcal{R} \cap \mathcal{C}$.
We use $\supp{\somevector}$ to denote the support of the vector $\somevector$, i.e.\ if $\mathcal{R}$ is the set of labels of the elements of $\somevector$, then $\supp{\somevector} = \{r\in\mathcal{R} \mid \somevector_r \neq 0\}$.

\begin{definition}[Flow-demand matrix]\label{Def:fdm}
    Given a \LOFG $(G,I,O,\ld)$, the corresponding \emph{flow-demand matrix} $M$ is a $\comp{O} \times \comp{I}$ matrix where for $v \in \comp{O}, u \in \comp{I}$:\begin{align*}
        M_{v,u} = \begin{cases}
            \ldZ(v)^{-1} A_{v,u} & \text{if}\ \ldX(v) = 0\\
            \ldX(v)^{-1} \delta_{v,u} & \text{if}\ \ldX(v) \ne 0
        \end{cases}
    \end{align*}
\end{definition}

The expressions in the definition above are valid, in the sense that the inverse of $\ldZ(v)$ is used only when $\ldX(v) = 0$, in which case $\ldZ(v) \ne 0$ as $\left(\ldX(v), \ldZ(v)\right) = \ld(v) \ne (0,0)$ by Definition~\ref{Def:LOFG}.
Moreover, the required inverses exists due to $\Fd$ being a field.
Similarly, the inverse of $\ldX(v)$ is only used when $\ldX(v) \ne 0$.
Furthermore $A_{v,u}$ exists in the reduced adjacency matrix as $v \in \comp{O}$ and $u \in \comp{I}$.
See Figure~\ref{Fig:M} for an example of the flow-demand matrix and Appendix~\ref{Sec:Pseudocode} for pseudocode that constructs this matrix from a specification of the \LOFG.

\begin{definition}[Order-demand matrix]\label{Def:odm}
    Given a \LOFG $(G,I,O,\ld)$, the corresponding \emph{order-demand matrix} $N$ is an $\comp{O} \times \comp{I}$ matrix where for $v \in \comp{O}, u \in \comp{I}$: $N_{v,u} = \ldX(v) A_{v,u} - \ldZ(v) \delta_{v,u}$.
\end{definition}

See Figure~\ref{Fig:N} for the order-demand matrix of the running example, and Appendix~\ref{Sec:Pseudocode} for pseudocode that generates the order-demand matrix from the specification of a \LOFG.

We can now state the main result of this section:

\begin{theorem}\label{Th:main result}
    Suppose $\Gamma = (G,I,O,\ld)$ is a \LOFG; let $M$ be its flow-demand matrix and $N$ its order-demand matrix.
    Then $\Gamma$ has \flow if and only if there exists a matrix $C$ (of shape $\comp{I} \times \comp{O}$) such that $MC = \id{\comp{O}}{\comp{O}}$ and $NC$ forms a (weighted) directed acyclic graph.
\end{theorem}

Here, the focusing conditions~\ref{FZ} and~\ref{FX} are fully captured by the requirement that $MC = \id{\comp{O}}{\comp{O}}$.
The condition \conref{C} is mostly, but not fully, captured by the requirement $MC = \id{\comp{O}}{\comp{O}}$.
Part of \conref{C} is instead captured by the property that the main diagonal of the product $NC$ is identically $0$, which follows from $NC$ being a DAG.
The condition \conref{O} is fully captured by $NC$ being a DAG.
The proof of Theorem~\ref{Th:main result} is given in Appendix~\ref{Sec:ProvingAlgebraicFormulation}.
See Figure~\ref{Fig:running example algebraic} for the application of this result to the running example from Figure~\ref{Fig:running example}.

When $|I|=|O|$, the flow-demand matrix is square and has at most one right inverse: the two-sided inverse. Therefore, there exists at most one focused \flow  (up to refinements of the partial order). Thus the focused flow uniqueness result known for qubits \cite{mhallaWhichGraphStates2014a,backensThereBackAgain2021,simmonsRelatingMeasurementPatterns2021} generalises:

\begin{corollary}\label{Cor:uniqueness}
    Suppose $\Gamma=(G,I,O,\ld)$ is a \LOFG where $|I|=|O|$, then all focused {\flow\unskip}s on $\Gamma$ share the same reduced correction matrix.
\end{corollary}

\begin{figure}
    \centering
    \begin{subfigure}[b]{0.18\textwidth}
        \vspace{0.5cm}
        \centering
        $\scalemath{0.85}{\begin{pNiceArray}{ccc:cc}[first-row, first-col]
            \bm{A} & a & b & d & o_1 & o_2 \\
            i & 2 & 0 & 0 & 0 & 0 \\
            \hdottedline
            a & 0 & 4 & 1 & 4 & 2 \\
            b & 4 & 0 & 3 & 3 & 0 \\
            d & 1 & 3 & 0 & 0 & 4 \\
        \end{pNiceArray}}$
        \caption{}
        \label{Fig:A}
    \end{subfigure}
    \begin{subfigure}[b]{0.14\textwidth}
        \centering
        $\scalemath{0.85}{\begin{pNiceArray}{c:ccc}[first-row, first-col]
            \bm{C} & i & a & b & d \\
            a & 3 & 0 & 0 & 0 \\
            b & 0 & 0 & 2 & 0 \\
            d & 0 & 0 & 0 & 4 \\
            \hdottedline
            o_1 & 0 & 2 & 0 & 0 \\
            o_2 & 0 & 0 & 1 & 3 \\
        \end{pNiceArray}}$
        \caption{}
        \label{Fig:Cred}
    \end{subfigure}
    \begin{subfigure}[b]{0.18\textwidth}
        \centering
        $\scalemath{0.85}{\begin{pNiceArray}{ccc:cc}[first-row, first-col]
            \bm{M} & a & b & d & o_1 & o_2 \\
            i & 2 & 0 & 0 & 0 & 0 \\
            \hdottedline
            a & 0 & 3 & 2 & 3 & 4 \\
            b & 0 & 3 & 0 & 0 & 0 \\
            d & 0 & 0 & 4 & 0 & 0 \\
        \end{pNiceArray}}$
        \caption{}
        \label{Fig:M}
    \end{subfigure}
    \begin{subfigure}[b]{0.18\textwidth}
        \vspace{0.5cm}
        \centering
        $\scalemath{0.85}{\begin{pNiceArray}{ccc:cc}[first-row, first-col]
            \bm{N} & a & b & d & o_1 & o_2 \\
            i & 0 & 0 & 0 & 0 & 0 \\
            \hdottedline
            a & 2 & 0 & 0 & 0 & 0 \\
            b & 3 & 0 & 1 & 1 & 0 \\
            d & 4 & 2 & 3 & 0 & 1 \\
        \end{pNiceArray}}$
        \caption{}
        \label{Fig:N}
    \end{subfigure}
    \begin{subfigure}[b]{0.15\textwidth}
        \centering
        $\scalemath{0.85}{\begin{pNiceArray}{c:ccc}[first-row, first-col]
            \bm{NC} & i & a & b & d \\
            i & 0 & 0 & 0 & 0 \\
            \hdottedline
            a & 1 & 0 & 0 & 0 \\
            b & 4 & 2 & 0 & 4 \\
            d & 2 & 0 & 0 & 0 \\
        \end{pNiceArray}}$
        \caption{}
        \label{Fig:NC}
    \end{subfigure}
    \begin{subfigure}[b]{0.13\textwidth}
        \centering
        $\scalediag[1]{ex1order}$
        \caption{}
        \label{Fig:order}
    \end{subfigure}
    \caption{The algebraic formulation of the \LOFG from Figure~\ref{Fig:running example}. \ref{Fig:A}: the reduced adjacency matrix which is obtained from the matrix in \ref{Fig:G} by removing output rows and input columns; the reduced correction matrix in \ref{Fig:Cred} is obtained from the matrix in \ref{Fig:C} by a transpose of this process. \ref{Fig:M}: the flow-demand matrix, which has the reduced correction matrix in \ref{Fig:Cred} as its right inverse. \ref{Fig:N}: the order-demand matrix. \ref{Fig:NC}: the product of the order-demand matrix with the reduced correction matrix. This matrix corresponds to the relation $\trl_C$ in \ref{Fig:order}, which is a DAG. Hence, the reduced correction matrix results in focused flow.}
    \label{Fig:running example algebraic}
\end{figure}

\subsection{Focused vectors}

In the qubit case, the kernel of the flow-demand matrix encodes a structure called \emph{focused sets} \cite[Section~3.4]{mitosekAlgebraicInterpretationPauli2026}.
These sets, originally defined by Simmons \cite[Definition~4.3]{simmonsRelatingMeasurementPatterns2021}, correspond to `do nothing' stabilisers, \ie stabilisers that do not help with correcting measurement by-products in any way.
We extend the notion of focused sets to higher dimensions and recover the relationship to the kernel of the flow-demand matrix.
As a corollary, since the kernel of a matrix forms a vector space and thus an Abelian group, we also recover the Simmons qubit result \cite[Lemma B.8]{simmonsRelatingMeasurementPatterns2021}.
Relevant proofs are included in Appendix~\ref{subsec:ProvingFocusedVectors}.

\begin{definition}[Focused vector]\label{Def:foc vec}
    Given a \LOFG $(G,I,O,\ld)$, a column vector $\somevector \in \FF_d^{\comp{I}}$ is \emph{focused over} $S \subseteq{\comp{O}}$ if:
    \begin{enumerate}
        \item[\conlab{Fs1}] $\forall w \in S \cap \supp{\somevector} . \ldX(w) = 0$ and
        \item[\conlab{Fs2}] $\forall w \in S \cap \supp{A\somevector} . \ldX(w) \ne 0$.
    \end{enumerate}
    A vector focused over $\comp{O}$ is simply called \emph{focused}. The set of all focused vectors for $\Gamma$ is denoted $\mathfrak{F}_{\Gamma}$. 
\end{definition}

Focused vectors could also be treated as multisets with support contained in $\comp{I}$ and multiplicities given by elements of $\FF_d$.
We usually treat them as vectors, as that is most convenient for the proofs, but all interpretations are equivalent.

\begin{theorem}\label{Th:focused sets are kernel of M}
    Let $\Gamma = (G,I,O,\ld)$ be a \LOFG and let $M$ be its flow-demand matrix.
    Then $\mathfrak{F}_{\Gamma} \cong \ker M_{\Gamma}$.
\end{theorem}

For example, the kernel of the flow-demand matrix $M$ from Figure~\ref{Fig:M} is $\ker M = \vspan\left((0,0,0,1,3)^T\right)$ and thus the only focused vectors of the \LOFG in Figure~\ref{Fig:LOFG} are those of form $k(0,0,0,1,3)^T$ for some $k \in \Fd$.
Under the multiset interpretation, these consist of vertex $o_1$ with multiplicity $k$ and vertex $o_2$ with multiplicity $3k$.

\section{Flow-finding}\label{Sec:Flow-finding}

With the new formulation of focused flow, we now generalise the flow-finding algorithm from \cite{mitosekAlgebraicInterpretationPauli2026} to qudits.
In our previous work \cite[Corollary 4.1]{mitosekAlgebraicInterpretationPauli2026}, we presented algorithms for qubit flow which effectively solve \flow-finding when $d=2$, in the process solving the following linear algebra problem over the field $\FF_2$ only\footnote{Our previous work in \cite{mitosekAlgebraicInterpretationPauli2026} addresses the more general setting of Pauli flow, in which measurements in the eigenbasis of a Pauli operator can be accorded special treatment. However when Pauli measurements are not treated separately, Pauli flow collapses to gflow, the latter being exactly the \flow where $d=2$.}:

\problemstatement{\laproblem}
{Two matrices $M$ and $N$ of matching shape $\mathcal{A} \times \mathcal{B}$ over $\Fd$.}
{A $\mathcal{B} \times \mathcal{A}$ matrix $C$ over $\Fd$ such that $MC = \id{\mathcal{A}}{\mathcal{A}}$ and $NC$ forms a directed acyclic $\Fd$-weighted graph, or the message that no such $C$ exists.}

The structure of the qubit algorithms can be understood as follows: given a labelled open graph, construct the flow-demand and order-demand matrices.
Next, solve $\laproblem$.
The qubit flow-finding algorithms do not assume anything about the structure of $M$ and $N$ after their initial construction.
These algorithms can be adapted to solve $\laproblem$, and thus \flow-finding, for any $d$ such that $\FF_d$ forms a fields.
The parameter $d$ does not have to be prime, but only a power of a prime.
We obtain:

\begin{theorem}\label{Th:On3algo}
    There exists a $\bigO(n^3)$ algorithm that, given a \LOFG, outputs the corresponding \flow, or gives a message that no \flow exists.
\end{theorem}

For explanation of how to incorporate the computational complexity of finite field arithmetic into the complexity analysis, see Appendix~\ref{Sec:FFA}.

In the qubit case, we split flow-finding into two cases depending on whether $|I| = |O|$, \ie on whether the input to $\laproblem$ is given by square matrices.
When $|I| = |O|$, the algorithm for the qubit case applies to qudits without any changes.
However, in the general case, the qubit algorithm must be adjusted.
We include the pseudocode for \flow-finding in the general case in Appendix~\ref{Sec:Pseudocode} together with an explanation of where it differs from the qubit case.
Rather than writing out what would be near identical proofs to those in \cite{mitosekAlgebraicInterpretationPauli2026}, we provide a working example of finding flow for the \LOFG from Figure~\ref{Fig:running example}, see Appendix~\ref{Sec:Flow run}.

Speed-ups to matrix multiplication and inversion can be deployed when $|I|=|O|$, reducing the complexity from $\bigO(n^3)$ to a possibly smaller exponent, see \cite[Section~4.3.1]{mitosekAlgebraicInterpretationPauli2026} for details.
We also obtain a lower bound by linking \flow-finding to the problems of matrix multiplication and matrix inversion over a finite field, see Appendix~\ref{Sec:LowerBound}.

\section{\flowintitles--preserving rewriting}\label{Sec:Flow-preserving rewriting}

So far, we have worked with a fixed \LOFG in defining and finding \flow, and relied only on the data contained in this \LOFG.
To implement a one-way computation, more information is needed, in particular it is necessary to specify a measurement operator (with its desired outcome) for each measured qudit.
Given such an expression for an entire computation in the one-way model, one immediate question is about alternative ways of implementing the same linear operation, which may be more efficient, better suited to a specific hardware connectivity, or have other desired properties.
Thus, in the qubit case, there is a thriving system of `flow-preserving rewrite rules' \cite{duncanGraphtheoreticSimplificationQuantum2020,backensThereBackAgain2021,simmonsRelatingMeasurementPatterns2021,mcelvanneyCompleteFlowPreservingRewrite2023,mcelvanneyFlowpreservingZXcalculusRewrite2023a,perezBackens2025} which map some given one-way computation to a different computation implementing the same linear operation.
These rules can be used to optimise quantum computational resources in the circuit or one-way models according to different metrics \cite{staudacherReducing2QuBitGate2023,holkerCausal2024}, or even to obfuscate quantum computations so they can be delegated into the cloud while preserving security properties \cite{caoMulti-agent2023}.

We will therefore now consider flow-preserving rewriting for \flow.

\subsection{Local scaling, local complementation, and pivot}
\label{s:scaling-LC-pivot}

Some relevant rewrite rules were already defined and proved flow-preserving by Booth \cite[Section~5.1]{boothMeasurementbasedQuantumComputation2022}.
These rules are based on the two different types of operations on an $\FF_d$-weighted graph which translate to local Clifford equivalences of the corresponding qudit graph states \cite[Section~D]{bahramgiriGraph2006}.
The first of these operations is \emph{local scaling}, which multiplies the weights of all edges incident on a chosen vertex $w$ by some non-zero constant $\gamma\in\FF_d^*$ and is usually denoted $G\overset{\gamma}{\circ} w$.
The second operation is \emph{local complementation}, which is also parametrised by a constant $\gamma\in\FF_d$ and changes the weights of all edges among the neighbours of the chosen vertex $w$ in a more complicated way, the resulting graph is denoted $G\overset{\gamma}{\star} w$.
In both cases, the graph state $\ket{G}$ corresponding to the initial graph can be transformed into the graph state corresponding to the final graph ($|G\overset{\gamma}{\circ} w\rangle$ or $|G\overset{\gamma}{\star} w\rangle$, respectively) by applying a product of single-qudit Clifford operations.
As a result, both operations preserve the existence of \flow, given certain updates to the \LOFG \cite[Propositions~5.5 and 5.8]{boothMeasurementbasedQuantumComputation2022}; formal statements of this are given in Appendix~\ref{s:graph-operations}.

A \emph{pivot}, a sequence of three local complementations about the endpoints of an edge, gives rise to a flow-preserving transformation on qubits and it interacts more smoothly with focused gflow and focused Pauli flow than individual local complementations \cite[Lemma~5.3]{perezBackens2025}.
The pivot has been generalised to a family of pivot operations on qutrits \cite[Theorem~3.5]{townsend-teagueSimplification2022} and a qudit pivot-and-delete operation (where the endpoints of the edge are deleted after the pivot) \cite[Lemma~10]{poorQupitStabiliserZXtravaganza2023}.
We work out the effects of a pivot operation (without deletion) for general prime-dimensional qudits in Appendix~\ref{s:graph-operations} and summarise them here.

\begin{proposition}[informal statement]\label{prop:pivot-informal}
	Let $G$ be an $\FF_d$-weighted graph, suppose $w, w'\in V$ satisfy $\epsilon := G_{w,w'} \ne 0$, and define $G\wedge w w'$ to be the graph after a pivot on the edge $w w'$.
	Let $F^{(\gamma)} := \sum_{j,k\in\FF_d} \omega^{\gamma j k } \ketbra{j}{k}$ for any $\gamma\in\FF_d^*$, then $\ket{G\wedge w w'}
	= F^{(-\epsilon)}_w F^{(-\epsilon)}_{w'} \prod_{v\in V} \phasegate_v^{2\epsilon^{-1} G_{v,w}G_{v,w'}} \ket{G}$.
\end{proposition}

We also argue that the above pivot operation can be implemented using three successive local complementations and at least one local scaling (the latter is needed to get a pivot operation that is symmetric under interchange of $w$ and $w'$).
This decomposition immediately implies that a pivot preserves the existence of \flow.

\subsection{Flow reversal}

In the qubit case, for a \LOG with $|I|=|O|$, it is possible to reverse the direction of flow \cite{mhallaWhichGraphStates2014a,mitosekAlgebraicInterpretationPauli2026} by swapping inputs with outputs (and assigning $XY$ measurement labels to the former outputs, \ie the new inputs).
Moreover, there is a close relationship between the flow in one direction and the flow in the other direction: for the fragment consisting only of $XY$ measurements, inputs and outputs, the correction matrix in the reverse direction is just the transpose of the original correction matrix, although this relationship is more complicated for the other measurement planes.
We lift this result to qudits.

\begin{definition}\label{def:reverse-LOFG}
    Let $\Gamma = (G,I,O,\ld)$ be a \LOFG which satisfies $|I|=|O|$.
    Then its \emph{reverse \LOFG} is $\Gamma' = (G,O,I,\ld')$, where
	$\ld'(v) = (0,1)$ if $v\in O$ and $\ld'(v)=\ld(v)$ if $v\in B$.
\end{definition}

\begin{theorem}[\flow reversibility]\label{Th:flow-rev}
    Let $\Gamma$ be a \LOFG which satisfies $|I|=|O|$.
    Let $\Gamma'$ be the reverse of $\Gamma$.
    Then $\Gamma$ has \flow if and only if $\Gamma'$ has \flow.
\end{theorem}

The proof is in Appendix~\ref{Sec:proof reversibility}; it is quite close to the qubit case, except that working over an arbitrary prime field requires extra care with subtraction (which over $\FF_2$ is equivalent to addition).

\subsection{\texorpdfstring{$Z$}{Z}-like removal and insertion}\label{Subsec:Z-like}

In a focused flow, certain kinds of measured vertices receive only by-products that are powers of $Z$, these vertices are called `$Z$-like'.
$Z$-like measured qubits can be removed from one-way computations \cite{duncanGraphtheoreticSimplificationQuantum2020,backensThereBackAgain2021,simmonsRelatingMeasurementPatterns2021} and, under certain conditions, inserted into one-way computations \cite{perezBackens2025} while preserving the existence of flow.
We lift these flow-preserving qubit removal and insertion results to qudits; the corresponding proofs are presented in Appendix~\ref{Sec:proving removal and insertion}.

For any $\Fd$-graph $G$ and any vertex $v$ in $G$, we denote by $G-v$ the graph that results from removing $v$ (and all its incident edges).

\begin{definition}
    Let $(G,I,O)$ be an \OFG and let $S \in \Fd^V$ be any vector.
    Define the \emph{\OFG extended by $v$} as the \OFG $(G',I,O)$ obtained by inserting a new vertex $v$ into $(G,I,O)$ with its connections determined by $S$, \ie $G' - v = G$ and $G'_{v,u} = S_u$ for all $u \in V$.
    Given $(G,I,O,\ld)$ and $\ld'(v)$, we define the \emph{\LOFG extended by $v$ with measurement label $\ld'(v)$} as the \LOFG $(G',I,O,\ld')$, where $(G',I,O)$ is the \OFG extended by $v$, and $\ld'(u) = \ld(u)$ for all $u \in \comp{O}$.
\end{definition}

From now on, we add a prime symbol for all notation referring to the (labelled) \OFG extended by $v$; this includes $G'$, $\ld'$, $M'$, $N'$, $\comp{O}'$, and $\comp{I}'$.

\begin{theorem}[$Z$-like deletion]\label{Th:Z-like removal}
    Let $\Gamma' := (G',I,O,\ld')$ be the \LOFG $(G,I,O,\ld)$ extended by $v$, where $\ldX'(v) \ne 0$, and suppose $\Gamma'$ has \flow.
    Then $(G,I,O,\ld)$ also has \flow.
\end{theorem}

For $Z$-like insertion, we restrict to the case $|I|=|O|$ as it is exactly the uniqueness of focused \flow in the case $|I|=|O|$ (\cf Corollary~\ref{Cor:uniqueness}) that makes the insertion of $Z$-like vertices and the corresponding updates to the \flow possible while performing fewer\footnote{In fact: strictly fewer, unless the lower bound from Theorem~\ref{Th:lower bound} is also an upper bound.} operations than would be required by the brute-force approach of inserting the vertex and running the full flow-finding procedure \cite[page 109]{perezBackens2025}.

\begin{theorem}[$Z$-like insertion]\label{Th:Z-like insertion}
    Let $\Gamma = (G,I,O,\ld)$ be a \LOFG with focused \flow, where $|I|=|O|$.
    Further suppose that $C, NC$ are known, where $C$ is the reduced correction matrix guaranteed to exist by Theorem~\ref{Th:main result} and to be unique by Corollary~\ref{Cor:uniqueness}; and $N$ is the order-demand matrix of $\Gamma$.
    Let $\Gamma' = (G',I,O,\ld')$ be the \OFG extended by $v$ via a known vector $S \in \Fd^{V}$.
    Then, in $\bigO(n^2)$ time complexity, it is possible to determine the set of all $\ld'(v)$ such that $\ldX'(v) \ne 0$ and $(G',I,O,\ld')$ has \flow.
\end{theorem}

\subsection{Generation of labelled open graphs with \flowintitles}

Benchmarking flow-based algorithms on qubits, such as finding and circuit extraction, already requires generating large labelled open graphs with flow, and such \LOG{}s have also found application in the generation of ansätze for optimisation and machine learning \cite{ewenApplication2025}.
For these applications, the flows should be \emph{non-trivial}, \ie the correction mechanism should not be obvious from the structure of the graphs alone.
Since most \LOG{}s do not have flow\footnote{Informally, most flow-demand matrices will not have maximal row-rank and thus have no right inverse. In particular, a certain submatrix of the adjacency matrix needs to be invertible, but most graphs do not have this property \cite{scirihaCharacterizationSingularGraphs2007}.}, it can be difficult to generate instances that have flow in a non-trivial way.
For example, one way to generate large instances with flow on qubits is to translate quantum circuits into measurement patterns.
However, this method results in diagrams with causal flow, where the correction mechanism is particularly simple: each row of the correction-matrix contains only a single non-zero value.
Another approach for generating {\LOG\unskip}s with flow is to compose smaller instances with flow.
However, this leads to a layered structure with limited connectivity.
Based on a recent result of one of the authors for qubits \cite{backensGenerating2026}, we now give an algorithmic method for generating instances with \flow that can overcome these shortcomings: the algorithm has non-zero probability of returning \emph{any} \LOFG with \flow, not only \LOFG{}s of the restricted types outlined above.

\begin{theorem}\label{Th:generation}
    There exists an algorithm (presented in Appendix~\ref{Sec:proving generation}) that takes a pair of integers $n$ and $k$ as input and it returns a random \LOFG with $|V|=n$ and $|I|=|O|=k$ that has \flow.
    Furthermore, for any \LOFG $\Gamma$ with $|V|=n$ and $|I|=|O|=k$ with \flow, the algorithm has non-zero probability of returning $\Gamma$.
\end{theorem}

The proof is in Appendix~\ref{Sec:proving generation}.
The method we use may not be very efficient, as it is based on attempts of $Z$-like insertions, which can show that insertion of a vertex with a fixed neighbourhood is simply impossible in an \flow-preserving manner.
However, the method is guaranteed to return cases with \flow.
Thus, a one-time generation could be run for an extended period to produce instances that can then be used for benchmarking flow algorithm implementations.
This method is more elegant than the trivial approach of generating large labelled open graphs and running the flow-finding algorithm.
We also expect it to be more practical for generating cases with \flow, though we leave the last fact without a formal proof.

\section{Conclusions and future directions}\label{Sec:Conclusions and future directions}

Building on the work of \cite{boothOutcomeDeterminismMeasurementbased2023} about determinism in qudit MBQC, we developed a definition of focused \flow.
We showed that the existence of \flow is equivalent to the existence of focused \flow; the uniqueness of focused \flow (at least when the numbers of inputs and outputs are equal) thus means focused \flow can be considered canonical.
We refined the algebraic formulation to allow a restatement of the \flow definition using matrix multiplication and a check whether a certain directed graph is acyclic, thus avoiding the previous complicated verification of conditions related to the order of measurements.
Utilising this new algebraic formulation, we obtained an $\bigO(n^3)$ \flow-finding algorithm, the same complexity as for qubits.
Furthermore, we laid the groundwork for applications of \flow in optimisation by deriving several new \flow-preserving rules.

The most natural direction for future work is to apply our results in optimisation.
All properties of flow and flow-preserving transformations used in \cite{duncanGraphtheoreticSimplificationQuantum2020} have now been generalised to higher dimensions suggesting that circuit optimisation procedure over qubits has a great potential to lift to qudits.
Furthermore, our results also generalise flow results essential in refinements of this optimisation method, for instance in \cite{kissingerReducingTcountZXcalculus2020}, albeit some results that do not concern flow remain to be generalised.

Another direction for the future work would be to develop a `causal \flow' analogous to \cite{danosDeterminismOnewayModel2006}.
While causal flow is not necessary for determinism in the qubit case, it corresponds more closely to the circuit model than gflow, which can be advantageous for certain applications \cite{holkerCausal2024}.
Moreover causal flow can be found in sub-cubic time \cite{debeaudrapCompleteAlgorithmFind2007,deBeaudrapFinding2008}.
Additionally, an extended version of causal flow plays a core role in the completeness result for flow-preserving rewriting on qubits \cite{backensCompleteness2026}, which could potentially be extended to qudits in the future.

In the qubit case, separate treatment of Pauli measurement leads to a different flow \cite{browneGeneralizedFlowDeterminism2007} with its own applications \cite{simmonsRelatingMeasurementPatterns2021} and algebraic presentation \cite{mitosekPauliFlowOpen2024,mitosekAlgebraicInterpretationPauli2026}; all
while offering a broader picture of determinism than gflow \cite{browneGeneralizedFlowDeterminism2007,mhallaShadowPauliFlow2025}.
It would be interesting to investigate the existence of a Pauli flow for qudits.

Lastly, our techniques exploit the field structure of $\Fd$ for an improved algebraic presentation.
In the present work on qudits and its predecessor \cite{boothOutcomeDeterminismMeasurementbased2023}, these fields are necessarily finite.
There also exists a definition of MBQC and an associated flow on quantum systems described by continuous variables \cite{boothCVflow2023}, which effectively works over the infinite field $\mathbb{R}$.
We expect our algebraic approach could be adjusted and applied to this continuous variable flow, however the infinite field might lead to difficulties on the algorithmic side of this work.
Care would be needed to choose a suitable representation of real numbers that permits efficient computation of relevant field operations.

\section*{Acknowledgements}
We thank Robert Booth for helpful discussions about \flow. PM is supported by Niedersächsisches Ministerium für Wissenschaft und Kultur. MB is supported by the Plan France 2030 through the PEPR integrated project EPiQ ANR-22-PETQ-0007 and the HQI platform ANR-22-PNCQ-0002.

\phantomsection

\addcontentsline{toc}{section}{References}

\begin{sloppypar}
\bibliographystyle{eptcs}
\bibliography{libref}
\end{sloppypar}

\appendix

\section{Comparison to the qubit case}\label{Sec:QubitComp}

Qubit MBQC and its flow properties are more widely studied and thus many readers might be familiar with qubit flow, but not qudit flow.
For such readers, in this appendix, we explain how various definitions and results correspond to the qubit-case literature, \ie to the case of $d=2$.
In particular, when $d=2$, all matrices become binary matrices.
From now on, we only focus on differences regarding notation, terminology, or provision of literature reference when $d=2$.
We skip over statements where $d=2$ implies all changes on its own.

\paragraph{Background}
\begin{itemize}
	\item The $\Fd$-graph from Definition~\ref{Def:FG} becomes a simple graph. In particular, it is no longer weighted.
	\item The codomain of the measurement labelling associated with a \LOFG (Definition~\ref{Def:LOFG}) codomain would become $\{ (0,1), (1,0), (1,1) \}$.
	However, all three measurement spaces are known in the literature under different names:\begin{itemize}
		\item The $(0,1)$ space becomes the $XY$-plane, \ie the plane spanned by the $X$ and $Y$ eigenstates; this is the plane of the Bloch sphere perpendicular to the $Z$ axis,
		\item the $(1,0)$ space becomes the $YZ$-plane defined analogously, and
		\item the $(1,1)$ space becomes the $XZ$-plane, again defined analogously.
	\end{itemize}
	The condition that all measured inputs must be measured in the $(0,*)$ space translates to a similar condition involving the $XY$-plane.
	\item The single qubit Paulis become $X$, $Y$, and $Z$ only (no nontrivial powers needed), where $Y \propto XZ$.
	\item The \flow from Definition~\ref{Def:full flow} becomes the gflow of \cite{browneGeneralizedFlowDeterminism2007}; a detailed translation is explained in \cite{boothOutcomeDeterminismMeasurementbased2023}.
	The translation requires mapping the correction matrix $C$ to a correction function $c$ by associating columns of $C$ with set indicator functions, \ie $c(v) = \{ u \mid C_{u,v} = 1 \}$.
\end{itemize}

\paragraph{Focusing \flowintitles}
\begin{itemize}
	\item The notion of focused \flow collapses to that of focused gflow from \cite{backensThereBackAgain2021}.
	\item The existence of focused \flow generalises the existence of focused gflow from \cite[Section 3.1]{mhallaWhichGraphStates2014a} and \cite{backensThereBackAgain2021}.
\end{itemize}

\paragraph{Algebraic interpretation of focused \flowintitles}
\begin{itemize}
	\item The reduced adjacency matrix from Definition~\ref{Def:ram} has been considered over qubits in \cite[Definition~6]{mhallaWhichGraphStates2014a}.
	\item The definitions of the flow-demand matrix (Definition~\ref{Def:fdm}) and the order-demand matrix (Definition~\ref{Def:odm}) become the similarly-named matrices over qubits from \cite[Definition~3.4 and Definition~3.5]{mitosekAlgebraicInterpretationPauli2026}.
	\item The new algebraic formulation from Theorem~\ref{Th:main result} becomes the main result from our previous work \cite[Theorem~3.1]{mitosekAlgebraicInterpretationPauli2026}.
	In particular, $NC$ is no longer a weighted graph.
	\item The uniqueness of focused \flow when $|I|=|O|$ (Corollary~\ref{Cor:uniqueness}), becomes the uniqueness of gflow under the same restriction from \cite{backensThereBackAgain2021}.
	\item The focused vectors of Definition~\ref{Def:foc vec} become the focused sets from \cite[Definition~4.3]{simmonsRelatingMeasurementPatterns2021}.
	Originally, focused sets were defined in the setting of Pauli flow, yet they apply just as well to gflow.
	\item Theorem~\ref{Th:focused sets are kernel of M} becomes \cite[Theorem~3.23]{mitosekAlgebraicInterpretationPauli2026}.
\end{itemize}

\paragraph{Flow-finding}
The correspondence between qudit \flow-finding and qubit flow-finding has been carefully explained in Section~\ref{Sec:Flow-finding} and its corresponding appendices.

\paragraph{\flowintitles-preserving rewriting}
\begin{itemize}
	\item Regarding the qubit equivalents of the operations in Section~\ref{s:scaling-LC-pivot}, local scaling is trivial and local complementations are not parametrised.
	Unlike for qudits, there are actually two variants of local complementation on a qubit, which differ by a non-trivial graph stabiliser; conversely, one variant corresponds to three successive applications of the other about the same qubit.
	The definition of the qubit pivot uses both variants.
	Flow preservation of local complementation and pivoting for gflow was proved in \cite[Lemma~3.1 and Corollary~3.3]{backensThereBackAgain2021}.
	\item Flow reversal in Theorem~\ref{Th:flow-rev} translates to \cite[Theorem~3.19]{mitosekAlgebraicInterpretationPauli2026}.
	\item The notion of `$Z$-like measurements' from Section~\ref{Subsec:Z-like} captures $YZ$- and $XZ$-measured vertices only, since we are not treating Pauli vertices separately and hence there are no $Z$ measurements.
	\item $Z$-like removal from Theorem~\ref{Th:Z-like removal} corresponds to \cite[Lemma~3.4]{backensThereBackAgain2021}.
	\item $Z$-like insertion from Theorem~\ref{Th:Z-like insertion} corresponds to results from \cite[Section~4 and Section~5]{perezBackens2025}.
	\item Generation of instances with \flow from Theorem~\ref{Th:generation} translates to a result about generating MBQC patterns with different types of flow from the upcoming paper \cite{backensGenerating2026}.
\end{itemize}

\section{Proving focusing}\label{Sec:ProvingFocusing}

To prove Theorem~\ref{Th:flow focusing}, we first prove that a \LOFG has focused \flow if and only if it has \flow, by giving a procedure that allows any \flow to be modified in a stepwise fashion until it is focused, with each step being a valid \flow.
This procedure generalises \cite[Section 3.3]{backensThereBackAgain2021} to the qudit case.
First, we specify a definition of focusing that applies to a single correction set and a single vertex.

\begin{definition}\label{def:focused-individual}
	Let $(G,I,O,\ld)$ be a \LOFG with \flow $(C,\prec)$, and let $u,v\in V$ be two distinct vertices, i.e.\ $u\neq v$.
	We say the correction $C_{*,u}$ is focused with respect to $v$ if the following two statements hold:
	\begin{itemize}
        \setlength{\itemindent}{0.35em}
		\item[\conlab{FXI}] $v\in O \vee (C_{v,v} = 0 \implies (GC)_{v,u} = 0)$, and
		\item[\conlab{FZI}] $C_{v,v} \ne 0 \implies C_{v,u} = 0$.
	\end{itemize}
\end{definition}

\begin{observation}\label{obs:focusing-individual-general}
	An \flow is focused according to Definition~\ref{Def:focused} if and only if \conref{FXI} and \conref{FZI} hold for all distinct pairs $u,v\in V$.
	Indeed, if Definition~\ref{def:focused-individual} holds for all distinct pairs of vertices $u,v$, then Definition~\ref{Def:focused} is immediate.
	
	For the converse, note that if $v\in O$, then $C_{v,v} = 0$ by \conref{ZF}, so both \conref{FXI} and \conref{FZI} are satisfied for any $C_{*,u}$.
	Similarly, if $(C_{v,u}, (GC)_{v,u}) = (0,0)$, then \conref{FXI} and \conref{FZI} are satisfied for $C_{*,u}$ and $v$.
	In particular, if $u\in O$, then $C_{w,u}=0$ for all $w\in V$ by \conref{ZF}; this then also implies $(GC)_{v,u} = \sum_{w\in V} G_{v,w} C_{w,u} = 0$.
	
	Recall that if $v\in\comp{O}$, then $\ld_X(v) = C_{v,v}$ by \conref{CF}.
	Thus, if an \flow $(C,\prec)$ is focused according to Definition~\ref{Def:focused}, then \conref{FXI} and \conref{FZI} hold for all distinct pairs $u,v\in\comp{O}$.
	Moreover, the two conditions hold whenever at least one of $u,v$ is in $O$ by the argument above.
	Thus, the two conditions hold for all pairs of distinct vertices.
\end{observation}

We now show how to focus a flow in a step-wise fashion.
The first argument is to show the construction that focuses a new correction set with respect to some chosen vertex, without breaking focusing conditions that already hold for other vertices.

\begin{lemma}\label{lem:focus-single}
	Let $(G,I,O,\ld)$ be a \LOFG with \flow $(C,\prec)$, and let $u,v\in V$ be two distinct vertices, i.e.\ $u\neq v$.
	Then we can construct an \flow $(C',\prec)$ on $(G,I,O,\ld)$ such that $C'_{*,u}$ is focused with respect to $v$.
	Moreover, the construction ensures that if there exists a vertex $v'\neq v$ such that $C_{*,v}$ is focused with respect to $v'$, then, for any $w\in V$, $C'_{*,w}$ is focused with respect to $v'$ if and only if $C_{*,w}$ is focused with respect to $v'$.
\end{lemma}
\begin{proof}
	By Observation~\ref{obs:focusing-individual-general}, \conref{FXI} and \conref{FZI} are satisfied if at least one of $u,v$ is an output, or if $(C_{v,u}, (GC)_{v,u}) = (0,0)$; in those cases we are done by taking $C'=C$ and the second part of the lemma is trivial.
	Hence it only remains to consider the case where $u,v\in\comp{O}$ and $(C_{v,u}, (GC)_{v,u}) \ne (0,0)$; by \conref{OF} this implies $u\prec v$.
	We distinguish cases according to whether or not $C_{v,v}$ is 0.
	\begin{itemize}
		\item 	Suppose $C_{v,v} = \ld_X(v) = 0$, then \conref{FZI} is trivially satisfied.
		If $(GC)_{v,u}=0$, \conref{FXI} is also satisfied and we are done by taking $C'=C$.
		
		Otherwise, note that by \conref{CF} and Definition~\ref{Def:LOFG} we have $\ld_Z(v) = (GC)_{v,v} \ne 0$.
		Define a new correction matrix $C'$ as
		\[
		C'_{w,w'} = \begin{cases}
			C_{w,w'} &\text{if } w' \ne u \\
			C_{w,u} + \frac{d-(GC)_{v,u}}{\ld_Z(v)} C_{w,v} &\text{if } w' = u.
		\end{cases}
		\]
		Then by linearity:
		\[
		(GC')_{w,w'} = \begin{cases}
			(GC)_{w,w'} &\text{if } w' \ne u \\
			(GC)_{w,u} + \frac{d-(GC)_{v,u}}{\ld_Z(v)} (GC)_{w,v}  &\text{if } w' = u.
		\end{cases}
		\]
		In particular, $(GC')_{v,u} = (GC)_{v,u} + \frac{d-(GC)_{v,u}}{\ld_Z(v)} \ld_Z(v) = 0$ over $\Fd$, meaning \conref{FXI} is satisfied.
		Furthermore, $(C',\prec)$ satisfies Definition~\ref{Def:full flow}:
		\begin{itemize}
			\item The only column of $C$ that has changed is $C_{*,u}$.
			We have $C_{u,v}=0 = (GC)_{u,v}$ by \conref{OF} since $(GC)_{v,u}\ne 0$; thus $(C'_{u,u}, (GC')_{u,u}) = (C_{u,u}, (GC)_{u,u})$ and \conref{CF} is satisfied.
			\item If $w'\in O$, then $w'\neq u$ and thus $C'_{w,w'} = C_{w,w'}$.
			Now $C_{w,w'} = 0$ by \conref{ZF} applied to $C$, so the desired property holds for $C'$ as well.
			Similarly, if $w\in I$, then $C_{w,u} = C_{w,v} = 0$ by \conref{ZF} applied to $C$ and thus $C'_{w,u} = 0$.
			Hence $C'$ satisfies \conref{ZF}.
			\item Suppose $C'_{w,u}\ne 0$ for some $w$, then $C_{w,u} \ne 0$ or $C_{w,v} \ne 0$.
			In the former case, $u\prec w$ and in the latter case $u\prec v\prec w$.
			The argument is analogous for $GC'$, thus \conref{OF} is satisfied.
		\end{itemize}
		Hence $(C',\prec)$ is a valid \flow.
	
		\item Now suppose instead $C_{v,v} = \ld_X(v)\ne 0$, then \conref{FXI} is trivially satisfied.
		If $C_{v,u} = 0$, then \conref{FZI} is also satisfied and we are done.
		Otherwise, define a new correction matrix $C'$ as
		\[
		C'_{w,w'} = \begin{cases}
			C_{w,w'} &\text{if } w' \ne u \\
			C_{w,u} + \frac{d-C_{v,u}}{\ld_X(v)} C_{w,v} &\text{if } w' = u.
		\end{cases}
		\]
		Then $C_{v,u} = C_{v,u} + \frac{d-C_{v,u}}{\ld_X(v)} \ld_X(v) = 0$ over $\Fd$, meaning \conref{FZI} is satisfied.
		The argument that $(C',\prec)$ is a valid \flow is analogous to the one in the previous case.
	\end{itemize}
	
	We have shown how to satisfy \conref{FXI} and \conref{FZI} for the pair $u,v$.
	Now suppose there exists $v'\neq v$ such that $C_{*,v}$ is focused with respect to $v'$.
	Assume without loss of generality that $v'\in\comp{O}$, otherwise \conref{FXI} and \conref{FZI} are trivial by Observation~\ref{obs:focusing-individual-general}.
	From the argument that $(C',\prec)$ is a valid \flow at the end of the first case, we know that $C'_{v',v'} = C_{v',v'}$.
	Now, for any $w\in V$, we have $C'_{v',w} = C_{v',w} + \ell C_{v',v}$, where
	\[
		\ell = \begin{cases}
			0 &\text{if } w\neq u \\
			\frac{d-(GC)_{v,u}}{\ld_Z(v)} &\text{if } w=u \text{ and } \ld_X(v) = C_{v,v} = 0 \\
			\frac{d-C_{v,u}}{\ld_X(v)} &\text{if } w=u \text{ and } \ld_X(v) = C_{v,v} \ne 0.
		\end{cases}
	\]
	As $u,v$ are fixed, the value of $\ell$ depends only on $w$.
	Furthermore, by linearity, $(GC')_{v',w} = (GC')_{v',w} + \ell (GC')_{v',v}$.
	Thus:
	\begin{itemize}
		\item If $\ldX(v') = C_{v',v'} = 0$, then \conref{FZI} is trivial.
		 Moreover, $C_{*,v}$ being focused with respect to $v'$ implies $C_{v',v} = 0$.
		 Thus, for all $w\in V$, we have $C'_{v',w} = C_{v',w} + \ell C_{v',v} = C_{v',w}$, so $C'_{*,w}$ is focused with respect to $v'$ if and only if $C_{*,w}$ is.
		\item If $C_{v',v'} \ne 0$, then \conref{FXI} is trivial.
		 Moreover, $C_{*,v}$ being focused with respect to $v'$ implies $(GC)_{v',v} = 0$.
		 Thus, for all $w\in V$, we have $(GC')_{v',w} = (GC)_{v',w} + \ell (GC)_{v',v} = (GC)_{v',w}$, so again $C'_{*,w}$ is focused with respect to $v'$ if and only if $C_{*,w}$ is.
	\end{itemize}
	This completes the proof.
\end{proof}

We can now prove Theorem~\ref{Th:flow focusing}:

\begin{proof}[Proof of Theorem~\ref{Th:flow focusing}]
    By Observation~\ref{obs:focusing-individual-general}, instead of working with Definition~\ref{Def:focused}, we may consider Definition~\ref{def:focused-individual} for all pairs of distinct $u,v\in V$.
   	A focused \flow can thus be focused using the following procedure.
   	
   	Let $<$ be any totalisation of $\prec$; for simplicity we will from now on identify $V$ with the set $[n] := \{1,2,\ldots,n\}$, where $n=\abs{V}$.
   	We will prove the following property by structural induction on $k$ ranging from 1 to $n$:
   	\begin{quotation}
   		\noindent There exists an \flow $(C',\prec)$ such that for all $v\in [n]\setminus[n-k]$ and all $u\ne v$, $C'_{*,u}$ is focused with respect to $v$.
   	\end{quotation}
   	The base case is $k=1$, and the property holds for $C':=C$ since $C_{u,n} = 0 = (GC)_{u,n}$ for any $u<n$ by \conref{OF}.
   	Now suppose the property holds for some $k$ such that $1\leq k < n$ and denote the corresponding \flow by $(C,\prec)$.
   	Consider what happens for $k+1$.
   	Let $a := n-(k+1)$ and define the following sequence of correction matrices: $C^{(0)} := C$ and for $j\in [a-1]$,
   	\[
   		C^{(j)}_{u,v} := \begin{cases}
   			C^{(j-1)}_{u,v} &\text{if } v\ne j \\
   			C^{(j-1)}_{u,j} + \frac{d-(GC^{(j-1)})_{a,j}}{(GC^{(j-1)})_{a,a}} C^{(j-1)}_{u,a} &\text{if } v = j \text{ and } C^{(j-1)}_{a,a} = 0 \\
   			C^{(j-1)}_{u,j} + \frac{d-C^{(j-1)}_{a,j}}{C^{(j-1)}_{a,a}} C^{(j-1)}_{u,a} &\text{if } v = j \text{ and } C^{(j-1)}_{a,a} \ne 0.
   		\end{cases}
   	\]
   	Note that if $C^{(j-1)}_{*,j}$ is already focused with respect to $a$, then $C^{(j)} = C^{(j-1)}$.
   	Then by Lemma~\ref{lem:focus-single}, $C^{(j)}_{*,j}$ is focused with respect to $a$.
   	Moreover, by the same lemma, $C^{(j)}_{*,b}$ is focused with respect to $a$ for all $b \in [j]$, and together with the induction hypothesis, for any $j \in [a-1]$, any $v > a$, and any $u \ne v$, we have that $C^{(j)}_{*,u}$ is focused with respect to $v$.
   	Finally, $C^{(a-1)}_{*,u}$ is focused with respect to $a$ for all $u > a$ since \conref{OF} implies that $C^{(a-1)}_{a,u} = (GC^{(a-1)})_{a,u} = 0$.
   	Thus the property holds for $k+1$ with correction matrix $C' = C^{(a-1)}$.
   	
   	This means the property holds for $k=n$.
   	The final \flow $(C',\prec)$ has the property that for all $v\in V$ and all $u\in V\setminus\{v\}$, the correction $C'_{*,u}$ is focused with respect to $v$.
   	By Observation~\ref{obs:focusing-individual-general}, this means $(C',\prec)$ satisfies Definition~\ref{Def:focused}, i.e.\ it is focused.
\end{proof}

\section{Proofs regarding algebraic formulation}\label{Sec:ProvingAlgebraicFormulation}

The main result shown in this appendix is Theorem~\ref{Th:main result}, but we also prove some related results here.
In Subsection~\ref{subsec:prelims}, we lay ground for the main proof.
Next, in Subsection~\ref{subsec:mainproof}, we provide detailed proof of Theorem~\ref{Th:main result}.
Lastly, in Subsection~\ref{subsec:ProvingFocusedVectors}, we prove some properties of the focused vectors that are associated with the new algebraic formulation.

\subsection{Preliminary simplifications}\label{subsec:prelims}

Before moving to the main proof, we perform some initial simplification by showing that the \flow definition can be restricted to operate only on the reduced adjacency matrix rather than the full matrix.
Moreover, we define a special relation induced by the correction matrix, which is sufficient to determine whether that correction matrix is compatible with any \flow, thus eliminating the need to consider arbitrary partial orders.

Note that Condition~\conref{C} from Definition~\ref{Def:flow} can also be simplified to $\forall u \in \comp{O} . \ld(u) = (C_{u,u}, (AC)_{u,u})$ under the convention that non-existent entries of a matrix are treated as $0$s.
Such a non-existent entry appears \eg for $C_{u,u}$ with $u \in I$.
We will now show that the simplified \flow definition is equivalent to the original one.

\begin{lemma}
    There exists \flow according to Definition~\ref{Def:full flow} if and only if there exists \flow according to Definition~\ref{Def:flow}.
\end{lemma}

\begin{proof}
    $(\Rightarrow):$ suppose there exists an \flow $(C,\prec)$ on the \LOFG $(G,I,O,\ld)$ according to Definition~\ref{Def:full flow}. Let $C'$ be the reduced correction matrix obtained from $C$ and let $\prec'$ be the restriction of $\prec$ to an order on $\comp{O}$ only.
    This restriction is immediately a strict partial order.
    Now, from \conref{ZF} we know $C_{i,v} = 0$ for all $i \in I$ and $C_{v,o}=0$ for all $o \in O$ (with $v\in V$ arbitrary), which implies $C'$ contains all non-zero entries of $C$.
    As all the non-zero entries are preserved, we furthermore have $(AC)_{u,v} = (GC)_{u,v}$ for all $u \in \comp{O}$ and $v \in \comp{I}$.
    Therefore, condition~\conref{OF} implies~\conref{O}, and~\conref{CF} implies~\conref{C}.
    As a consequence, $(G,I,O,\ld)$ has \flow $(C',\prec')$ according to Definition~\ref{Def:flow}.

    $(\Leftarrow):$ suppose there exists an \flow $(C',\prec')$ on the \LOFG $(G,I,O,\ld)$ according to Definition~\ref{Def:flow}.
    Let $C$ be the extension of $C'$ to a matrix of shape $V \times V$ by filling all missing entries with $0$: this immediately makes $C$ satisfy \conref{ZF}.
    Let $\prec$ be the extension of $\prec'$ to an order on $V$ by adding relations $u \prec o$ for all $u \in \comp{O}$ and $o \in O$.
    This process makes $\prec$ a strict partial order.
    Similar to the proof above, $(AC)_{u,v} = (GC)_{u,v}$ for all $u \in \comp{O}$ and $v \in \comp{I}$ as once again $C'$ contains all non-zero entries of $C$.
    Therefore condition~\conref{C} implies~\conref{CF}.
    The matrices $C'$ and $AC'$ are lower-triangular according to any totalisation of $\prec'$, hence $C$ and $GC$ are lower triangular according to any totalisation of $\prec$, because ${\prec'} \subset {\prec}$ and rows and columns of $C$ not featured in $C'$ are identically $0$.
    In other words, condition~\conref{O} implies~\conref{OF}.
    Thus $(G,I,O,\ld)$ has \flow $(C,\prec)$ according to Definition~\ref{Def:full flow}.
\end{proof}

The definition of focused \flow (\cf Definition~\ref{Def:focused}) is compatible with both \flow definitions.
Thus, also, the existence of focused \flow with respect to one \flow definition is equivalent to the existence of focused \flow with respect to the other \flow definition:

\begin{corollary}
    There exists focused \flow according to Definition~\ref{Def:full flow} if and only if there exists focused \flow according to Definition~\ref{Def:flow}.
\end{corollary}

The usual flow definitions simultaneously describe the correction function and the strict partial order.
Yet the conditions~\conref{FX} and~\conref{FZ} do not depend on the strict partial order, but only on the correction matrix.
Thus, it makes sense to introduce the notion of a `focused correction matrix', dropping the partial order from the definition:

\begin{definition}[Focused correction matrix]
    Let $(G,I,O,\ld)$ be a \LOFG and $C$ a reduced correction matrix. We say that $C$ is \textit{focused} if it satisfies \conref{FX} and \conref{FZ}.
\end{definition}

\begin{definition}[Extensive correction matrix]
    Let $\Gamma = (G,I,O,\ld)$ be a \LOFG and let $C$ be a reduced correction matrix on $\Gamma$.
    We call $C$ \textit{extensive} when there exists a strict partial order $\prec$ such that $(C,\prec)$ is an \flow on $\Gamma$.
    We call $C$ \textit{focused extensive} when there exists $\prec$ such that $(C,\prec)$ is a focused \flow on $\Gamma$.
\end{definition}

When determining whether a correction matrix is extensive, one must determine whether there exists a strict partial order that satisfies condition~\conref{O} from Definition~\ref{Def:flow}.
It turns out that, for any correction matrix, there exists a unique relation that determines extensivity.

\begin{definition}[Induced relation]\label{Def:induced rel}
    Let $\Gamma = (G,I,O,\ld)$ be a \LOFG and let $C$ be a correction matrix on $\Gamma$.
    The \emph{induced relation} $\trl_C$ is the coarsest relation on $\comp{O}$ such that for any pair of distinct $u, v \in \comp{O}$ we have $u \trl_C v$ if and only if at least one of the following hold:\begin{itemize}
        \item[\conlab{O1}] $(AC)_{v,u} \ne 0$, or
        \item[\conlab{O2}] $C_{v,u} \ne 0$.
    \end{itemize}
    When $\trl_C$ is an acyclic relation, we denote the transitive closure of $\trl_C$ as $\prec_C$.
\end{definition}

The entries of the product $AC$ do not depend on any ordering of the vertices, and thus the above definition is sound.
The relation $\trl_C$ captures all information regarding which off-diagonal entries of $AC$ are non-zero, motivating the following lemma:

\begin{lemma}[Induced order containment]\label{Lemma:ind ord containment}
    Let $\Gamma = (G,I,O,\lambda)$ be a \LOFG and let $(C,\prec)$ be an \flow.
    Then ${\prec_C} \subseteq {\prec}$.
\end{lemma}

\begin{proof}
    Let $\prec'$ be any totalisation of $\prec$.
    By Definition~\ref{Def:flow}, both $AC$ and $C$ must be lower triangular with respect to the ordering $\prec'$.
    Suppose that $u \prec_C v$.
    By Definition~\ref{Def:induced rel}, $(AC)_{v,u} \ne 0$ or $C_{v,u} \ne 0$.
    Now if $v \prec' u$, then either $AC$ or $C$ would not be lower triangular, a contradiction.
    Since $\prec'$ is a total order and $\neg(v \prec' u)$, we get that necessarily $u \prec' v$.
    Repeating for all elements of $\prec_C$, we find that ${\prec_C} \subseteq {\prec'}$.
    This inclusion must hold for all totalisations $\prec'$ of $\prec$ and thus $\prec_C$ forms a subset of the intersection of all such totalisations.
    This intersection is precisely $\prec$, ending the proof.
\end{proof}

We now have all the parts needed to prove that an extensive correction matrix determines a valid \flow.

\begin{theorem}[Extending a correction matrix to an \flow]\label{Th:extending correction matrix to flow}
    Let $\Gamma = (G,I,O,\lambda)$ be a \LOFG and let $C$ be a correction matrix on $\Gamma$. Then $C$ is (focused) extensive if and only if $(C, \prec_C)$ is a (focused) \flow.
\end{theorem}

\begin{proof}
    $(\Rightarrow)$: Suppose that $C$ is extensive.
    Then there exists a strict partial order $\prec$ on $\comp{O}$ such that $(C,\prec)$ is an \flow.
    By Lemma~\ref{Lemma:ind ord containment}, ${\prec_C} \subseteq {\prec}$.
    Thus, since $\prec$ is a strict partial order, $\prec_C$ is also a strict partial order.
    Since $C$ is extensive, it must satisfy condition~\conref{C} (which does not depend on the partial order).
    Similarly, if $C$ is focused extensive, it must necessarily satisfy conditions \conref{FX} and \conref{FZ}, which are also independent of the partial order.
    It remains to show that $(C,\prec_C)$ satisfies condition~\conref{O}.
    Suppose for a contradiction that \conref{O} fails for $\prec_C$.
    Then, for some totalisation $\prec_C'$ of $\prec_C$, we must have $v \prec_C' u$ and either $C_{v,u} \ne 0$ or $(AC)_{v,u} \ne 0$.
    However, either of $C_{v,u}$ or $(AC)_{v,u}$ being non-zero implies $u \prec_C v$ and thus $u \prec_C' v$.
    This contradicts the assumption $v \prec_C' u$.
    Thus $(C,\prec_C)$ satisfies \conref{O} and hence forms an \flow, which moreover is focused if and only if $C$ is focused extensive.

    $(\Leftarrow)$: Immediate from the definitions.
\end{proof}

Theorem~\ref{Th:extending correction matrix to flow} states that in order to determine whether $C$ is (focused) extensive, it is sufficient to check whether $(C,\prec_C)$ forms a (focused) \flow.
In fact, based on the proof, a slightly stronger property holds.
If $C$ satisfies condition~\conref{C} (and conditions~\conref{FX} and~\conref{FZ}) and $\prec_C$ is a strict partial order, then $(C,\prec_C)$ is a (focused) \flow, effectively skipping verification of condition~\conref{O} apart from ensuring that $\prec_C$ is a strict partial order.
This allows us to avoid considering arbitrary totalisations of $\prec_C$ when verifying the \flow conditions.
Furthermore, since $\prec_C$ is a transitive closure of $\trl_C$, $\prec_C$ is a strict partial order if and only if $\trl_C$ forms a directed acyclic graph.
Together, this amounts to the following simpler algebraic formulation of \flow:

\begin{theorem}\label{Th:partial alg}
    Let $\Gamma = (G,I,O,\ld)$ be a \LOFG and let $C$ be any correction matrix on $\Gamma$.
    Then $C$ is extensive if and only if $C$ satisfies condition~\conref{C} and $\prec_C$ is a strict partial order.
    Assuming the above holds, $C$ is focused extensive if and only if it satisfies conditions~\conref{FX} and~\conref{FZ}.
\end{theorem}

\subsection{Main proof}\label{subsec:mainproof}

We now move on to proving the main result, Theorem~\ref{Th:main result}.
The first step is to split condition~\conref{C} into several smaller parts that will be encoded in different parts of the algebraic \flow definition:

\begin{observation}\label{Obs:C split}
    Let $(G,I,O,\lambda)$ be a \LOFG and let $C$ be a correction matrix of shape $\comp{I}\times\comp{O}$.
    Define the following conditions for all $v \in \comp{O}$:
    \begin{enumerate}
        \item[\conlab{C1}] $\ldX(v) = 0 \rightarrow (AC)_{v,v} = \ldZ(v)$
        \item[\conlab{C2}] $\ldX(v) \ne 0 \rightarrow C_{v,v} = \ldX(v)$
        \item[\conlab{C3}] $\ldX(v) = 0 \rightarrow C_{v,v} = 0$
        \item[\conlab{C4}] $\ldX(v) \ne 0 \rightarrow (AC)_{v,v} = \ldZ(v)$
    \end{enumerate}
    Then, \conref{C} is equivalent to the four above conditions, as $\ld_X(v) = C_{v,v}$ is captured by \conref{C2} and \conref{C3}, while $\ld_Z(v) = (AC)_{v,v}$ is captured by \conref{C1} and \conref{C4}.
\end{observation}

The goal is to show that $MC = \id{\comp{O}}{\comp{O}}$ corresponds to \conref{C1}, \conref{C2}, \conref{FX}, and \conref{FZ}; while $NC$ being a DAG corresponds to \conref{C3}, \conref{C4}, and the condition that $\prec_C$ is a strict partial order.
We split the proof of this into two lemmata:

\begin{lemma}\label{Lemma:MC}
    Let $(G,I,O,\ld)$ be a \LOFG and let $C$ be a correction matrix of shape $\comp{I}\times\comp{O}$.
    Then $MC = \id{\comp{O}}{\comp{O}}$ if and only if $C$ satisfies \conref{C1}, \conref{C2}, \conref{FX}, and \conref{FZ}.
\end{lemma}
\begin{proof}
    First, consider some relevant properties of the entries of $M$ in relation to the entries of $A$ and $C$.
    Let $v \in \comp{O}$ and $u\in\comp{I}$, then by Definition~\ref{Def:fdm}:\begin{equation*}
        M_{v,u} = \begin{cases}
        	\ldZ(v)^{-1} A_{v,u} & \text{if}\ \ldX(v) = 0\\
        	\ldX(v)^{-1} \delta_{v,u} & \text{if}\ \ldX(v) \ne 0.
        \end{cases}
    \end{equation*}
    Let $w\in\comp{O}$ and consider two cases depending on $\ldX(v)$:
    \begin{enumerate}
        \item Suppose $\ldX(v) = 0$ and thus $M_{v,u} = \ldZ(v)^{-1} A_{v,u}$.
        Multiplying by $C_{u,w}$, we get:\begin{equation*}
            M_{v,u} C_{u,w} = \ldZ(v)^{-1} A_{v,u} C_{u,w}
        \end{equation*}
        \noindent and hence by linearity:\begin{equation}
            (MC)_{v,w} = \ldZ(v)^{-1} (AC)_{v,w} \label{MCvwX0}
        \end{equation}

        \item Suppose $\ldX(v) \ne 0$ and thus $M_{v,u} = \ldX(v)^{-1} \delta_{v,u}$.
        When $\ldX(v) \ne 0$, then necessarily $v \notin I$ and thus $M$ has a $v$-column.
        Multiplying $M_{v,u}$ be ${C_{u,w}}$, we get:\begin{equation*}
            M_{v,u} C_{u,w} = \ldX(v)^{-1} \delta_{v,u} C_{u,w}
        \end{equation*}
        \noindent which equals $0$ unless $v=u$ when it equals $\ldX(v)^{-1} C_{v,w}$.
        Using this fact, we find that:\begin{equation}
            (MC)_{v,w} = \sum_{u \in \comp{I}} M_{v,u} C_{u,w} = \ldX(v)^{-1} C_{v,w} \label{MCvwXne0}
        \end{equation}
    \end{enumerate}
    \noindent Now we can split the proof into two directions.

    $(\Rightarrow):$ Suppose $MC = \id{\comp{O}}{\comp{O}}$; this implies $(MC)_{v,v} = 1$ and $(MC)_{v,w} = 0$ for all $w \in \comp{O}\setminus\{v\}$.
    The following hold:
    \begin{enumerate}
        \item For all $v$ such that $\ldX(v) = 0$, by Equation~\eqref{MCvwX0}, we have $1 = (MC)_{v,v} = \ldZ(v)^{-1}(AC)_{v,v}$.
        Thus $(AC)_{v,v} = \ldZ(v)$, \ie \conref{C1} holds.
        \item Furthermore, for all $v$ such that $\ldX(v) = 0$ and for all $w$ such that $v \ne w$, by Equation~\eqref{MCvwX0} we have $0 = (MC)_{v,w} = \ldZ(v)^{-1} (AC)_{v,w}$.
        Since $\ldZ(v) \ne 0$ (as $\ld(v) \ne (0,0)$), this implies $(AC)_{v,w} = 0$ and thus \conref{FX} holds.
        \item For all $v$ such that $\ldX(v) \ne 0$, by Equation~\eqref{MCvwXne0}, we have $1 = (MC)_{v,v} = \ldX(v)^{-1}C_{v,v}$.
        Thus $C_{v,v} = \ldX(v)$, \ie \ref{C2} holds.
        \item Finally, for all $v$ such that $\ldX(v) \ne 0$ and for all $w$ such that $v \ne w$, by Equation~\ref{MCvwXne0} we have $0 = (MC)_{v,w} = \ldX(v)^{-1} C_{v,w}$.
        This implies $C_{v,w} = 0$ and thus \conref{FZ} holds.
    \end{enumerate}
    Hence the desired properties are all satisfied.

    $(\Leftarrow):$ Suppose the four properties hold.
    We must show that $(MC)_{v,v} = 1$ and $(MC)_{v,w} = 0$ for all $v \ne w \in \comp{O}$.
    Observe that all four points in the $(\Rightarrow)$ part of the proof are reversible.
    Therefore $(MC)_{v,v} = 1$ follows from reversing the two first points of the $(\Rightarrow)$ part of the proof and $(MC)_{v,w} = 0$ follows from reversing the latter two points of the $(\Rightarrow)$ part of the proof.
    Then $MC = \Id_{\comp{O}}^{\comp{O}}$ as desired.
\end{proof}

The proof for the product $NC$ being a DAG follows the same structure.
We split it up into three parts for simplicity.
Before proving those lemmata, first note that for any $v,w \in \comp{O}$, by Definition~\ref{Def:odm}:
\begin{equation*}
	N_{v,u} = \ldX(v) A_{v,u} - \ldZ(v) \delta_{v,u}
\end{equation*}
\noindent for any $u \in \comp{I}$.
Since $\delta_{v,u} = 0$ unless $u=v$, we get the following:\begin{align}\label{eq:NC-decomposition}
	(NC)_{v,w} = \sum_{u \in \comp{I}} N_{v,u} C_{u,w} = \left(\sum_{u \in \comp{I}} \ldX(v) A_{v,u} C_{u,w}\right) - \ldZ(v) C_{v,w} = \ldX(v)(AC)_{v,w} - \ldZ(v) C_{v,w}
\end{align}
\noindent where we use the convention $C_{v,w} = 0$ for $v \in I$ to make sense of the final expression despite the rows of $NC$ being labelled by $\comp{O}$ and the rows of $C$ being labelled by $\comp{I}$.

\begin{lemma}\label{lem:NC-zero-diagonal}
	Let $(G,I,O,\ld)$ be a \LOFG and let $C$ be a correction matrix of shape $\comp{I}\times\comp{O}$ satisfying \conref{C1}, \conref{C2}, \conref{FX}, and \conref{FZ}.
	Then the main diagonal of $NC$ is identically zero if and only if \conref{C3} and \conref{C4} hold.
\end{lemma}
\begin{proof}
	Fix $v=w\in\comp{O}$ in \eqref{eq:NC-decomposition} and consider two cases depending on whether $\ld_X(v)$ is zero:
	\begin{enumerate}
		\item Suppose $\ldX(v) = 0$.
		Then $(NC)_{v,v} = -\ldZ(v) C_{v,v}$ which equals $0$ if and only if $C_{v,v} = 0$.
		Considering all $v$ with $\ldX(v) = 0$, we obtain equivalence of $(NC)_{v,v} = 0$ to \conref{C3}.
		\item Suppose $\ldX(v) \ne 0$.
		By the assumption that $C$ satisfies \conref{C2}, we have $C_{v,v} = \ldX(v)$.
		Then:
		\begin{equation*}
			(NC)_{v,v} = \ldX(v)(AC)_{v,v} - \ldZ(v) C_{v,v} = \ldX(v)(AC)_{v,v} - \ldZ(v) \ldX(v) = \ldX(v)((AC)_{v,v} - \ldZ(v))
		\end{equation*}
		\noindent which equals $0$ if and only if $(AC)_{v,v} = \ldZ(v)$.
		Considering all $v$ with $\ldX \ne 0$, we obtain equivalence of $(NC)_{v,v} = 0$ to \conref{C4}.
	\end{enumerate}
	Thus, the above two cases imply that the main diagonal of $NC$ is identically zero if and only if \conref{C3} and \conref{C4} hold.
\end{proof}

\begin{lemma}\label{lem:NC-trl-C}
	Let $(G,I,O,\ld)$ be a \LOFG and let $C$ be a correction matrix of shape $\comp{I}\times\comp{O}$ satisfying \conref{C1}, \conref{C2}, \conref{FX}, and \conref{FZ}.
	Then the off-diagonal part of $NC$ corresponds to $\trl_C$ in the sense that, for any pair of distinct $v,w\in\comp{O}$, we have $w \trl_C v$ if and only if $(NC)_{v,w} \ne 0$.
\end{lemma}
\begin{proof}
	Let $v,w \in \comp{O}$ be distinct.
	We show that $w \trl_C v$ if and only if $(NC)_{v,w} \ne 0$.
	
	$(\Rightarrow)$: Suppose that $w \trl_C v$.
		As before, we consider two cases depending on whether $\ld_X(v)$ is zero:
		\begin{enumerate}
			\item Suppose $\ldX(v) = 0$.
			By the assumption that $C$ satisfies \conref{FX}, we have $(AC)_{v,w} = 0$.
			Thus the condition of~\ref{O1} does not hold.
			Ergo the condition of~\ref{O2} necessarily applies, leading to $C_{v,w} \ne 0$.
			Therefore \eqref{eq:NC-decomposition} becomes:\begin{equation*}
				(NC)_{v,w} = \ldX(v)(AC)_{v,w} - \ldZ(v) C_{v,w} = -\ldZ(v) C_{v,w} \ne 0
			\end{equation*}
			
			\item Suppose $\ldX(v) \ne 0$.
			By the assumption that $C$ satisfies \conref{FZ}, we have $C_{v,w} = 0$.
			Thus the condition of~\ref{O2} does not apply.
			Ergo the condition of~\ref{O1} necessarily applies, leading to $(AC)_{v,w} \ne 0$.
			Therefore \eqref{eq:NC-decomposition} becomes:\begin{equation*}
				(NC)_{v,w} = \ldX(v)(AC)_{v,w} - \ldZ(v) C_{v,w} = \ldX(v) (AC)_{v,w} \ne 0
			\end{equation*}
		\end{enumerate}
		In both cases, $(NC)_{v,w} \ne 0$ as desired.
		
	$(\Leftarrow)$: Suppose that $(NC)_{v,w} \ne 0$, which means:\begin{equation*}
			\ldX(v)(AC)_{v,w} - \ldZ(v) C_{v,w} \ne 0
		\end{equation*}
		\noindent If both $(AC)_{v,w} = 0$ and $C_{v,w} = 0$, the above cannot hold.
		Thus at least one of conditions~\conref{O1} or~\conref{O2} applies, and we get $w \trl_C v$.
		
	Therefore the off-diagonal part of $NC$ corresponds to $\trl_C$.
\end{proof}

\begin{lemma}\label{Lemma:NC}
    Let $(G,I,O,\ld)$ be a \LOFG and let $C$ be a correction matrix of shape $\comp{I}\times\comp{O}$ satisfying \conref{C1}, \conref{C2}, \conref{FX}, and \conref{FZ}.
    Then $NC$ forms a (weighted) directed acyclic graph if and only if $C$ satisfies \conref{C3}, \conref{C4}, and $\prec_C$ is a strict partial order.
\end{lemma}
\begin{proof}
    Let $v,w \in \comp{O}$.
    We combine the previous two lemmata to get the full proof.
    
    $(\Rightarrow)$: Under the assumption that $NC$ forms a (weighted) DAG, by Lemma~\ref{lem:NC-trl-C}, we get that $\trl_C$ is a DAG and thus $\prec_C$ is a strict partial order.
    Furthermore, $NC$ being a DAG implies $(NC)_{v,v} = 0$ for all $v \in \comp{O}$ and thus, by Lemma~\ref{lem:NC-zero-diagonal}, conditions~\conref{C3} and~\conref{C4} both hold.
    
    $(\Leftarrow)$: \conref{C3} and~\conref{C4} imply that $(NC)_{v,v} = 0$ by Lemma~\ref{lem:NC-zero-diagonal} and $\prec_C$ being a strict partial order implies that $\trl_C$ is a DAG.
    The latter implies that the off-diagonal part of $NC$ is a DAG by Lemma~\ref{lem:NC-trl-C}; by combining the two parts the entire matrix $NC$ forms a DAG.
\end{proof}

We can now prove the main result.

\begin{proof}[Proof of Theorem~\ref{Th:main result}]
    $(\Rightarrow)$: suppose that $\Gamma$ has \flow.
    Then by Theorem~\ref{Th:flow focusing} it has a focused \flow $(C,\prec)$ and by Theorem~\ref{Th:extending correction matrix to flow} $(C,\prec_C)$ is also a focused flow.
    Hence, by Theorem~\ref{Th:partial alg}, $C$ satisfies \conref{C}, \conref{FX}, \conref{FZ}, and $\prec_C$ is a strict partial order.
    By Observation~\ref{Obs:C split}, $C$ satisfies \conref{C1}, \conref{C2}, \conref{C3}, and \conref{C4}.
    Therefore, by Lemma~\ref{Lemma:MC} we get $MC = \id{\comp{O}}{\comp{O}}$, and by Lemma~\ref{Lemma:NC} we get that $NC$ forms a (weighted) DAG.

    $(\Leftarrow)$: all steps in the $(\Rightarrow)$ part of the proof follow from if and only if statements and thus the above argument can be reversed.
\end{proof}

\subsection{Proofs regarding focused vectors}\label{subsec:ProvingFocusedVectors}

Before proving Theorem~\ref{Th:focused sets are kernel of M}, we first show that given $\somevector\in\FF_d^{\comp{I}}$, one can find a maximal set over which $\somevector$ is focused (cf.\ Definition~\ref{Def:foc vec}) by considering the product $M\somevector$, where $M$ is the flow-demand matrix.
Note that the definition of focused vectors involves the reduced adjacency matrix $A$, whereas we now use the flow-demand matrix; this change is justified in the following lemma.

\begin{lemma}\label{Lemma:M product gives gocusing condition}
    Let $(G,I,O,\ld)$ be a \LOFG and let $\somevector \in \FF_d^{\comp{I}}$ be a vector.
    Then, $\somevector$ is focused over $\comp{O} \setminus \supp{M\somevector}$ and $\somevector$ is not focused over any larger subset of the non-outputs.
\end{lemma}
\begin{proof}
    Let $F_{\somevector} = \comp{O} \setminus \supp{M\somevector}$.
    By inspecting conditions \conref{Fs1} and~\conref{Fs2}, for any sets $S,S'\sse\comp{O}$, the vector $\somevector$ is simultaneously focused over $S$ and focused over $S'$ if and only if it is focused over $S \cup S'$.
    Therefore, it is sufficient to show that for all $v \in \comp{O}$: $v \in F_{\somevector}$ if and only if $\somevector$ is focused over $\{v\}$.

    $(\Rightarrow)$: Let $v \in F_{\somevector}$.
    Then $v \notin \supp{M\somevector}$, so $(M\somevector)_v = 0$.
    We consider two cases depending on $\ldX(v)$, following a similar proof structure as in Lemma~\ref{Lemma:MC}:
    \begin{enumerate}
        \item Suppose $\ldX(v) = 0$, then:
        \begin{equation*}
            0 = (M\somevector)_v = \sum_{u \in \comp{I}} M_{v,u}\somevector_u = \ldZ(v)^{-1} \sum_{u \in \comp{I}} A_{v,u} \somevector_u = \ldZ(v)^{-1} (A\somevector)_v
        \end{equation*}
        \noindent Since $\ld_X(v) = 0$, necessarily $\ldZ(v) \ne 0$.
        This implies $(A\somevector)_v = 0$, \ie $v \notin \supp{A\somevector}$.
        By contraposition: if $v \in \supp{A\somevector}$, then $\ldX(v) \ne 0$, \ie \conref{Fs2} holds.
        \item Suppose $\ldX(v) \ne 0$.
        Since $M_{v,u} = 0$ for $v \ne u$, we have:
        \begin{equation*}
            0 = (M\somevector)_v = \sum_{u \in \comp{I}} M_{v,u}\somevector_u = \ld_X(v)^{-1} \somevector_v
        \end{equation*}
        \noindent Then as $\ld_X(v) \ne 0$, we have $\somevector_v = 0$, \ie $v \notin \supp{\somevector}$.
        Again by contraposition: if $v \in \supp{A}$, then $\ldX(v) = 0$, meaning \conref{Fs1} holds.
    \end{enumerate}
    Hence $\somevector$ satisfies conditions~\conref{Fs1} and~\conref{Fs2} over the set $\{ v \}$ as required.

    $(\Leftarrow)$: Let $\somevector$ be focused over $\{ v \}$.
    Once again, we consider two cases on $\ldX(v)$:
    \begin{enumerate}
        \item Suppose $\ldX(v) = 0$, then by \conref{Fs2}: $v \notin \supp{A\somevector}$, \ie $(A\somevector)_v = 0$.
        Following the same sequence of equalities as in the first case of the $(\Rightarrow)$ part of the proof, we find $(M\somevector)_v = 0$.
        \item Suppose $\ldX(v) \ne 0$, then by \conref{Fs1}: $v \notin \supp{\somevector}$, \ie $(\somevector)_v = 0$.
        Following the same sequence of equalities as in the second case of the $(\Rightarrow)$ part of the proof, we again find $(M\somevector)_v = 0$.
    \end{enumerate}
    Therefore, $(M\somevector)_v = 0$, meaning $v \notin \supp{M\somevector}$, ergo $v \in F_{\somevector}$, which ends the proof.
\end{proof}

Using the above lemma, we find that the kernel of the flow-demand matrix forms is exactly the set of focused vectors.

\begin{proof}[Proof of Theorem~\ref{Th:focused sets are kernel of M}]
    By Lemma~\ref{Lemma:M product gives gocusing condition}, a vector $\somevector \in \FF_d^{\comp{I}}$ is focused over $\comp{O} \setminus \supp{M\somevector}$.
    Now, suppose $\somevector \in \mathfrak{F}_{\Gamma}$.
    Then $\somevector$ is focused over $\comp{O}$.
    Hence $\comp{O} = \comp{O} \setminus \supp{M\somevector}$.
    But $M\somevector$ is a vector in $\FF_d^{\comp{O}}$, thus $\supp{M\somevector} = \emptyset$, \ie $\somevector \in \ker M$.
    All these steps are reversible, which leads to $\mathfrak{F}_{\Gamma} \cong \ker M_{\Gamma}$.
\end{proof}

\section{Finite field arithmetic}\label{Sec:FFA}
In finite fields, the complexity of arithmetic operations (by which we mean addition, negation, multiplication, and inversion) depends on the order $d$ of the field \cite{gashkovComplexityComputationFinite2013}.
In main part of the paper, we consider $d$ to be fixed (or at least bounded above), \ie we treat the cost of arithmetic operations in $\Fd$ as $\bigO(1)$.
In general, the cost of such operations can be expressed as $\Theta(\poly(\log d))$, \ie it is polynomial in $\log d$.
The representation of an element of a finite field is possible in space $\Theta(\log d)$.
Addition and negation take $\Theta(\log d)$ time, but the cost for multiplication and inversion is slightly higher.
It is known that $\bigO(\log d \log\log d)$ time suffices for multiplication and $\bigO(\log d \log\log^2 d)$ for inversion \cite{gashkovComplexityComputationFinite2013,harveyIntegerMultiplicationTime2021,harveyPolyMultFiniteFields2022}.
The result from \cite{harveyPolyMultFiniteFields2022} is critical here and is conditional on a widely believed hypothesis regarding prime numbers; this requirement may be omitted for fields with $d = p^{\ell}$ where $\ell$ is small (and $p$ is prime) \cite[Theorem 3.35 and Theorem 3.41]{sutherlandLectureNotes2025}).
In particular, the stated complexity holds unconditionally for prime fields.
However, the methods with the best known complexity bounds are impractical for small $d$.
Yet again, for bounded $d$, we can treat these costs as $\bigO(1)$.

The exact complexity results, \ie without fixing $d$, can be derived from the ones stated in the main part of the paper by including the cost of arithmetic operations.
For the time complexity, that means increasing lower bounds by $\log d$ factors, and increasing upper bounds by $\bigO(\log d \log\log^2 d)$.
The memory complexity increases by $\Theta(\log d)$, \eg our \flow-finding method requires $\bigO(n^2 \log d)$ memory.

\section{Pseudocode}\label{Sec:Pseudocode}

\algrenewcommand\algorithmicdo{}
\algrenewcommand\algorithmicthen{}
\begin{breakablealgorithm}
\caption{Construction of flow and order-demand matrices}\label{M and N pseudo}
\begin{algorithmic}[1]
\Statex Returns the flow-demand matrix of a given \LOFG
\Procedure{Flow-DemandMatrix}{$G,I,O,\lambda$}
    \State initialise $\comp{O} \times \comp{I}$ matrix $M$
    \For {$v \in \comp{O}$}
        \For {$u \in \comp{I}$} \Comment{Loop over all entries of $M$}
            \If $\ld_X(v) = 0$
                \State $M_{v,u} = \ld(Z)^{-1} G_{v,u}$
            \Else
                \State $M_{v,u} = \ld(X)^{-1} \delta_{v,u}$
            \EndIf
        \EndFor
    \EndFor
    \State \Return $M$
\EndProcedure
\[\]
\Statex Returns the order-demand matrix of a given \LOFG
\Procedure{Order-DemandMatrix}{$G,I,O,\lambda$}
    \State initialise $\comp{O} \times \comp{I}$ matrix $N$
    \For {$v \in \comp{O}$}
        \For {$u \in \comp{I}$} \Comment{Loop over all entries of $N$}
            \State $N_{v,u} = \ld_X(v)G_{v,u} - \ldZ(v)\delta_{v,u}$
        \EndFor
    \EndFor
    \State \Return $M$
\EndProcedure
\end{algorithmic}
\end{breakablealgorithm}

The pseudocode for the general \flow-finding algorithm (which works even if $\abs{I}\neq\abs{O}$) then differs from the corresponding qubit algorithm in \cite[Appendix B]{mitosekAlgebraicInterpretationPauli2026} on only a few lines:
\begin{itemize}
    \item On line~\ref{step:find-v-dependent-rows CHANGE}, we must look for non-zero values rather than values exactly equal to $1$. (In the $d=2$ case, the only possible non-zero value is $1$.)
    \item On lines~\ref{step:initial eliminations CHANGE}, \ref{step:row-addition-rk CHANGE}, and~\ref{step:manual GE CHANGE}, instead of simply adding a row to another row, we must add a row multiplied by some value such that the addition will cancel an undesired dependency in the second row. (In the $d=2$ case, the multiplicity is always $1$.)
    \item On a new line~\ref{step:new multiply row EXTRA CHANGE}, we must multiply the row by a suitable value so that its pivot becomes $1$. (In the $d=2$ case, a pivot always necessarily equals $1$.)
\end{itemize}
Furthermore, while line~\ref{step:linear system solution EXTRA CHANGE} matches the previous (qubit) algorithm, we point out that the linear system comprises equations of the following form, where $val$ corresponds to the value from the second block:
\begin{equation*}
    a_1x_1 + \dots + a_{n_O-n_I}x_{n_O-n_I} + val = 0
\end{equation*}
\noindent and not:\begin{equation*}
    a_1x_1 + \dots a_{n_O-n_I}x_{n_O-n_I} = val
\end{equation*}
It means, that the values encoded in the matrix $\starpart$ (initialised on line~\ref{step:initialise-L-P-2}) are effectively negations of the values one would obtain if interpreting the system in the usual way. (When $d=2$, negation leaves elements invariant.)

\begin{breakablealgorithm}
\caption{Finding flow in the general case}\label{general algo}
\begin{algorithmic}[1]
\Statex Checks if a \LOFG $\Gamma = (G,I,O,\lambda)$ has \flow and, if it exists, returns the focused flow.
\Procedure{FindFlowGeneral}{$G,I,O,\lambda$}
    \State $M = \Call{Flow-DemandMatrix}{G,I,O,\lambda}$\label{step:M}
    \State $N = \Call{Order-DemandMatrix}{G,I,O,\lambda}$\label{step:N}
    \If {$\rank M \ne n-n_O$}\label{step:rank M}
        \State \Return `NO FLOW EXISTS'
    \EndIf
    \State $C_0 = \text{any right inverse of $M$}$\label{step:any right inverse}
    \State $F = \text{matrix with columns forming basis of $\ker M$}$\label{step:kernel}
    \State $C' = \left[ C_0 \mid F \right]$\label{step:C'}
    \State $\rowbasechange{N}{\mathcal{B}} = NC'$\label{step:basis change}
    \State $N_L = \text{first $n-n_O$ columns of $\rowbasechange{N}{\mathcal{B}}$}$\label{step:NL-and-NR-1}
    \State $N_R = \text{last $n_O-n_I$ columns of $\rowbasechange{N}{\mathcal{B}}$}$\label{step:NL-and-NR-2}
    \State $\ILS, \LS = \left[ N_R \mid N_L \mid Id_{\comp{O}} \right]$ \Comment{Two independent copies of the same linear system}\label{step:LS-and-ILS}
    \Statex \Comment{From now, we refer to the three parts of $\ILS$ and $\LS$ as the first block, the second block, and the third block}
    \State $\text{bring first block of $\LS$ into row echelon form}$\label{step:gaussian-eliminate-LS}
    \State $\solved = \emptyset$ \label{step:initialise-L-P-1}
    \State $\text{initialize $(n_O-n_I) \times (n-n_O)$ matrix $\starpart$}$ \Comment{Columns of $\starpart$ correspond to non-outputs $\comp{O}$}\label{step:initialise-L-P-2}
    \While {$\solved \ne \comp{O}$}: \label{step:while-loop}
        \State $\text{find the first row $r_z$ in $\LS$ whose first $n_O-n_I$ entries equal $0$}$ \label{step:find-post-pivot-rows}
        \State $\tosolve = \{ v \in \comp{O} \setminus \solved \mid \text{column $v$ in the first $\LS$ block has all entries from row $r_z$ on equal to $0$} \}$ \label{step:find-L}
        \If {$\tosolve = \emptyset$}
            \State \Return `NO FLOW EXISTS'
        \EndIf
        \For {$v \in \tosolve$}:\label{step:loop solutions}
            \State $\starpart_{*,v} = \Call{SolveLinearSystem}{\text{first block of $\LS$, $v$ column in second block of $\LS$}}$\label{step:linear system solution EXTRA CHANGE}
        \EndFor
        \For {$v \in \tosolve$}:\label{step:main for loop}
            \State $\solved = \solved \cup \left\{ v \right\}$ \label{step:update-S}
            \State $R = \left[ r \mid \text{$r$ is a row whose intersection with the $v$ column in the third block of $\LS$ is $\ne 0$} \right]$\label{step:find-v-dependent-rows CHANGE}
            \State $r_{last} = \text{last element of $R$}$
            \For {$r \in R[:-1]$} \Comment{Iterate over all but last element}\label{step:row-additions-r1-rk-1}
                \State add row $r_{last}$ of $\LS$ with correct multiplicity to row $r$ in $\LS$ to remove dependency\label{step:initial eliminations CHANGE}
            \EndFor
            \State $\text{add row $v$ of $\ILS$ with correct multiplicity to row $r_{last}$ in $\LS$ to remove dependency}$ \label{step:row-addition-rk CHANGE}
            \For $r \in \text{rows of $\LS$ except $r_{last}$}$: \Comment{Iterate from top}\label{step:simplify-rk}
                \If {$\text{row $r$ of $\LS$ has first $n_O-n_I$ entries $0$}$}
                    \State Break
                \EndIf
                \State $y = \text{column of the leading $1$ in row $r$ of $\LS$}$
                \If {$\text{intersection of $r_{last}$ and column $y$ in $\LS$ is $\ne 0$}$}
                    \State $\text{add row $r$ with correct multiplicity to $r_{last}$ in $\LS$ to remove dependency}$\label{step:manual GE CHANGE}
                \EndIf
            \EndFor
            \State multiply $r_{last}$ by a correct value that makes the first non-zero entry of $r_{last}$ in the first block equal $1$\label{step:new multiply row EXTRA CHANGE}
            \State swap $r_{last}$ with other rows to bring first block of $\LS$ back into row echelon form\label{step:swap-rk}
        \EndFor
    \EndWhile
    \State $\colbasechange{C}{\mathcal{B}} = \left[ \frac{Id_{\comp{O}}}{\starpart} \right]$\label{step:build-c-B}
    \State \Return $\left(C' \colbasechange{C}{\mathcal{B}}, N C'\colbasechange{C}{\mathcal{B}}\right)$\label{step:output}
\EndProcedure
\end{algorithmic}
\end{breakablealgorithm}

\section{Run of flow-finding algorithm}\label{Sec:Flow run}

We now provide a detailed overview of how the pseudocode from Appendix~\ref{Sec:Pseudocode} runs on the \LOFG from Figure~\ref{Fig:running example}.

\paragraph{Construction of the linear system}
Firstly, in lines~$\ref{step:M}$ and~$\ref{step:N}$, the flow-demand and order-demand matrices are created, these are the same matrices as presented in Figures~\ref{Fig:M} and~\ref{Fig:N}.
Since $\rank M = 4$, the check in line~\ref{step:rank M} passes.
Next, in lines~\ref{step:any right inverse}, \ref{step:kernel}, and~\ref{step:C'} three further matrices are created: any right inverse $C'$ of $M$, a basis $F$ of $\ker M$, and the concatenation $C_0$ of these two.
For example:\begin{equation*}
    C_0 = [C'\mid F] = \scalemath{0.85}{\begin{pNiceArray}{c:ccc|c}[first-row, first-col]
          & i & a & b & d & F_1 \\
        a & 3 & 0 & 0 & 0 & 0 \\
        b & 0 & 0 & 2 & 0 & 0 \\
        d & 0 & 0 & 0 & 4 & 0 \\
        \hdottedline
        o_1 & 2 & 2 & 4 & 1 & 1 \\
        o_2 & 1 & 0 & 3 & 1 & 3 \\
    \end{pNiceArray}}
\end{equation*}
\noindent where we use $F_1$ as the label of the single column in $F$.
Next, in line~\ref{step:basis change}, we compute the product of $N$ and $C'$, denoted $N_{\mathcal{B}}$:\begin{equation*}
    N_{\mathcal{B}} = \scalemath{0.85}{\begin{pNiceArray}{c:ccc|c}[first-row, first-col]
          & i & a & b & d & F_1 \\
        i & 0 & 0 & 0 & 0 & 0 \\
        \hdottedline
        a & 1 & 0 & 0 & 0 & 0 \\
        b & 1 & 2 & 4 & 0 & 1 \\
        d & 3 & 0 & 2 & 3 & 3 \\
    \end{pNiceArray}}
\end{equation*}
In lines~\ref{step:NL-and-NR-1} and~\ref{step:NL-and-NR-2}, we break $N_{\mathcal{B}}$ into two submatrices $N_L$ and $N_R$, and use them to obtain two (initially identical) copies $\LS$ and $\ILS$ of the following linear system.
We give primed labels to columns from the third block, to distinguish them from the corresponding columns in the second block:\begin{equation}\label{eq:original-system}
    \ILS = \LS = \scalemath{0.85}{\begin{pNiceArray}{c|cccc|cccc}[first-row, first-col]
          & F_1 & i & a & b & d & i' & a' & b' & d' \\
        i & 0 & 0 & 0 & 0 & 0 & 1 & 0 & 0 & 0 \\
        \hdottedline
        a & 0 & 1 & 0 & 0 & 0 & 0 & 1 & 0 & 0 \\
        b & 1 & 1 & 2 & 4 & 0 & 0 & 0 & 1 & 0 \\
        d & 3 & 3 & 0 & 2 & 3 & 0 & 0 & 0 & 1 \\
    \end{pNiceArray}}
\end{equation}
Next, in line~\ref{step:gaussian-eliminate-LS}, we perform Gaussian elimination on the first block of $\LS$, \ie on the first column.
The independent copy $\ILS$ remains unchanged.
The Gaussian elimination involves adding the third row to the fourth row with multiplicity $2$ so that the original pivot of the fourth block becomes $0$.
After that, we swap the first and third rows.
We drop the previous row labels, as they no longer make sense after performing row operations. (Nevertheless, the relationship of each modified row to the original vertex-labelled rows continues to be encoded in the third block.)
We obtain the following updated $\LS$ linear system:
\begin{equation}\label{eq:solve-b}
    \LS = \scalemath{0.85}{\begin{pNiceArray}{c|cccc|cccc}[first-row, first-col]
          & F_1 & i & a & b & d & i' & a' & b' & d' \\
        1 & 1 & 1 & 2 & 4 & 0 & 0 & 0 & 1 & 0 \\
        2 & 0 & 1 & 0 & 0 & 0 & 0 & 1 & 0 & 0 \\
        3 & 0 & 0 & 0 & 0 & 0 & 1 & 0 & 0 & 0 \\
        4 & 0 & 0 & 4 & 0 & 3 & 0 & 0 & 2 & 1 \\
    \end{pNiceArray}}
\end{equation}
We initialise $\solved$ and $\starpart$ in lines~\ref{step:initialise-L-P-1} and~\ref{step:initialise-L-P-2} to:\begin{equation*}
    \starpart = \scalemath{0.85}{\begin{pNiceArray}{cccc}[first-row, first-col]
          & i & a & b & d\\
        F_1 & * & * & * & *\\
    \end{pNiceArray}} \qquad\qquad \solved = \emptyset
\end{equation*}

\paragraph{First run of the main loop}
We go into the main loop in line~\ref{step:while-loop}.
During the first run, we identify that vertex $b$ can be solved (line~\ref{step:find-L}).
This is because the entries of column $b$ in the second block of the matrix in \eqref{eq:solve-b} are all $0$ from the second row onwards.
Therefore $\tosolve = \left\{ b \right\}$.
In line~\ref{step:linear system solution EXTRA CHANGE}, we identify the solution of the following linear system with a single equation (arising from the first block of $\LS$ together with column $b$), where we name the unique variable $f_1$:\begin{equation*}
    1f_1 + 4 = 0
\end{equation*}
As the running example is over $\mathbb{F}_5$, $f_1 = 1$ is the only solution:\begin{equation*}
    \starpart = \scalemath{0.85}{\begin{pNiceArray}{cccc}[first-row, first-col]
          & i & a & b & d\\
        F_1 & * & * & 1 & *\\
    \end{pNiceArray}}
\end{equation*}

Now, we must transform the system of linear equations to the form it would have had if $b$ had never been included (this is the loop starting in line~\ref{step:main for loop}).
We add $b$ to the set $\solved$ of solved vertices in line~\ref{step:update-S}.
Next, we identify the set $R$ of rows in the updated matrix that depend on the original row $b$ of \eqref{eq:original-system} by looking for non-zero values in column $b'$ of \eqref{eq:solve-b}.
This yields $R = \left\{ 1, 4 \right\}$ and $r_{last} = 4$.

We enter the next for loop in line~\ref{step:row-additions-r1-rk-1}: here, this loop runs only once, for $r = 1$, \ie the first row.
In line~\ref{step:row-addition-rk CHANGE}, we add row $r_{last}$, \ie the fourth row, to the first row to cancel dependency on original $b$ row from the latter.
From column $b'$, the first row depends on $1$ times the original row $b$ while the fourth row depends on $2$ times the original row $b$.
Thus, we must add $2$ times the fourth row to the first row to cancel the dependency, obtaining:\begin{equation}\label{eq:dependency-cancelling}
    \LS = \scalemath{0.85}{\begin{pNiceArray}{c|cccc|cccc}[first-row, first-col]
        & F_1 & i & a & b & d & i' & a' & b' & d' \\
        1 & 1 & 0 & 4 & 1 & 0 & 0 & 0 & 0 & 2 \\
        2 & 0 & 1 & 0 & 0 & 0 & 0 & 1 & 0 & 0 \\
        3 & 0 & 0 & 0 & 0 & 0 & 1 & 0 & 0 & 0 \\
        4 & 0 & 0 & 4 & 0 & 3 & 0 & 0 & 2 & 1 \\
    \end{pNiceArray}}
\end{equation}

We exit the loop that started on line~\ref{step:row-additions-r1-rk-1}.
At this point, only the fourth row has any dependency on row $b$ from initial system $\ILS$.
More specifically, the dependency is on $2$ times row $b$.
We remove this dependency in line~\ref{step:row-addition-rk CHANGE} by subtracting 2 times the $b$ row of $\ILS$ from the fourth row of the current $\LS$ given in \eqref{eq:dependency-cancelling}.
This is equivalent to adding 3 times row $b$ from $\ILS$ to row 4 of $\LS$, since $-2 = 3$ in $\FF_5$; it yields:\begin{equation*}
    \LS = \scalemath{0.85}{\begin{pNiceArray}{c|cccc|cccc}[first-row, first-col]
        & F_1 & i & a & b & d & i' & a' & b' & d' \\
        1 & 1 & 0 & 4 & 1 & 0 & 0 & 0 & 0 & 2 \\
        2 & 0 & 1 & 0 & 0 & 0 & 0 & 1 & 0 & 0 \\
        3 & 0 & 0 & 0 & 0 & 0 & 1 & 0 & 0 & 0 \\
        4 & 3 & 3 & 0 & 2 & 3 & 0 & 0 & 0 & 1 \\
    \end{pNiceArray}}
\end{equation*}

In the subsequent for loop from line~\ref{step:simplify-rk}, we iterate over those rows of $\LS$ containing pivots, to cancel as much of fourth row as possible.
Here, this leads to the following system, obtained by adding 2 times the first row to the fourth row:\begin{equation}\label{eq:dependency-a}
    \LS = \scalemath{0.85}{\begin{pNiceArray}{c|cccc|cccc}[first-row, first-col]
        & F_1 & i & a & b & d & i' & a' & b' & d' \\
        1 & 1 & 0 & 4 & 1 & 0 & 0 & 0 & 0 & 2 \\
        2 & 0 & 1 & 0 & 0 & 0 & 0 & 1 & 0 & 0 \\
        3 & 0 & 0 & 0 & 0 & 0 & 1 & 0 & 0 & 0 \\
        4 & 0 & 0 & 0 & 0 & 0 & 0 & 0 & 0 & 0 \\
    \end{pNiceArray}}
\end{equation}
In the second run of the loop from line~\ref{step:simplify-rk}, we immediately break because the second row of $\LS$ has no pivot in the first block.
Finally, in lines~\ref{step:new multiply row EXTRA CHANGE} and~\ref{step:swap-rk} nothing happens, as the fourth row is already identically zero in the first block (or in fact in all blocks) and is correctly placed.
This completes the first run of the main loop from line~\ref{step:while-loop}.

\paragraph{Second run of the main loop}
In the second run of the main loop, we observe that vertices $a$ and $d$ can be solved, but not the vertex $i$.
This is because the columns $a$ and $d$ in the second block are identically zero from the second row onwards, but the same is not true for column $i$.
Hence, $\tosolve = \{ a, d \}$.
We find and store the solutions for these vertices.
The linear system for $a$ is:\begin{equation*}
    1f_1 + 0 = 0
\end{equation*}
\noindent hence $f_1 = 0$ is the only solution.
The linear system for $d$ is:\begin{equation*}
    1f_1 + 1 = 0
\end{equation*}
\noindent hence $f_1 = -1 = 4$ is the only solution.
Thus, at the start of line~\ref{step:main for loop}, we have:\begin{equation*}
    \starpart = \scalemath{0.85}{\begin{pNiceArray}{cccc}[first-row, first-col]
            & i & a & b & d\\
        F_1 & * & 0 & 1 & 4\\
    \end{pNiceArray}}
\end{equation*}
Now, we have to remove the dependencies on the original rows $a$ and $d$.
We start with $a$: here the situation is simple, as only the second row of \eqref{eq:dependency-a} depends on $a$ (as indicated by column $a'$ in the third block).
Thus, after the entire run of the loop starting on line~\ref{step:main for loop}, with $v = a$, we obtain the updated system:\begin{equation*}
    \LS = \scalemath{0.85}{\begin{pNiceArray}{c|cccc|cccc}[first-row, first-col]
        & F_1 & i & a & b & d & i' & a' & b' & d' \\
        1 & 1 & 0 & 4 & 1 & 0 & 0 & 0 & 0 & 2 \\
        2 & 0 & 0 & 0 & 0 & 0 & 0 & 0 & 0 & 0 \\
        3 & 0 & 0 & 0 & 0 & 0 & 1 & 0 & 0 & 0 \\
        4 & 0 & 0 & 0 & 0 & 0 & 0 & 0 & 0 & 0 \\
    \end{pNiceArray}}
\end{equation*}

In the next run of the loop starting on line~\ref{step:main for loop}, we have $v = d$.
Here, again only one row, the first row, depends on the original row $d$ and the dependency has multiplicity $2$.
After this run of the loop, we get the system:\begin{equation*}
    \LS = \scalemath{0.85}{\begin{pNiceArray}{c|cccc|cccc}[first-row, first-col]
        & F_1 & i & a & b & d & i' & a' & b' & d' \\
        1 & 0 & 0 & 0 & 0 & 0 & 0 & 0 & 0 & 0 \\
        2 & 0 & 0 & 0 & 0 & 0 & 0 & 0 & 0 & 0 \\
        3 & 0 & 0 & 0 & 0 & 0 & 1 & 0 & 0 & 0 \\
        4 & 0 & 0 & 0 & 0 & 0 & 0 & 0 & 0 & 0 \\
    \end{pNiceArray}}
\end{equation*}
In the process, $\solved$ is updated to contain $b$, $a$, and $d$.
This ends the second run of the main loop from line~\ref{step:while-loop}.

\paragraph{Third run of the main loop}
In the third run of the loop starting on line~\ref{step:while-loop}, $i$ can now be solved.
The corresponding linear system is:\begin{equation*}
    0f_1 + 0 = 0
\end{equation*}
\noindent which is singular and consistent.
It is satisfied by any value of $f_1$.
So, for example, we can choose the solution $f_1 = 0$, leading to:\begin{equation*}
    \starpart = \scalemath{0.85}{\begin{pNiceArray}{cccc}[first-row, first-col]
            & i & a & b & d\\
        F_1 & 0 & 0 & 1 & 4\\
    \end{pNiceArray}}
\end{equation*}

Vertex $i$ is added to $\solved$, \ie at this point $\solved = \{ b, a, d, i \} = \comp{O}$.
Afterwards, the system is still updated to the form it would have had if $i$ was never included, which here will result in a system that has $0$s everywhere.

\paragraph{Final steps}
Since $\solved = \comp{O}$, no more runs of the loop from line~\ref{step:while-loop} occur.
Instead, on line~\ref{step:build-c-B}, we construct $C^{\mathcal{B}}$:\begin{equation*}
    C^{\mathcal{B}} = \scalemath{0.85}{\begin{pNiceArray}{c:ccc}[first-row, first-col]
          & i & a & b & d \\
        i & 1 & 0 & 0 & 0 \\
        \hdottedline
        a & 0 & 1 & 0 & 0 \\
        b & 0 & 0 & 1 & 0 \\
        d & 0 & 0 & 0 & 1 \\
        \hline
        F_1 & 0 & 0 & 1 & 4 \\
    \end{pNiceArray}}
\end{equation*}
\noindent and finally output $C'C^{\mathcal{B}}$, which is the reduced correction matrix, and $NC'C^{\mathcal{B}}$, which encodes the partial order:\begin{equation*}
    C'C^{\mathcal{B}} = \scalemath{0.85}{\begin{pNiceArray}{c:ccc}[first-row, first-col]
          & i & a & b & d \\
        a & 3 & 0 & 0 & 0 \\
        b & 0 & 0 & 2 & 0 \\
        d & 0 & 0 & 0 & 4 \\
        \hdottedline
        o_1 & 2 & 2 & 0 & 0 \\
        o_2 & 1 & 0 & 1 & 3 \\
    \end{pNiceArray}} \qquad \qquad NC'C^{\mathcal{B}} = \scalemath{0.85}{\begin{pNiceArray}{c:ccc}[first-row, first-col]
          & i & a & b & d \\
        i & 0 & 0 & 0 & 0 \\
        \hdottedline
        a & 1 & 0 & 0 & 0 \\
        b & 1 & 2 & 0 & 4 \\
        d & 3 & 0 & 0 & 0 \\
    \end{pNiceArray}}
\end{equation*}

Note that the \flow found here differs from that presented in Figures~\ref{Fig:running example} and~\ref{Fig:running example algebraic}, which is fine as the example has more outputs than inputs and so it can have multiple focused \flow{}s.
If, during the third run of the main loop, we had chosen the solution for $i$ to be $3$ instead of $0$, then we would instead recover the solution in Figure~\ref{Fig:running example algebraic}.

Lastly, we note that the solutions found each run of the main loop from line~\ref{step:while-loop} form the `layers' of the partial order.
For the worked example, these layers are thus $[b]$, $[a,d]$, and $[i]$.
The algorithm for \flow-finding from \cite{boothOutcomeDeterminismMeasurementbased2023} would produce the same layers.
However, our method also returns the induced partial order $\prec_C$ corresponding to the correction matrix $C$, as encoded in the product $NC = NC'C^{\mathcal{B}}$.
This order is the minimal one (\cf Lemma~\ref{Lemma:ind ord containment}).
In general, the partial order returned by Algoritm~\ref{general algo} could be less restrictive than an order found via layers\footnote{Nevertheless, for this particular example, the minimal order is precisely the same as the order implied by ordering via found layers, \ie $i \prec a \prec b$ and $i \prec d \prec b$.}.

\section{Lower bound for the complexity of \flowintitles-finding}\label{Sec:LowerBound}

The complexity bound of our flow-finding algorithm is $\bigO(n^3)$ only when $|I|\ne|O|$.
For $|I|=|O|$, we invoke subroutines for the problem of matrix inversion and multiplication, which, at least in principle, can be solved with sub-cubic complexity.
Here, we show that, in fact, reducing the other way is also possible, thus tightly linking the complexity of flow-finding with the aforementioned fundamental problems of linear algebra.

Firstly, the Theorem~\ref{Th:main result} implies the following corollary, which is more suitable for the discussion of lower bounds:

\begin{corollary}\label{Cor:flow finding intro}
    Let $\Gamma = (G,I,O,\ld)$ be a \LOFG, let $M$ be its flow-demand matrix, and let $N$ be its order-demand matrix.
    Then $\Gamma$ has \flow if and only if there exists a reduced correction matrix $C$ over $\FF_d$ such that:\begin{itemize}
        \item $C$ has shape $\comp{I} \times \comp{O}$,
        \item $MC = \id{\comp{O}}{\comp{O}}$, and
        \item $NC$ is the adjacency matrix of a directed acyclic $\Fd$-weighted graph.
    \end{itemize}
    The \flow on $\Gamma$ implied by the existence of this matrix $C$ is $(C,\prec_c)$, where furthermore the following holds for any $v,w \in \comp{O}$:
    \begin{itemize}
        \item $v \trl_c w \Leftrightarrow (NC)_{w,v} \ne 0$ and
        \item $\prec_c$ is the transitive closure of $\trl_c$.
    \end{itemize}
\end{corollary}

\begin{proof}
    By inspection, this follows from Theorem~\ref{Th:main result} and its supporting Lemmata~\ref{Lemma:M product gives gocusing condition} and~\ref{Lemma:NC}.
\end{proof}

We adapt the lower bound argument for Pauli flow on qubits \cite[Theorem 4.3]{mitosekAlgebraicInterpretationPauli2026} to \flow:

\begin{theorem}[\flow-finding lower bound]\label{Th:lower bound}
    The problem of finding \flow is at least as computationally expensive as the problem of finding the inverse of a matrix over $\Fd$.
\end{theorem}
\begin{proof}
    Let $M$ be an arbitrary $n \times n$ matrix over $\Fd$.
    Define $\Gamma := (G,I,O,\ld)$ as follows:\begin{align*}
        I &= \left( i_1, \dots, i_n \right) \\
        O &= \left( o_1, \dots, o_n \right) \\
        V &= \left( I \cup O \right) \\
        G_{i_k,o_{\ell}} &= M_{k,\ell} \quad \left( \forall k,\ell \in \left\{ 1, \dots, n \right\} \right) \\
        \ld(i_k) &= (0,1) \quad \left( \forall k \in \left\{ 1, \dots, n \right\} \right)
    \end{align*}
    Then the flow-demand matrix $M_{\Gamma}$ of $(G,I,O,\ld)$ is precisely equal to $M$ and the order-demand matrix $N_{\Gamma}$ is identically zero.
    Thus, by Corollary~\ref{Cor:flow finding intro}, $\Gamma$ has \flow if and only if $M$ is invertible.
    If \flow exists, the correction matrix is given by the inverse of $M$.
\end{proof}

The above theorem leads to an $\Omega(n^2)$ lower bound for \flow-finding.
While this bound is already implied by the size of the input, it is generally expected that matrix inversion and multiplication require $\omega(n^2)$ field operations; currently a formal proof of this is known only for matrices over $\mathbb{R}$ or $\mathbb{C}$ \cite[Section~28.2]{cormenIntroductionAlgorithmsFourth2022}.

\section{Details about \flowintitles-preserving rewriting}

This section goes into more detail about the \flow-preserving rewrite rules and related procedures summarised in Section~\ref{Sec:Flow-preserving rewriting}.
In Subsection~\ref{s:graph-operations}, we first explain known results about which graph operations can be implemented using local Clifford unitaries.
We use these in Subsection~\ref{s:pivot} to formally define pivoting and show that it is \flow-preserving.
In Subsection~\ref{Sec:proof reversibility}, we prove that \flow reversal is also \flow-preserving.
In Subsection~\ref{Sec:proving removal and insertion}, we show how to remove and insert Z-like vertices in an \flow-preserving way.
Finally, in Subsection~\ref{Sec:proving generation}, we combine Z-like insertion with other Clifford-based transformations to sketch an algorithm that generates arbitrary instances of {\LOFG\unskip}s with \flow.

\subsection{Graph operations that can be implemented via local Clifford unitaries}
\label{s:graph-operations}

Several relevant rewrite rules were already defined and proved flow-preserving by Booth \cite[Section~5.1]{boothMeasurementbasedQuantumComputation2022}.
These rules are based on the two different types of operations on an $\FF_d$-weighted graph which translate to local Clifford equivalences of the corresponding qudit graph states \cite[Section~D]{bahramgiriGraph2006}.
The first of these operations is \emph{local scaling}, which multiplies the weights of all edges incident on a chosen vertex $w$ by some non-zero constant $\gamma\in\FF_d^*$.
Denoting by $\_\overset{\gamma}{\circ}w$ the local scaling operation just described, we have:

\begin{equation}\label{eq:local-scaling-Clifford}
	(G\overset{\gamma}{\circ}w)_{u,v}
	:= \begin{cases}
		\gamma G_{u,v} &\text{if } u=w \text{ or } v=w \\
		G_{u,v} &\text{otherwise}
	\end{cases}
	\qquad\quad\text{and}\qquad\quad
	\ket{G\overset{\gamma}{\circ}w} = \multiplier{\gamma^{-1}}_w \ket{G},
\end{equation}
where $\multiplier{\alpha}\ket{k} = \ket{\alpha k}$ for any $\alpha,k\in\FF_d$ and the subscript $w$ indicates that this single-qudit Clifford unitary is applied to qudit~$w$~\cite[Section~5.1.1]{boothMeasurementbasedQuantumComputation2022}.

The second operation is \emph{local complementation}, which is also parametrised by a constant\footnote{Unlike for local scaling, we may allow $\gamma=0$ in this case. This yields a trivial transformation that leaves the graph invariant.} $\gamma\in\FF_d$ and changes the weights of all edges among the neighbours of some chosen vertex $w$ in a more complicated way.
Denoting by $\_\overset{\gamma}{\star} w$ the local complementation operation of weight $\gamma$, we have:
\begin{equation}\label{eq:lc-definition}
	(G\overset{\gamma}{\star} w)_{u,v}
	= \begin{cases}
		G_{u,v} + \gamma G_{u,w} G_{w,v} &\text{if } u\neq v \\
		G_{u,v} &\text{otherwise.}
	\end{cases}
\end{equation}
Note that $G_{u,w}G_{w,v} = 0$ unless both $u$ and $v$ are neighbours of $w$, so it is only edges between two neighbours of $w$ that change.

Let $\tau$ be the $d$-th root of unity that satisfies $\tau^2 = \omega$ and $\tau^d = 1$.
Denote by $\phasegate$ a qudit phase gate\footnote{Note that this is not the usual qudit phase gate, which is generally defined as $S\ket{k} = \tau^{k(k-1)} = \phasegate Z^{-2^{-1}}$, where $2^{-1}$ is the multiplicative inverse of $2$ in $\FF_d$.} satisfying $\phasegate\ket{k} = \tau^{k^2} \ket{k}$ and by $\tilde{\phasegate} := H \phasegate H^{-1}$ an equivalent gate in the $X$-eigenbasis.
Then the corresponding relationship between graph states is\footnote{This relationship is not entirely consistent with that given in \cite[Section~5.1.2]{boothMeasurementbasedQuantumComputation2022}, yet the differences consist only of Pauli operations which can be pushed past the Paulis used to define measurement spaces and corrections up to some scalar factor. Hence we can nevertheless use results from the latter source about the changes of measurement planes and the preservation of \flow during local complementations.}
\begin{equation}\label{eq:lc-Clifford}
	\ket{G\overset{\gamma}{\star} w} = \tilde{\phasegate}^\gamma_w \prod_{v\in V} \phasegate^{-\gamma G_{w,v}^2}_v \ket{G}
\end{equation}
\cf \cite[Proposition~12 of arXiv version]{boothCompleteZXcalculiStabiliser2022}, see Observation~\ref{obs:zx-lc} in Appendix~\ref{s:pivot} for the translation from the ZX-calculus.
Both of the above operations have been shown to preserve \flow according to Definition~\ref{Def:full flow}, i.e.\ working with full (not reduced) correction matrices.
By Theorem~\ref{Th:flow focusing}, this implies that focused \flow is also preserved.
Yet the changes necessary to preserve focusing might be more complicated\footnote{C.f.\ \cite[Lemma D.15]{simmonsRelatingMeasurementPatterns2021} for the explicit expression of the focused flow after a local complementation in the qubit case.} and will not be worked out here.

\begin{proposition}[{\cite[Proposition~5.5]{boothMeasurementbasedQuantumComputation2022}}]\label{prop:local-scaling}
	Suppose the \LOFG $(G,I,O,\ld)$ has an \flow given by $(C,\prec)$.
	Then, for any $w\in V$ and any non-zero $\gamma\in\ZZ_d^*$, the \LOFG $(G\overset{\gamma}{\circ} w,I,O,\ld')$ with new measurement labels
	\[
	\ld'(u)
	:= \begin{cases}
		(\gamma^{-1} C_{w,w}, \gamma (GC)_{w,w}) &\text{if } u=w \\
		(C_{u,u}, (GC)_{u,u}) &\text{otherwise}
	\end{cases}
	\]
	has an \flow $(C',\prec)$ defined as
	\[
	C'_{u,v}
	:= \begin{cases}
		\gamma^{-1} C_{w,v} &\text{if } u=w \\
		C_{u,v} &\text{otherwise.}
	\end{cases}
	\]
\end{proposition}

\begin{proposition}[{\cite[Proposition~5.8]{boothMeasurementbasedQuantumComputation2022}}]\label{prop:lc}
	Suppose the \LOFG $(G,I,O,\ld)$ has an \flow given by $(C,\prec)$.
	Then, for any $w\in\comp{I}$ and any $\gamma\in\ZZ_d$, the \LOFG $(G\overset{\gamma}{\star} w,I,O,\ld')$ with new measurement labels\footnote{Note there is a typo in the definition of the updated labels in \cite[Eq.~(5.61)]{boothMeasurementbasedQuantumComputation2022}, this has been corrected here.}
	\[
	\ld'(u)
	:= \begin{cases}
		(C_{w,w} - \gamma (GC)_{w,w}, (GC)_{w,w}) &\text{if } u=w \\
		(C_{u,u}, (GC)_{u,u} - \gamma G_{u,w}^2 C_{u,u}) &\text{otherwise}
	\end{cases}
	\]
	has an \flow $(C',\prec)$ defined as
	\[
	C'_{u,v}
	:= \begin{cases}
		C_{w,v} - \gamma (GC)_{w,v} &\text{if } u = w \\
		C_{u,v} &\text{otherwise.}
	\end{cases}
	\]
\end{proposition}

\begin{remark}
	The above transformation of the flow is much simpler than that presented in the qubit local complementation proof for (extended) gflow \cite[Lemma~3.1]{backensThereBackAgain2021}.
	It can in fact be adapted to give a simpler transformation for gflow, though again this transformation does not preserve focusing.
\end{remark}

In the binary case, given an edge $ww'$ in a graph, the sequence of three local complementations alternating about the two endpoints of the edge (in either order), $G \wedge w w' = ((G\star w) \star w') \star w = ((G\star w') \star w) \star w'$, is known as a \emph{pivot} or an \emph{edge-local complementation}.
Pivots also give rise to flow-preserving transformations and they interact better with focused gflow and Pauli flow than individual local complementations \cite[Lemma~5.3]{perezBackens2025}.
Again in the qubit case, the associated graph states satisfy $\ket{(G\star w) \star w') \star w} = H_w H_{w'} \prod_{v\in N_G(w)\cap N_G(w')} Z_v \ket{G}$, i.e.\ a pivot corresponds to applying Hadamard gates to the two endpoints of the edge and $Z$-gates to all joint neighbours of both endpoints.

Generalising the qubit case, a family of pivot operations on qutrits was defined in \cite[Theorem~3.5]{townsend-teagueSimplification2022}.
Another work presents a qudit generalisation of the pivot-and-delete operation, where the vertices $w$ and $w'$ are deleted from the graph after the pivot \cite[Lemma~10]{poorQupitStabiliserZXtravaganza2023}.

Motivated by the elegance of the qubit pivot operation, and the usefulness of having a graph operation that corresponds to Hadamard gates on certain qudits, we generalise the pivot operation to qudits.

\begin{proposition}[Formal version of Proposition~\ref{prop:pivot-informal}]\label{prop:pivot}
	Let $G$ be an $\FF_d$-weighted graph with vertices $V$, suppose $w, w'\in V$ satisfy $\epsilon := G_{w,w'} \ne 0$, and define $G\wedge w w'$ to be the graph on $V$ with the following adjacency matrix:
	\[
	(G\wedge w w')_{u,v} := 
	\begin{cases}
		0 &\text{if } u = v \\
		- \epsilon &\text{if } \{u,v\} = \{w,w'\} \\
		G_{w',v} &\text{if } u=w \text{ and } v\notin\{w,w'\} \\
		G_{u,w'} &\text{if } v=w \text{ and } u\notin\{w,w'\} \\
		G_{w,v} &\text{if } u=w' \text{ and } v\notin\{w,w'\} \\
		G_{u,w} &\text{if } v=w' \text{ and } u\notin\{w,w'\} \\
		G_{u,v} - \epsilon^{-1} (G_{u,w'} G_{w,v} + G_{u,w} G_{w',v}) &\text{otherwise,}
	\end{cases}
	\]
	Then the associated qudit graph state satisfies:
	\[
	\ket{G\wedge w w'}
	= F^{(-\epsilon)}_w F^{(-\epsilon)}_{w'} \prod_{v\in V} R_v^{2\epsilon^{-1} G_{v,w}G_{v,w'}} \ket{G}
	\]
	where $F^{(\gamma)} := \sum_{j,k\in\FF_d} \omega^{\gamma j k } \ketbra{j}{k}$ for any $\gamma\in\FF_d^*$, and $R = \sum_{k\in\FF_d} \tau^{k^2} \proj{k}$.
\end{proposition}

As proving this result using standard matrix notation showed itself to be onerous and error-prone, we prove the above proposition in Appendix~\ref{s:pivot} using the qudit stabiliser ZX-calculus of \cite{poorQupitStabiliserZXtravaganza2023}, which is also introduced in that appendix.

Analogous to the qubit case, a pivot operation can be built up from three successive local complementations.
Yet for qudits, at least one local scaling operation is also needed to get a pivot operation that is symmetric under interchange of $w$ and $w'$.

\begin{observation}
	Let $G$ be an $\FF_d$-weighted graph with vertices $V$, suppose $w, w'\in V$ satisfy $G_{w,w'} \ne 0$, and define $\gamma := G_{w,w'}^{-1}$.
	Then:
	\[
	G \wedge w w'
	= (((G\overset{\gamma}{\star} w) \overset{-\gamma}{\star} w') \overset{\gamma}{\star} w) \overset{-1}{\circ} w
	= (((G\overset{\gamma}{\circ} w')\overset{1}{\star} w) \overset{-1}{\star} w') \overset{1}{\star} w)\overset{-1/\gamma}{\circ} w
	\]
	and similarly with the roles of $w,w'$ swapped.
\end{observation}

The proof is straightforward but sufficiently onerous that we leave it out.
Nevertheless, as a pivot corresponds to a sequence of local complementations and local scalings, it is immediate that it preserves the existence of \flow.

\subsection{Proving the pivoting property using the qudit stabiliser ZX-calculus}
\label{s:pivot}

The ZX-calculus is a graphical language for reasoning about quantum computations.
While the qubit version is the most developed, there is also a complete ZX-calculus for the stabiliser fragment of $d$-dimensional qudits for every prime $d$.
The case $d=2$ was shown in \cite{backensZXcalculusCompleteStabilizer2014} and odd primes were treated in \cite{boothCompleteZXcalculiStabiliser2022,poorQupitStabiliserZXtravaganza2023}.
Since generalising the pivot rule from the qubit case to qudits using algebraic techniques requires long, detailed, but not very insightful manipulations, and there is already a ZX-calculus proof of the related pivot-and-delete rule (also called ``pivot simplification'') \cite[Lemma~10]{poorQupitStabiliserZXtravaganza2023}, we prove a (non-deletion version) of the pivot rule using the qudit ZX-calculus.
We will introduce here only those aspects of the rules and interpretation that are immediately necessary for the proof.
A reader interested in more detail may find a full introduction in \cite{poorQupitStabiliserZXtravaganza2023}.

The qudit stabiliser ZX-calculus is generated by three kinds of components: green $Z$-spiders and red $X$-spiders, which have a phase label in $\ZZ_d^2$ and may have arbitrary number of inputs and output wires, as well as Hadamard gates, which always have one input and one output:
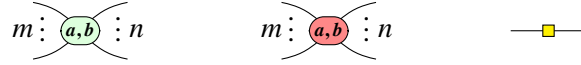
\[
\input{green-spider.tikz} \qquad\qquad
\input{red-spider.tikz} \qquad\qquad
\input{hadamard.tikz}
\]
A phase $0,0$ on a spider is usually left out.
ZX-diagrams may be composed sequentially by connecting the input wires of one diagram with the output wires of another, and they may be composed in parallel via vertical juxtaposition.
Recall that $\tau$ denotes the $d$-th root of unity satisfying $\tau^2=\omega$, and that $\ket{k:U}$ for some $k\in\FF_d$ denotes the eigenstate of the unitary operator $U$ with eigenvalue $\omega^k$.
Then the components are interpreted as follows \cite{poorQupitStabiliserZXtravaganza2023}:
\begin{align*}
	\intf{\input{green-spider.tikz}} &= d^{\frac{n+m-2}{4}} \sum_{k\in\ZZ_d} \tau^{a k + b k^2} \ket{k:Z}\t{n}\bra{k:Z}\t{m} \\
	\intf{\input{red-spider.tikz}} &= d^{\frac{n+m-2}{4}} \sum_{k\in\ZZ_d} \tau^{a k + b k^2} \ket{-k:X}\t{n}\bra{k:X}\t{m} \\
	\intf{\input{hadamard.tikz}} &= \sum_{k\in\ZZ_d} \ketbra{k:Z}{k:X} \\
	\intf{\input{identity.tikz}} &= \sum_{k\in\ZZ_d} \ketbra{k:Z}{k:Z} \\
	\intf{\input{swap.tikz}} &= \sum_{k,\ell\in\ZZ_d} \ketbra{k,\ell:Z}{\ell,k:Z}
\end{align*}

Every ZX-calculus diagram can be treated as a labelled graph in the sense that, as long as the input and output wires of the entire diagram remain in the same order, only the connectivity of the components matters; this is made rigorous via category theory.

The following two syntactic sugars will be useful.
Firstly, denote by \input{h-minus.tikz} the inverse of the Hadamard operation so that
\ctikzfig{hada-inverse}
Secondly, introduce the \emph{H-box}: a Hadamard with a label $a\in\FF_d$ called its \emph{weight}.
This represents $a$ parallel Hadamard edges between a pair of $Z$-spiders:
\[
\input{hada-n.tikz} \quad:=\quad \input{hada-n-def.tikz}
\]

\begin{observation}\label{obs:graph-state-zx}
	An $\FF_d$-weighted graph state $\ket{G}$ (cf.\ Definition~\ref{def:graph-state}) is represented in the ZX-calculus as follows:
	\begin{itemize}
		\item For each graph vertex, there is a green spider of phase $(0,0)$ which is connected to one output wire.
		\item For each edge of weight $k\in\ZZ_d$, there is an H-box of weight $k$ connecting the spiders corresponding to the two endpoints of the edge.
	\end{itemize}
\end{observation}
For example, the $\FF_d$-weighted graph state underlying the running example of Figure~\ref{Fig:LOFG} corresponds to the following ZX-diagram:
\ctikzfig{zx-graph-state-ex}

\begin{proposition}[{\cite[Proposition~12]{boothCompleteZXcalculiStabiliser2022}}]\label{prop:zx-lc}
	For any $\FF_d$-weighted graph $G$ identified with its adjacency matrix $G\in\FF_d^{V\times V}$, for any $\gamma\in\FF_d$, and for any $w\in V = \{1,\ldots,n\}$:
	\ctikzfig{zx-lc}
	where the labelled white boxes denote the graph state diagrams described in Observation~\ref{obs:graph-state-zx} and, on the left-hand side, the $X$-spiders are connected to vertex $w$.
\end{proposition}

\begin{observation}\label{obs:zx-lc}
	The ZX-diagram equality of Proposition~\ref{prop:zx-lc} implies:
	\[
		\ket{G\overset{\gamma}{\star} w} = \tilde{\phasegate}^\gamma_w \prod_{v\in V} \phasegate^{-\gamma G_{w,v}^2}_v \ket{G}
	\]
	since the matrix $G$ is symmetric and the two-legged spiders satisfy:
	\begin{align*}
		\intf{\input{red-phase.tikz}}
		&= \left(\sum_{k\in\FF_d} \tau^{0} \ketbra{-k:X}{k:X}\right)\left(\sum_{k\in\FF_d} \tau^{\gamma k^2} \ketbra{-k:X}{k:X}\right)
		= \sum_{k\in\FF_d} \tau^{\gamma k^2} \ketbra{k:X}{k:X} = \tilde{\phasegate}^{\gamma} \\
		\intf{\input{phase-G.tikz}} &= \sum_{k\in\FF_d} \tau^{-\gamma G_{v,w}^2 k^2} \proj{k} = \phasegate^{-\gamma G_{v,w}^2}
	\end{align*}
\end{observation}

To prove an analogous result for pivots, we will need the following rewrite rules of the qudit stabiliser ZX-calculus.

\begin{proposition}[{\cite[Proposition~6]{poorQupitStabiliserZXtravaganza2023}}]\label{prop:h-box-merge}
	H-boxes of weight 0 and 1 satisfy \input{lemma-h-box0.tikz} and \input{lemma-h-box1.tikz}.
	Additionally, H-boxes compose as follows:
	\ctikzfig{lemma-h-box-merge}
	where in the first equality $y$ has to be non-zero, but otherwise $x,y\in\ZZ_d$ are arbitrary.
\end{proposition}

\begin{proposition}[{\cite[Proposition~8 of arXiv version]{boothCompleteZXcalculiStabiliser2022}}]\label{prop:h-box-basics}
	The inverse Hadamard is equal to the H-box of weight $-1$, for any $x\in\ZZ_d^*$ the H-boxes of weight $x$ and weight $-x$ are inverse to each other, and repeating a non-zero--weighted H-box three times yields its inverse:\footnote{The definition of the ZX-calculus in \cite{boothCompleteZXcalculiStabiliser2022} differs slightly from the `well-tempered' normalisation used in \cite{poorQupitStabiliserZXtravaganza2023} and given above, but those differences are not relevant to the results used here.}
	\ctikzfig{hada-n-inverse}
\end{proposition}

\begin{corollary}\label{cor:double-hada-mult}
	Following from the two propositions above, for any $x\in\ZZ_d$, we have:
	\ctikzfig{cor-double-hada}
\end{corollary}

\begin{lemma}[{\cite[Lemma~31]{poorQupitStabiliserZXtravaganza2023}}]\label{lem:bialgebra}
	The bialgebra law holds for any $m,n\in\mathbb{N}$ such that $m,n\geq 2$:
	\ctikzfig{bialgebra}
	where on the right-hand side each of the $m$ $X$-spiders is connected to each of the $n$ $Z$-spiders.
\end{lemma}

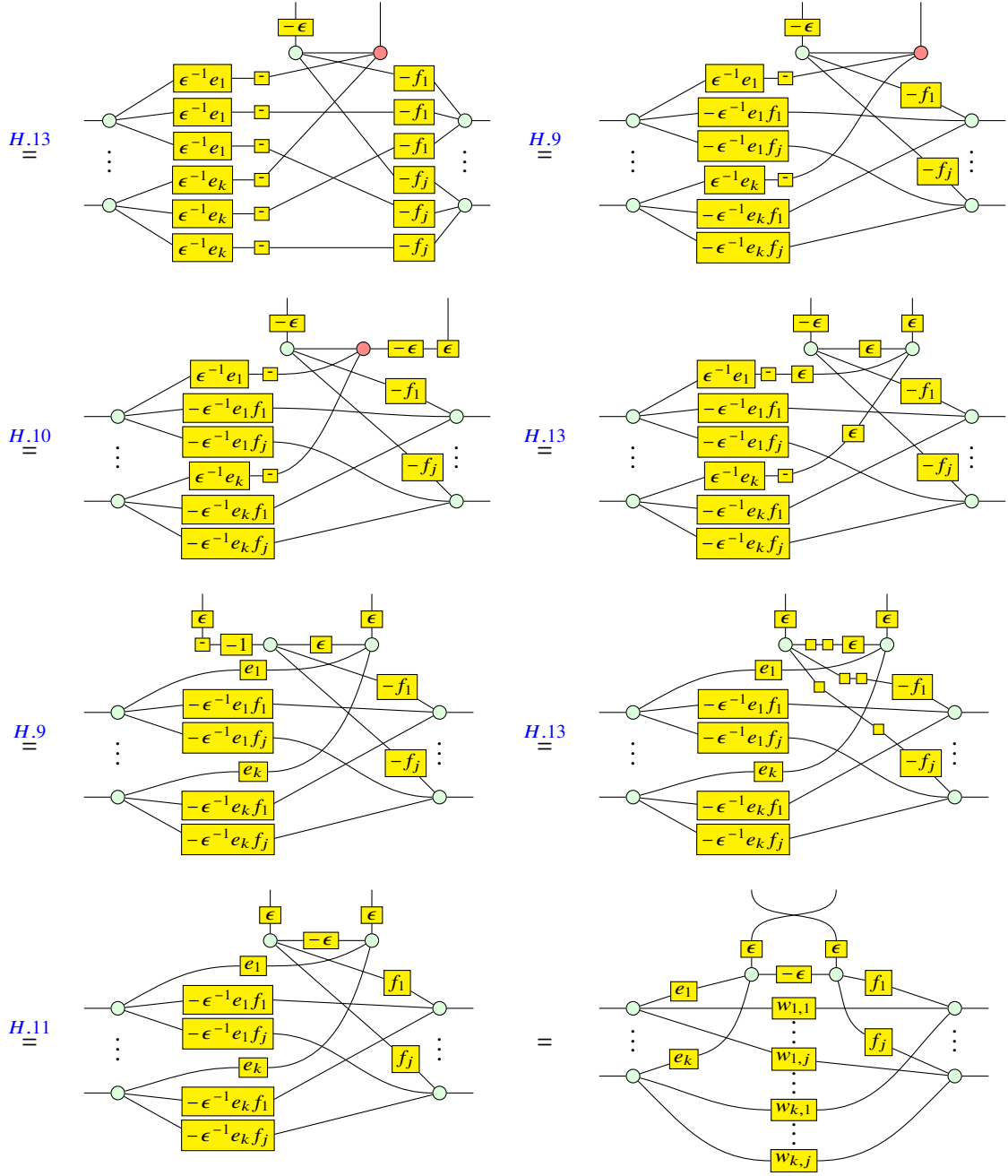
\begin{lemma}[{\cite[Lemma~53]{poorQupitStabiliserZXtravaganza2023}}]\label{lem:H-box-push}
	H-boxes can be `pushed' through spiders for any $a,b\in\ZZ_d$ and $x\in\ZZ_d^*$:
	\ctikzfig{lemma-h-box-push}
\end{lemma}

Having introduced the results required to prove the pivoting transformation, we are now ready to begin.
We follow the proof of \cite[Lemma~10]{poorQupitStabiliserZXtravaganza2023} with the only difference being that we have dangling output wires (and trivial phases) on the two endpoints of the pivot edge.
Nevertheless, we reproduce the proof here as we are interested in the resulting local unitaries on the two endpoints of the pivot edge, which are not immediately obvious from the pivot-and-delete version of the proof.

\begin{lemma}\label{lem:pivot-helper}
	The following special case of the qudit pivot rule holds:
	\[
	\input{pivot-helper0.tikz} \quad = \quad \input{pivot-helper-final.tikz}
	\]
	where $w_{m,n} = -\epsilon^{-1} e_m f_n$ for all $m\in [k]$ and all $n\in [j]$.
\end{lemma}
\begin{proof}
    \begin{alignat*}{2}
    &\phantom{\fixedaligneqq{\phantom{\ref{lem:H-box-push}}}} \input{pivot-helper0.tikz} \quad
    &&\fixedaligneqq{\ref{lem:H-box-push}} \quad \input{pivot-helper1.tikz} \qquad\qquad
    \\[10pt]
    &\fixedaligneqq{\ref{prop:h-box-merge}} \quad \input{pivot-helper2.tikz} \quad
    &&\fixedaligneqq{\ref{lem:bialgebra}} \quad \input{pivot-helper3.tikz}
    \\[10pt]
    &\fixedaligneqq{\ref{lem:H-box-push}} \quad \input{pivot-helper4.tikz} \quad
    &&\fixedaligneqq{\ref{lem:H-box-push}} \quad \input{pivot-helper5.tikz}
    \\[10pt]
    &\fixedaligneqq{\ref{lem:H-box-push}} \quad \input{pivot-helper6.tikz} \quad
    &&\fixedaligneqq{\ref{prop:h-box-merge}} \quad \input{pivot-helper7.tikz}
    \\[10pt]
    &\fixedaligneqq{\ref{prop:h-box-basics}} \quad \input{pivot-helper8.tikz} \quad
    &&\fixedaligneqq{\ref{lem:H-box-push}} \quad \input{pivot-helper9.tikz}
    \\[10pt]
    &\fixedaligneqq{\ref{prop:h-box-merge}} \quad \input{pivot-helper10.tikz} \quad
    &&\fixedaligneqq{\ref{lem:H-box-push}} \quad \input{pivot-helper11.tikz}
    \\[10pt]
    &\fixedaligneqq{\ref{cor:double-hada-mult}} \quad \input{pivot-helper12.tikz} \quad
    &&\fixedaligneqq{\phantom{\ref{lem:H-box-push}}} \quad \input{pivot-helper-final.tikz}
    \end{alignat*}
\end{proof}

\begin{lemma}\label{lem:qudit-pivot}
	The following qudit pivot rule holds:
	\[
		\input{pivot-start.tikz} \quad = \quad \input{pivot-end.tikz}
	\]
	where $\alpha_m = -2 \epsilon^{-1} e_m f_m$ and $w_{m,n} = -\epsilon^{-1} (e_m f_n + e_n f_m)$ for all $m,n\in [k]$.
\end{lemma}
\begin{proof}
	This is analogous to the proof of~\cite[Lemma~10]{poorQupitStabiliserZXtravaganza2023}, using our Lemma~\ref{lem:pivot-helper} instead of their Lemma~9.
\end{proof}

We can now show which local unitary connects two graph states related by a pivot operation.

\begin{proof}[Proof of Proposition~\ref{prop:pivot}]
	Express the graph states in the ZX-calculus as in Observation~\ref{obs:graph-state-zx} and note that we may treat non-existent edges as H-boxes of weight 0 (\cf Proposition~\ref{prop:h-box-merge}).
	Apply Lemma~\ref{lem:qudit-pivot} (bending the inputs into outputs) and use the standard interpretation to find
	\[
		\ket{G} = \left( \intf{\input{h-box-epsilon.tikz}}_w \otimes \intf{\input{h-box-epsilon.tikz}}_{w'} \otimes \bigotimes_{v\in V\setminus\{w,w'\}} \intf{\input{phase.tikz}}_v \right) \ket{G\wedge w w'}
	\]
	where $\alpha_v = - 2 \epsilon^{-1} G_{v,w} G_{v,w'}$.
	The interpretation of the graph underlying the right-hand side of Lemma~\ref{lem:qudit-pivot} as $\ket{G\wedge w w'}$ is correct by inspection.	
	
	Now, it is straightforward to verify that
	\begin{align*}
		\intf{\input{h-box-epsilon.tikz}} &= \sum_{j,k\in\FF_d} \omega^{\epsilon j k } \ketbra{j}{k} = F^{(\epsilon)} \\
		\intf{\input{phase.tikz}} &= \sum_{k\in\FF_d} \tau^{\alpha_v k^2} \proj{k} = \phasegate^{\alpha_v}
	\end{align*}
	Note that $(F^{(\gamma)})^{-1} = F^{(-\gamma)}$ for any $\gamma\in\FF_d^*$.
	Hence the desired result follows by moving the single-qudit operators to the other side of the equation, noting that we may allow $\phasegate$ operators on $w$ and $w'$ as $\alpha_w = \alpha_{w'} = 0$, and that the single-qudit operators all commute, meaning $\bigotimes$ can be changed to $\prod$.
\end{proof}

\subsection{Proving flow reversibility}\label{Sec:proof reversibility}

Before proving Theorem~\ref{Th:flow-rev}, we provide a build-up with various supporting observations, like we have done in our proof for the qubit case \cite[Subsection 3.3]{mitosekAlgebraicInterpretationPauli2026}.

We start by naming relevant types of measurements:

\begin{definition}
    We distinguish two types of measurements denoted `X-like measurements' $\Xlike$ and `Z-like measurements' $\Zlike$ as follows (as a reminder, $B$ is the set of vertices internal vertices, \cf Definition~\ref{Def:OFG}):\begin{align*}
        \Xlike &:= \{ v \in B \mid \ldX(v) = 0 \}\\
        \Zlike &:= \{ v \in B \mid \ldX(v) \ne 0 \}
    \end{align*}
\end{definition}

We exclude each vertex that is an input and an output simultaneously:

\begin{observation}
    Vertices that are both inputs and outputs do not impact the existence of \flow at all, neither in $\Gamma$ nor in its reverse $\Gamma'$.
    Thus, from now on, we assume $I \cap O = \emptyset$ and $V = I \cup \Xlike \cup \Zlike \cup O$.
\end{observation}

Given a \LOFG and some vertex $v \in \comp{O}$, \flow is preserved if $\ld(v)$ is replaced by its scalar multiple for any non-zero scalar in $\Fd^*$ \cite[Section 3]{boothOutcomeDeterminismMeasurementbased2023}.
For example, in the \LOFG from Figure~\ref{Fig:running example}, vertex $d$ could instead be given the measurement label $3 \cdot (4,2) = (2,1)$ and this would preserve the existence of \flow.

Now, for any measurement label $(a,b)$, there exists $c\in\FF_d$ such that $(ac,bc)$ is of the form $(0,1)$ or $(1,k)$ for some $k \in \Fd$: if $a  = 0$, then necessarily $b \ne 0$, so taking $c = b^{-1}$ works; and if $a \ne 0$, taking $c = a^{-1}$ works.
Therefore, without loss of generality, one may work with normalised {\LOFG\unskip}s in the following sense:

\begin{definition}\label{Def:normalised}
    A \LOFG is called \emph{normalised} if all measurements labels are of the form $(0,1)$ or $(1,k)$ for some $k \in \Fd$.
\end{definition}

For example, normalisation of the \LOFG from Figure~\ref{Fig:running example} results in the \LOFG in Figure~\ref{fig:normalisation}.

\begin{figure}
    \centering
    $\scalediag[1]{ex1norm}$
    \caption{The \LOFG from Figure~\ref{Fig:running example} after normalisation of measurement labels.}
    \label{fig:normalisation}
\end{figure}

After the normalisation of measurement labels, we will now introduce some notation related to the graph part of a \LOFG.
Since we identify the graph with its adjacency matrix, let us name relevant submatrices of $G$ as follows, where $*$ stands for an unnamed submatrix:\begin{equation*}
    G =: \scalemath{0.85}{\begin{pNiceArray}{ccc}[first-row, last-col]
        I & \Xlike \cup O & \Zlike & \\
        * & G_{00} & G_{01} & I \cup \Xlike \\
        * & * & * & O \\
        * & G_{10} & G_{11} & \Zlike
    \end{pNiceArray}}
\end{equation*}

Now, we can rewrite the flow-demand matrix in block form:

\begin{observation}\label{Obs:fdm-block}
    Let $\Gamma = (G,I,O,\ld)$ be a normalised \LOFG which satisfies $|I|=|O|$ and $I \cap O = \emptyset$; and let $\Gamma'$ be its reverse (cf.\ Definition~\ref{def:reverse-LOFG}).
    Then the flow-demand matrix $M$ of $\Gamma$ and the flow-demand matrix $M'$ of $\Gamma'$ satisfy:\begin{equation*}
        M = \scalemath{0.85}{\begin{pNiceArray}{cc}[first-row, last-col]
            \Xlike \cup O & \Zlike & \\
            G_{00} & G_{01} & I \cup \Xlike \\
            0 & \Id & \Zlike
        \end{pNiceArray}}
        \quad \text{and} \quad
        M' = \scalemath{0.85}{\begin{pNiceArray}{cc}[first-row, last-col]
            I \cup \Xlike & Z & \\
            (G_{00})^T & (G_{01})^T & \Xlike \cup O \\
            0 & \Id & \Zlike
        \end{pNiceArray}} = \scalemath{0.85}{\begin{pNiceArray}{cc}
            G_{00} & 0 \\
            G_{10} & \Id
        \end{pNiceArray}}^T
    \end{equation*}
\end{observation}

Similarly, we can present the order-demand matrix in block form: 

\begin{observation}\label{Obs:odm-block}
    Let $\Gamma = (G,I,O,\ld)$ be a normalised \LOFG which satisfies $|I|=|O|$ and $I \cap O = \emptyset$; and let $\Gamma'$ be its reverse.
    In the order-demand matrix $N$ of $\Gamma$, rows corresponding to vertices in $I$ and in $\Xlike$ are defined analogously; yet this results in the rows corresponding to vertices in $I$ being identically zero.
    We can write the order-demand matrix as a product $N = PS$, where $P$ is a diagonal $\comp{O} \times \comp{O}$ matrix identifying the internal vertices, \ie:\begin{equation*}
        P := \diag{1 \text{ for } v \in B \text{ and } 0 \text{ for } v \in I} \quad \text{and} \quad S := \scalemath{0.85}{\begin{pNiceArray}{cc}[first-row, last-col]
            \Xlike \cup O & \Zlike & \\
            J & 0 & I \cup \Xlike \\
            G_{10} & G_{11} - D & \Zlike
        \end{pNiceArray}}
    \end{equation*}
    \noindent with $J_{vw} = \delta_{vw}$ and $D$ being a diagonal matrix expressing $\ldZ(v)$, \ie $D = \diag{\ldZ(v) \text{ for } v \in \Zlike}$.
    From now on, we write $K_{11} := G_{11} - D$.
    Note that $K_{11}$ is symmetric because $G_{11}$ is a diagonal block of the (symmetric) adjacency matrix and $D$ is diagonal.
    Moreover, while $J$ is square, it is not generally the identity matrix because the sets of row and column labels need not match.
    
    Similarly to $N$, the order-demand matrix $N'$ of $\Gamma'$ can be written as a product $N' = P'S'$ where $P'$ is a diagonal $\comp{I}\times\comp{I}$ matrix:\begin{equation*}
        P' := \diag{1 \text{ for } v \in B \text{ and } 0 \text{ for } v \in O} \quad \text{and} \quad S' := \scalemath{0.85}{\begin{pNiceArray}{cc}[first-row, last-col]
            I \cup \Xlike & \Zlike & \\
            J^T & 0 & \Xlike \cup O \\
            (G_{10})^T & K_{11} & \Zlike
        \end{pNiceArray}} = \scalemath{0.85}{\begin{pNiceArray}{cc}
            J & G_{01} \\
            0 & K_{11}
        \end{pNiceArray}}^T
    \end{equation*}
    \noindent where the last equality follows from the symmetry of $K_{11}$.
\end{observation}

We are now ready to prove Theorem~\ref{Th:flow-rev}:

\begin{proof}[Proof of Theorem~\ref{Th:flow-rev}]
    Without loss of generality, we assume that $\Gamma$ is normalised and that $I \cap O = \emptyset$.
    
    $(\Rightarrow):$ assume $\Gamma$ has \flow.
    By Theorem~\ref{Th:main result}, there exists $C$ such that $MC = \Id$ and $NC$ forms a DAG, where $M$ and $N$ are the flow-demand and order-demand matrices of $\Gamma$, respectively.
    By Observation~\ref{Obs:fdm-block}:\begin{equation*}
        M = \scalemath{0.85}{\begin{pNiceArray}{cc}[first-row, last-col]
            \Xlike \cup O & \Zlike & \\
            G_{00} & G_{01} & I \cup \Xlike \\
            0 & \Id & \Zlike
        \end{pNiceArray}}
    \end{equation*}

    Since $|I|=|O|$, $M$ is square and thus, by Corollary~\ref{Cor:uniqueness}, $C$ is unique and $C = M^{-1}$.
    Break $C$ into blocks corresponding to the block structure of $M$:\begin{equation*}
        C = \scalemath{0.85}{\begin{pNiceArray}{cc}[first-row, last-col]
            I \cup \Xlike & \Zlike & \\
            C_{00} & C_{01} & \Xlike \cup O \\
            C_{10} & C_{11} & \Zlike
        \end{pNiceArray}}
    \end{equation*}
    \noindent Then:\begin{equation*}
        MC = \scalemath{0.85}{\begin{pNiceArray}{cc}[first-row, last-col]
            I \cup \Xlike & \Zlike & \\
            G_{00}C_{00} + G_{01}C_{10} & G_{00}C_{01} + G_{01}C_{11} & I \cup \Xlike \\
            C_{10} & C_{11} & \Zlike
        \end{pNiceArray}}
    \end{equation*}
    \noindent Therefore $C_{10} = 0$ and $C_{11} = \Id$ by inspection.
    Using these, the top left block implies $\Id = G_{00}C_{00}$, \ie $C_{00} = (G_{00})^{-1}$ (note that $G_{00}$ is square).
    Similarly, the top right block becomes $0 = G_{00}C_{01} + G_{01}$.
    Hence $G_{00}C_{01} = -G_{01}$ and finally $C_{01} = -(G_{00})^{-1}G_{01} = -C_{00}G_{01}$.
    To summarize:\begin{equation*}
        C = \scalemath{0.85}{\begin{pNiceArray}{cc}[first-row, last-col]
            I \cup \Xlike & \Zlike & \\
            C_{00} & -C_{00}G_{01} & \Xlike \cup O \\
            0 & \Id & \Zlike
        \end{pNiceArray}}
    \end{equation*}
    Again by Observation~\ref{Obs:fdm-block}, the flow-demand matrix of $\Gamma'$ satisfies:\begin{equation*}
        M' = \scalemath{0.85}{\begin{pNiceArray}{cc}
            G_{00} & 0  \\
            G_{10} & \Id
        \end{pNiceArray}}^T
    \end{equation*}
    \noindent As $G_{00}$ is invertible (with inverse $C_{00}$) and because of its block structure, $M'$ is also invertible. By an argument analogous to that for $M$, we find that:\begin{equation*}
        C' := (M')^{-1} = \scalemath{0.85}{\begin{pNiceArray}{cc}
            C_{00} & 0 \\
            -G_{10}C_{00} & \Id
        \end{pNiceArray}}^T
    \end{equation*}
    To show that $\Gamma'$ has \flow, by Theorem~\ref{Th:main result}, it remains to show that $N'C'$ forms a DAG, where $N'$ is the order-demand matrix of $\Gamma'$.
    First, by Observation~\ref{Obs:odm-block}:\begin{equation*}
        NC = PSC = P\scalemath{0.85}{\begin{pNiceArray}{cc}
            J & 0 \\
            G_{10} & K_{11}
        \end{pNiceArray}} \scalemath{0.85}{\begin{pNiceArray}{cc}
            C_{00} & -C_{00}G_{01} \\
            0 & \Id
        \end{pNiceArray}} = P\scalemath{0.85}{\begin{pNiceArray}{cc}
            JC_{00} & -JC_{00}G_{01} \\
            G_{10}C_{00} & -G_{10}C_{00}G_{01} + K_{11}
        \end{pNiceArray}}
    \end{equation*}
    \noindent and similarly:\begin{align*}
        N'C' &= P'S'C' = P'\scalemath{0.85}{\begin{pNiceArray}{cc}
            J & G_{01} \\
            0 & K_{11}
        \end{pNiceArray}}^T \scalemath{0.85}{\begin{pNiceArray}{cc}
            C_{00} & 0 \\
            -G_{10}C_{00} & \Id
        \end{pNiceArray}}^T = P'\left( \scalemath{0.85}{\begin{pNiceArray}{cc}
            C_{00} & 0 \\
            -G_{10}C_{00} & \Id
        \end{pNiceArray}} \scalemath{0.85}{\begin{pNiceArray}{cc}
            J & G_{01} \\
            0 & K_{11}
        \end{pNiceArray}} \right)^T\\
        &= P'\scalemath{0.85}{\begin{pNiceArray}{cc}
            C_{00}J & C_{00}G_{01} \\
            -G_{10}C_{00}J & -G_{10}C_{00}G_{01} + K_{11}
        \end{pNiceArray}}^T
    \end{align*}

    Note that $NC$ is a $\comp{O}\times\comp{O}$ matrix satisfying $(NC)_{vw} = 0$ whenever $v \in I$ and thus \emph{rows} labelled by $I$ are irrelevant to determining whether $NC$ is a DAG.
    This means we can also exclude \emph{columns} labelled by $I$, as all inputs will be initial in the partial order and thus may have arbitrary successors among the non-inputs.
    Hence it suffices to consider only the submatrix of $NC$ whose rows and columns are both labelled by $B$.
    An analogous argument holds for $N'C'$, where it is rows and columns labelled by $O$ that are excluded to leave only the submatrix whose rows and columns are both labelled by $B$.

    Now, the three matrices $P$, $P'$, and $J$ each reduce to an identity matrix when restricted to rows and columns labelled by $B$.
    In the following, we denote a submatrix obtained by restricting rows and columns to vertices in $B$ by adding a subscript and superscript next to the matrix.
    Then the $B \times B$ submatrices satisfy:\begin{align*}
        (N'C')_B^B &= \left( \scalemath{0.85}{\begin{pNiceArray}{cc}
            C_{00}J & C_{00}G_{01} \\
            -G_{10}C_{00}J & -G_{10}C_{00}G_{01} + K_{11}
        \end{pNiceArray}}^T \right)_B^B = \left( \scalemath{0.85}{\begin{pNiceArray}{cc}
            C_{00} & C_{00}G_{01} \\
            -G_{10}C_{00} & -G_{10}C_{00}G_{01} + K_{11}
        \end{pNiceArray}}_B^B \right)^T
    \intertext{since taking submatrix commutes with taking the transpose; and:}
        (NC)_B^B &= \scalemath{0.85}{\begin{pNiceArray}{cc}
            JC_{00} & -JC_{00}G_{01} \\
            G_{10}C_{00} & -G_{10}C_{00}G_{01} + K_{11}
        \end{pNiceArray}}_B^B = \scalemath{0.85}{\begin{pNiceArray}{cc}
            C_{00} & -C_{00}G_{01} \\
            G_{10}C_{00} & -G_{10}C_{00}G_{01} + K_{11}
        \end{pNiceArray}}_B^B
    \end{align*}
    \noindent Finally, observe that the support of the transpose of the first matrix and the support of the second matrix are identical, as the block expressions only differ by signs, which do not impact the support.
    By the assumption of $C$ specifying an \flow on $\Gamma$, $NC$ is a DAG.
    Since only the support is relevant when considering whether a matrix encodes a DAG, $N'C'$ is also a DAG.
    Thus, by Theorem~\ref{Th:main result}, $\Gamma'$ has \flow, ending this part of the proof.
    
    $(\Leftarrow):$ observe that, due to the normalisation of measurement labels, the reversal operation of Definition~\ref{def:reverse-LOFG} is an involution, \ie $(\Gamma')' = \Gamma$.
    Hence the desired result follows from the $(\Rightarrow)$ part of the proof.
\end{proof}

\subsection{Proving that \texorpdfstring{$Z$}{Z}-like removal and insertion are \flowintitles-preserving}\label{Sec:proving removal and insertion}

\begin{proof}[Proof of Theorem~\ref{Th:Z-like removal}]
    Let $M$ and $N$ be the flow-demand and order-demand matrices for $(G,I,O,\ld)$ respectively and let $M'$ and $N'$ be flow-demand and order-demand matrices for $(G',I,O,\ld')$ respectively.
    Furthermore, let $\comp{O}'$ and $\comp{I}'$ be the set of non-output vertices and the set of non-input vertices of $(G',I,O,\ld')$.
    
    Let $C'$ be the correction matrix for $(G',I,O,\ld')$, whose existence follows by Theorem~\ref{Th:main result} from the assumption that $\Gamma'$ has \flow.
    Thus $M'C' = \id{\comp{O}'}{\comp{O}'}$ and $N'C'$ forms a DAG.
    Since $\ldX'(v) \ne 0$, necessarily $v \notin I$ and thus both a $v$-labelled row and a $v$-labelled column exist in each of $M'$, $N'$, $C'$, and $N'C'$.
    
    By construction (\cf Definitions~\ref{Def:fdm} and~\ref{Def:odm}), $M$ and $N$ are submatrices of $M'$ and $N'$.
    Let $C$ be the restriction of $C'$ to shape $\comp{O} \times \comp{I}$, the shape compatible with $M$ and $N$.
    We now show that $MC = \id{\comp{O}}{\comp{O}}$ and that $NC$ forms a DAG.
    
    The \flow guaranteed by Theorem~\ref{Th:main result} is focused, so we can apply condition~\conref{FZ} to $C'$ to find $C'_{v,w} = 0$ for any $w\in\comp{O}'\setminus\{v\}$.
    Therefore, for $u, w \in \comp{O}$:\begin{equation*}
        (M'C')_{u,w} = \sum_{t \in \comp{I}'} M'_{u,t} C'_{t,w} = \left( \sum_{t \in \comp{I}} M'_{u,t} C'_{t,w} \right) + M'_{u,v} C'_{v,w} = \left( \sum_{t \in \comp{I}} M_{u,t} C_{t,w} \right) + 0 = (MC)_{u,w}
    \end{equation*}
    Thus $MC$ is the $\comp{O} \times \comp{O}$ submatrix of $M'C'$, which implies $MC = \id{\comp{O}}{\comp{O}}$.
    By replacing $M'$ with $N'$ in the above equation, $NC$ is the $\comp{O} \times \comp{O}$ submatrix of $N'C'$; and since $N'C'$ forms a DAG, so must $NC$.
    Hence, by Theorem~\ref{Th:main result}, $(G,I,O,\ld)$ has \flow.
\end{proof}

\begin{proof}[Proof of Theorem~\ref{Th:Z-like insertion}]
    Let $M'$ and $N'$ be the flow-demand and order-demand matrices of $\Gamma'$, and analogously define $M$ and $N$ for $\Gamma$.
    Suppose that $(G',I,O,\ld')$ has \flow.
    Our goal is to find sufficient and necessary conditions that $\ld'$ must satisfy and which can be checked in $\bigO(n^2)$.
    
    Let $\comp{O}', \comp{I}'$ be the non-outputs and non-inputs of $\Gamma'$, respectively, \ie $\comp{O}' = \comp{O} \cup \{ v \}$ and analogously for $\comp{I}'$.
    By the constructions of flow-demand and order-demand matrices (\ie Definitions~\ref{Def:fdm} and~\ref{Def:odm}), we have:\begin{equation*}
        M' = \scalemath{0.85}{\begin{pNiceArray}{cc}[first-row, first-col]
              & \comp{I} & v \\
            \comp{O} & M & \mathbf{x} \\
            v & \mathbf{0} & \ldX'(v)^{-1} \\
        \end{pNiceArray}} \qquad \text{and} \qquad N' = \scalemath{0.85}{\begin{pNiceArray}{cc}[first-row, first-col]
              & \comp{I} & v \\
            \comp{O} & N & \mathbf{z} \\
            v & \mathbf{n} & -\ldZ'(v) \\
        \end{pNiceArray}}
    \end{equation*}
    \noindent where $\mathbf{0}$ is identically zero.
    The remaining three submatrices $\mathbf{x}, \mathbf{z} \in (\Fd^{\comp{O}})$ and  $\mathbf{n} \in (\Fd^{\comp{I}})^T$ are vectors defined as follows for any $u \in \comp{O}$:\begin{align}
        \mathbf{x}_u &= \begin{cases}
            \ldZ(u)^{-1} S_{u} & \text{when}\ \ldX(u) = 0\\
            0 & \text{otherwise}
        \end{cases}\notag \\
        \mathbf{z}_u &= \ldX(u) S_{u}\notag \\
        \mathbf{n} &= \ldX'(v) (S\mid_{\comp{I}})^T \label{eq:ncond}
    \end{align}
    \noindent In particular, $\mathbf{x}$ and $\mathbf{z}$ do not depend on $\ld'(v)$ at all, while $\mathbf{n}$ is scaled by $\ldX'(v)$ but otherwise is also independent of $\ld'(v)$.

    According to Theorem~\ref{Th:main result}, there exists $C'$ such that $M'C' = \id{\comp{O}'}{\comp{O}'}$ and $N'C'$ forms a DAG.
    Since $|I|=|O|$, by Corollary~\ref{Cor:uniqueness}, this $C'$ is unique.
    Similarly, $C$ is also unique, and by utilising the proof of Theorem~\ref{Th:Z-like removal} which describes the construction of a valid $C$ as a submatrix of $C'$, necessarily $C$ forms a submatrix of $C'$.
    Furthermore, since Theorem~\ref{Th:main result} guarantees focused \flow, by \conref{FZ} we obtain $C'_{v,u} = 0$ for $u \ne v$ and by Definition~\ref{Def:flow} we obtain $C'_{v,v} = \ldX'(v)$.
    Thus:\begin{equation}
        C' = \scalemath{0.85}{\begin{pNiceArray}{cc}[first-row, first-col]
              & \comp{O} & v \\
            \comp{I} & C & \mathbf{k} \\
            v & 0 & \ldX'(v) \\
        \end{pNiceArray}} \label{eq:C'}
    \end{equation}
    \noindent for some $\mathbf{k} \in \Fd^{\comp{I}}$.

    Now, we have: \begin{alignat}{10}
        M'C' &= \scalemath{0.85}{\begin{pNiceArray}{cc}[first-row, first-col]
              & \comp{I} & v \\
            \comp{O} & M & \mathbf{x} \\
            v & \mathbf{0} & \ldX'(v)^{-1} \\
        \end{pNiceArray}} \;\, &{} \scalemath{0.85}{\begin{pNiceArray}{cc}[first-row, last-col]
            \comp{O} & v \\
            C & \mathbf{k} & \comp{I} \\
            \mathbf{0} & \ldX'(v) & v \\
        \end{pNiceArray}} \;\;&=&{} \scalemath{0.85}{\begin{pNiceArray}{cc}[first-row, first-col]
              & \comp{O} & v \\
            \comp{O} & MC & M\mathbf{k} + \mathbf{x}\ldX'(v)\\
            v & \mathbf{0} & 1 \\
        \end{pNiceArray}}\;\; \label{eq:M'C'}
    \intertext{and:}
        N'C' &= \scalemath{0.85}{\begin{pNiceArray}{cc}[first-row, first-col]
              & \comp{I} & v \\
            \comp{O} & N & \mathbf{z} \\
            v & \mathbf{n} & -\ldZ'(v) \\
        \end{pNiceArray}} \;\, &{} \scalemath{0.85}{\begin{pNiceArray}{cc}[first-row, last-col]
            \comp{O} & v \\
            C & \mathbf{k} & \comp{I} \\
            \mathbf{0} & \ldX'(v) & v \\
        \end{pNiceArray}} \;\;&=&{}\;\; \scalemath{0.85}{\begin{pNiceArray}{cc}[first-row, first-col]
              & \comp{O} & v \\
            \comp{O} & NC & N\mathbf{k} + \mathbf{z}\ldX'(v)\\
            v & \mathbf{n}C & \mathbf{n}\mathbf{k} - \ldX'(v) \ldZ'(v) \\
        \end{pNiceArray}} \label{eq:N'C'}
    \end{alignat}
    \noindent We want $M'C'=\id{\comp{O}'}{\comp{O}'}$ and $N'C'$ to be a DAG, so we must necessarily have:\begin{subequations}\begin{align}
        M\mathbf{k} + \mathbf{x}\ldX'(v) &= \mathbf{0}\label{eq:Mcond}\\
        \mathbf{n}\mathbf{k} - \ldX'(v)\ldZ'(v) &= 0 \label{eq:Ncond}
    \end{align}\end{subequations}
    Since $M$ is square and $C$ is its right inverse, $C$ is the unique inverse of $M$, hence also $CM = \id{\comp{I}'}{\comp{I}'}$ (note that $|\comp{O}'| = |\comp{I}'|$).
    Multiplying both sides of \eqref{eq:Mcond} by $C$ on the left, we find:\begin{align*}
        CM\mathbf{k} + C\mathbf{x}\ldX'(v) &= \mathbf{0}
    \intertext{which simplifies to:}
        \mathbf{k} + C\mathbf{x}\ldX'(v) &= \mathbf{0}
    \end{align*}
    Therefore:\begin{equation}
        \mathbf{k} = -\ldX'(v) C\mathbf{x}\label{eq:kcond}
    \end{equation}
    Plugging the above into \eqref{eq:Ncond}, we find:\begin{equation*}
        -\ldX'(v)\mathbf{n}C\mathbf{x} -\ldX'(v)\ldZ'(v) = 0
    \end{equation*}
    \noindent and since $\ldX'(v) \ne 0$ and by \eqref{eq:ncond} we find:\begin{equation}\label{eq:ldZ}
        \ldZ'(v) = -\mathbf{n}C\mathbf{x} = -\ldX'(v) (S\mid_{\comp{I}})^T C \mathbf{x}
    \end{equation}

    Next, we focus on not-yet considered parts of the product of $N'C'$ from \eqref{eq:N'C'}.
    We have by \eqref{eq:ncond}:\begin{equation*}
        \mathbf{n}C = \ldX'(v) (S\mid_{\comp{I}})^T C
    \end{equation*}
    \noindent and by substituting via \eqref{eq:kcond}:\begin{equation*}
        N\mathbf{k}+\mathbf{z}\ldX'(v) = -\ldX'(v)NC\mathbf{x} + \ldX'(v)\mathbf{z} = \ldX'(v)(\mathbf{z}-NC\mathbf{x})
    \end{equation*}
    \noindent \ie these expressions have supports that do not depend on $\ldX'(v)$.
    Thus, if $N'C'$ is a DAG for some $\ldX'(v)$, it is a DAG for all $\ldX'(v)$.
    
    Note that the requirement $M'C'=\Id$ also does not fully fix $\ld'(v)$, instead \eqref{eq:ldZ} gives a formula for $\ld_Z'(v)$ as a function of $\ld'_X(v)$.
    This is to be expected since treating a measurement label as a vector and scaling it by a non-zero number does not change whether there is an \flow, cf.\ the discussion before Definition~\ref{Def:normalised} in Section~\ref{Sec:proof reversibility}.
    Now we may treat $\ldX'(v)$ as a variable; let $a := \ldX'(v)$.
    We find that if $(G',I,O,\ld')$ has \flow, then necessarily, for any $a \in \Fd^*$:\begin{align*}
        \ldX'(v) &= a\\
        \ldZ'(v) &= -a (S\mid_{\comp{I}})^T C \mathbf{x}
    \end{align*}
    \noindent with correction matrix (from substituting into \eqref{eq:C'}):\begin{equation*}
        C' = \scalemath{0.85}{\begin{pNiceArray}{cc}[first-row, first-col]
              & \comp{O} & v \\
            \comp{I} & C & -a C \mathbf{x} \\
            v & 0 & a \\
        \end{pNiceArray}}
    \end{equation*}
    \noindent Moreover, the following matrix (from substituting into \eqref{eq:N'C'}) must be a DAG:\begin{equation*}
        \scalemath{0.85}{\begin{pNiceArray}{cc}[first-row, first-col]
              & \comp{O} & v \\
            \comp{O} & NC & a(\mathbf{z} - NC\mathbf{x}) \\
            v & a (S\mid_{\comp{I}})^T C & 0 \\
        \end{pNiceArray}}
    \end{equation*}

    It is sufficient to check whether the final matrix is a DAG for just one value of $a \in \Fd^*$, for instance $a=1$.
    All relevant vectors and the above matrix can be computed in $\bigO(n^2)$ (note that $NC$ is known so it does not need to be computed from scratch).
    Checking whether a matrix corresponds to a DAG can also be done in $\bigO(n^2)$ \cite[Section 20.4]{cormenIntroductionAlgorithmsFourth2022}.
    Thus, in time $\bigO(n^2)$ we verify that either there are no $\ld'(v)$ resulting in \flow, or that all pairs $(a,-a (S\mid_{\comp{I}})^T C \mathbf{x})$ for any $a \in \Fd^*$ (and no other pairs) result in \flow.
\end{proof}

\begin{remark}
    The proof above also implies that, for any $\ld'(v)$ with $\ldX'(v) \ne 0$ such that $(G',I,O,\ld')$ has \flow, it is possible to explicitly return the corresponding pair of the correction matrix $C'$ and the DAG $N'C'$ in $\bigO(n^2)$ time, \ie we can not only find the possible labels for the added vertex, but also update the \flow accordingly.
\end{remark}

\subsection{Proofs relating to the generation of instances with \flowintitles}\label{Sec:proving generation}

Usually, the term `flow-preserving rewriting' refers to transformations that preserve both the linear operation implemented by a computation and the existence of flow.
Yet the existence of flow (or \flow) depends only on the \LOFG, which does not contain enough information to determine a specific linear operator.
For example, it was pointed out in \cite{backensThereBackAgain2021} that the deletion of a $YZ$-measurement is always flow-preserving, even though it preserves the interpretation only if the measurement angle is 0.
Thus, as argued in \cite{backensGenerating2026}, it can also be useful to analyse rewrite rules that preserve the existence of flow but not necessarily the interpretation.

We expect that the completeness proof from \cite{backensGenerating2026} can be adapted to the qudit setting, but leave that to future work, focusing solely on the generation of \LOFG{}s with \flow for now. 
First, we prove a supporting lemma which shows that \flow-preserving rewriting can be used to transform any \LOFG with \flow into a \LOFG without internal vertices.

\begin{lemma}[Deconstruction of a \LOFG with \flow]\label{Th:deconstruction}
    Given a \LOFG $\Gamma$ with \flow, a sequence of local-complementations and $Z$-like removals can bring $\Gamma$ to a \LOFG with no internal vertices. 
\end{lemma}
\begin{proof}
    Suppose $v \in B$.
    If $\ldX(v) \ne 0$, then, by Theorem~\ref{Th:Z-like removal}, $v$ can be removed while preserving \flow.
    If instead $\ldX(v) = 0$, then a local complementation about $v$ (with any non-zero parameter $\gamma$) is flow-preserving and transforms $\ldX(v)$ to some $\ldX'(v)$ with $\ldX'(v) \ne 0$, which again can be removed (\cf Proposition~\ref{prop:lc}).
    Repeating this process for all $v \in B$ while ignoring any Clifford effects on outputs eventually results in a \LOFG with $B = \emptyset$.
\end{proof}

\begin{proof}[Proof of Theorem~\ref{Th:generation}]
    Start by generating a random invertible matrix $M$ over $\Fd$, for example by following \cite{randallEfficientGenerationRandom1993}.
    Construct the \LOFG $\Xi$ with flow-demand matrix $M$ via the construction from Theorem~\ref{Th:lower bound}.
    By Theorem~\ref{Th:main result}, $\Xi$ has \flow and $\Xi$ is normalised (\cf Definition~\ref{Def:normalised}).
    Undo normalisation: for each input $i$, multiply $\ld(i)$ by some non-zero scalar from $\Fd^*$, transforming $\Xi$ into $\Gamma^{(0)}$.
    Now repeat the following $k$ times:\begin{itemize}
        \item Choose a random vector $S \in \Fd^{\comp{V}}$ and using Theorem~\ref{Th:Z-like insertion}, attempt inserting a new $Z$-like internal vertex with connections determined by $S$, trying all measurement labels of the form $(a,b)\in\Fd$ where $a\neq 0$.
        If the insertion is not \flow-preserving for any measurement label, choose a different $S$ and try again, otherwise pick any suitable label and move on to the second step.
        \item Perform a local complementation around the newly inserted vertex with a randomly chosen (and not necessarily non-zero) parameter $\gamma \in \Fd$, ignoring any resulting local Clifford operators on outputs.
    \end{itemize}
    Finally, return the resulting labelled open graph $\Gamma^{(k)}$.
    
    The above procedure always produces a \LOFG with \flow.
    In particular, the procedure has a non-zero probability of reversing all steps from Lemma~\ref{Th:deconstruction} when used on $\Gamma$.
    Thus, with non-zero probability, the above procedure returns $\Gamma$.
\end{proof}

As of now, our method is not necessarily efficient: Choosing a random vector $S \in \Fd^{\comp{V}}$ for connectivity of the new vertex may fail many times before finding a vector that does permit vertex insertion.
One can observe that an identically zero $S$ always permits insertion, thus, with probability $1$, our algorithm eventually performs a vertex insertion (albeit potentially a trivial one).
However, bounding the efficiency of the insertion in a more meaningful way requires deeper analysis of the structure of all vectors $S$ permitting insertions.
We leave this to future work.

\end{document}

%% file: thesisstyles.tikzstyles
\tikzstyle{ZX Z none}=[inner sep=0mm, minimum size=2mm, shape=circle, draw=black, fill={rgb,255: red,208; green,252; blue,204}, tikzit category=ZX]
\tikzstyle{ZX Z in}=[minimum size=4mm, font={\footnotesize\boldmath}, shape=rectangle, rounded corners=1.6mm, inner sep=0.2mm, outer sep=-2mm, scale=0.8, tikzit shape=circle, draw=black, fill={rgb,255: red,208; green,252; blue,204}, tikzit draw=blue, tikzit category=ZX]
\tikzstyle{ZX Z out}=[font={\footnotesize\boldmath}, rounded rectangle, rounded rectangle arc length=90, fill={rgb,255: red,208; green,252; blue,204}, inner sep=2pt, scale=0.8, tikzit draw=red, tikzit category=ZX]
\tikzstyle{ZX H}=[fill=yellow, shape=rectangle, inner sep=0mm, minimum size=2mm, draw=black, tikzit shape=rectangle, tikzit category=ZX]
\tikzstyle{ZX X none}=[inner sep=0mm, minimum size=2mm, shape=circle, draw=black, fill={rgb,255: red,255; green,136; blue,136}, tikzit category=ZX]
\tikzstyle{ZX X in}=[minimum size=4mm, font={\footnotesize\boldmath}, shape=rectangle, rounded corners=1.6mm, inner sep=0.2mm, outer sep=-2mm, scale=0.8, tikzit shape=circle, draw=black, fill={rgb,255: red,255; green,136; blue,136}, tikzit draw=blue, tikzit category=ZX]
\tikzstyle{ZX X out}=[font={\footnotesize\boldmath}, rounded rectangle, rounded rectangle arc length=90, fill={rgb,255: red,255; green,136; blue,136}, inner sep=2pt, scale=0.8, tikzit draw=red, tikzit category=ZX]

\tikzstyle{ZH Star}=[star, fill=black, draw=black, minimum size=2mm, inner sep=0pt, tikzit category=ZH, tikzit fill=brown]
\tikzstyle{ZH Z none}=[inner sep=0mm, minimum size=2mm, shape=circle, draw=black, fill=white, tikzit category=ZH]
\tikzstyle{ZH Z in}=[minimum size=4mm, font={\footnotesize\boldmath}, shape=rectangle, rounded corners=1.6mm, inner sep=0.2mm, outer sep=-2mm, scale=0.8, tikzit shape=circle, draw=black, fill=white, tikzit draw=blue, tikzit category=ZH]
\tikzstyle{ZH Z out}=[font={\footnotesize\boldmath}, rounded rectangle, rounded rectangle arc length=90, inner sep=2pt, scale=0.8, tikzit draw=red, tikzit category=ZH]
\tikzstyle{ZH H}=[shape=rectangle, inner sep=0mm, minimum size=2mm, draw=black, fill=white, tikzit shape=rectangle, tikzit category=ZH]
\tikzstyle{ZH H in}=[shape=rectangle, minimum size=4mm, font={\footnotesize\boldmath}, shape=rectangle, inner sep=0.2mm, outer sep=-2mm, scale=0.8, tikzit shape=rectangle, draw=black, fill=white, tikzit draw=blue, tikzit category=ZH]
\tikzstyle{annotated box}=[shape=rectangle, minimum width=10mm, minimum height=5mm, font={\footnotesize\boldmath}, shape=rectangle, inner sep=0.2mm, outer sep=-2mm, scale=0.8, tikzit shape=rectangle, draw=purple, text=purple, fill=white, tikzit fill=purple, tikzit category=ZH, align=center]
\tikzstyle{ZH X none}=[inner sep=0mm, minimum size=2mm, shape=circle, draw=black, fill={rgb,255: red,200; green,200; blue,200}, tikzit category=ZH]
\tikzstyle{ZH X in}=[minimum size=4mm, font={\footnotesize\boldmath}, shape=rectangle, rounded corners=1.6mm, inner sep=0.2mm, outer sep=-2mm, scale=0.8, tikzit shape=circle, draw=black, fill={rgb,255: red,200; green,200; blue,200}, tikzit draw=blue, tikzit category=ZH, tikzit fill=gray]
\tikzstyle{ZH X out}=[font={\footnotesize\boldmath}, rounded rectangle, rounded rectangle arc length=90, fill={rgb,255: red,200; green,200; blue,200}, inner sep=2pt, scale=0.8, tikzit draw=red, tikzit category=ZH, tikzit fill=gray]

\tikzstyle{ZW W none}=[inner sep=0mm, minimum size=2mm, shape=circle, draw=black, fill=black, tikzit category=ZH]
\tikzstyle{ZW W in}=[minimum size=5mm, font={\footnotesize\boldmath}, shape=rectangle, rounded corners=2mm, inner sep=0.2mm, outer sep=-2mm, scale=0.8, tikzit shape=circle, draw=black, fill=black, text=white, tikzit draw=blue, tikzit category=ZH]
\tikzstyle{ZW W out}=[font={\footnotesize\boldmath}, rounded rectangle, rounded rectangle arc length=90, fill=black, inner sep=2pt, scale=0.8, text=white, tikzit draw=red, tikzit category=ZH]

\tikzstyle{GR Tiny Black}=[fill=black, draw=black, shape=circle, inner sep=0pt, minimum size=1mm, tikzit category=GR]
\tikzstyle{GR Tiny Empty}=[draw=black, shape=circle, inner sep=0pt, minimum size=1mm, tikzit category=GR]
\tikzstyle{GR Grey Diamond}=[shape=diamond, draw={rgb,255: red,119; green,119; blue,119}, fill={rgb,255: red,119; green,119; blue,119}, inner sep=0pt, minimum size=1.4mm, tikzit category=GR, tikzit category=GR]

\tikzstyle{BlackTEXT}=[align=center, text=black, font={\small\boldmath}, scale=0.8] 
\tikzstyle{GreyTEXT}=[align=center, text={rgb,255: red,119; green,119; blue,119}, font={\scriptsize\boldmath}, scale=0.8] 
\tikzstyle{OrangeTEXT}=[align=center, text=orange, font={\scriptsize\boldmath}, tikzit draw=orange, tikzit fill=orange, scale=0.8] 
\tikzstyle{PurpleTEXT}=[align=center, text={rgb,255: red,255; green,0; blue,255}, font={\scriptsize\boldmath}, tikzit draw={rgb,255: red,255; green,0; blue,255}, tikzit fill={rgb,255: red,255; green,0; blue,255}, scale=0.8] 

\tikzstyle{cnot targ}=[fill=white, draw=white, shape=circle, tikzit category=circuit, label={center:$\oplus$}, inner sep=0pt, minimum width=2.1mm, tikzit fill={rgb,255: red,102; green,204; blue,255}, tikzit draw=black]
\tikzstyle{ket}=[fill=white, draw=black, shape=regular polygon, regular polygon sides=3, regular polygon rotate=-30, minimum size=8mm, font={\footnotesize\boldmath}, inner sep=0.2mm, outer sep=-2mm, scale=0.8, tikzit category=circuit, tikzit shape=rectangle, tikzit fill=green]
\tikzstyle{bra}=[fill=white, draw=black, shape=regular polygon, regular polygon sides=3, regular polygon rotate=30, minimum size=8mm, font={\footnotesize\boldmath}, inner sep=0.2mm, outer sep=-2mm, scale=0.8, tikzit category=circuit, tikzit shape=rectangle, tikzit fill=red]

\tikzstyle{grey}=[-, draw={rgb,255: red,119; green,119; blue,119}]
\tikzstyle{blue}=[-, draw=blue]
\tikzstyle{red}=[-, draw=red]
\tikzstyle{green}=[-, draw=green]
\tikzstyle{orange}=[-, draw=orange]
\tikzstyle{purple}=[-, draw={rgb,255: red,255; green,0; blue,255}]
\tikzstyle{dashed}=[-, densely dashed]
\tikzstyle{grey dashed}=[-, draw={rgb,255: red,119; green,119; blue,119}, densely dashed]
\tikzstyle{blue dashed}=[-, draw=blue, densely dashed]
\tikzstyle{red dashed}=[-, draw=red, densely dashed]
\tikzstyle{green dashed}=[-, draw=green, densely dashed]
\tikzstyle{orange dashed}=[-, draw=orange, densely dashed]
\tikzstyle{purple dashed}=[-, draw={rgb,255: red,255; green,0; blue,255}, densely dashed]
\tikzstyle{brace edge}=[-, tikzit draw=blue, decorate, decoration={brace,amplitude=1mm,raise=-1mm}]
\tikzstyle{area grey}=[-, fill={rgb,255: red,191; green,191; blue,191}, draw={rgb,255: red,191; green,191; blue,191}]
\tikzstyle{directed}=[->]
\tikzstyle{thick}=[-, line width=1mm]

%% file: zx.tikzstyles
\tikzstyle{box}=[shape=rectangle, text height=1.5ex, text depth=0.25ex, yshift=0.5mm, fill=white, draw=black, minimum height=5mm, yshift=-0.5mm, minimum width=5mm, font={\small}]
\tikzstyle{Z dot}=[inner sep=0mm, minimum size=2mm, shape=circle, draw=black, fill={rgb,255: red,221; green,255; blue,221}]
\tikzstyle{Z phase dot}=[minimum size=1.2em, font={\footnotesize\boldmath}, shape=rectangle, rounded corners=0.5em, inner sep=0.2em, outer sep=-0.2em, scale=0.8, tikzit shape=circle, draw=black, fill={rgb,255: red,221; green,255; blue,221}, tikzit draw=blue]
\tikzstyle{X dot}=[Z dot, shape=circle, draw=black, fill={rgb,255: red,255; green,136; blue,136}]
\tikzstyle{X phase dot}=[Z phase dot, tikzit shape=circle, tikzit draw=blue, fill={rgb,255: red,255; green,136; blue,136}, font={\footnotesize\boldmath}]
\tikzstyle{hadamard}=[fill=yellow, draw=black, shape=rectangle, inner sep=0.6mm, minimum height=1.5mm, minimum width=1.5mm, font={\scriptsize}]
\tikzstyle{vertex}=[inner sep=0mm, minimum size=1mm, shape=circle, draw=black, fill=black]
\tikzstyle{vertex set}=[inner sep=0mm, minimum size=1mm, shape=circle, draw=black, fill=white, font={\footnotesize\boldmath}]
\tikzstyle{target}=[inner sep=0mm, minimum size=3mm, shape=circle, draw=black]
\tikzstyle{white dot}=[inner sep=0mm, minimum size=2mm, shape=circle, draw=black, fill=white]
\tikzstyle{bn}=[inner sep=0mm, minimum size=2mm, shape=circle, draw=black, fill=black]
\tikzstyle{black dot}=[bn, tikzit fill=black, tikzit draw=black]
\tikzstyle{XY}=[bn, tikzit fill=black, tikzit draw=black]
\tikzstyle{YZ}=[fill=black, draw=black, shape=rectangle, minimum size=2mm, inner sep=0mm]
\tikzstyle{XZ}=[fill=black, draw=black, shape=triangle, minimum size=2mm, inner sep=0mm, tikzit shape=rectangle, tikzit fill={rgb,255: red,75; green,227; blue,9}, tikzit draw={rgb,255: red,240; green,26; blue,18}]
\tikzstyle{Pauli dot}=[inner sep=0mm, minimum size=2mm, shape=rectangle, draw=white, fill=white, tikzit draw={rgb,255: red,5; green,9; blue,255}, font={\tiny}]
\tikzstyle{input}=[shape=rectangle, fill=white, draw=black, minimum size=2mm, label={center:$\bullet$}]
\tikzstyle{output}=[inner sep=0mm, minimum size=2mm, shape=circle, draw=black, fill=white]

\tikzstyle{hadamard edge}=[-, dashed, dash pattern=on 2pt off 1.5pt, thick, draw={rgb,255: red,68; green,136; blue,255}]
\tikzstyle{brace edge}=[-, tikzit draw=blue, decorate, decoration={brace,amplitude=1mm,raise=-1mm}]
\tikzstyle{diredge}=[->]
\tikzstyle{dashed edge}=[-, dashed, dash pattern=on 2pt off 0.5pt, draw=black]

%% file: figures/green-spider.tikz
\begin{tikzpicture}
	\begin{pgfonlayer}{nodelayer}
		\node [style=Z phase dot] (0) at (0, 0) {$a,b$};
		\node [style=none] (1) at (-1.25, 0.75) {};
		\node [style=none] (2) at (-1, 0.25) {$\vdots$};
		\node [style=none] (3) at (-1.25, -0.75) {};
		\node [style=none] (4) at (1.25, 0.75) {};
		\node [style=none] (5) at (1, 0.25) {$\vdots$};
		\node [style=none] (6) at (1.25, -0.75) {};
		\node [style=none] (7) at (-1.5, 0) {$m$};
		\node [style=none] (8) at (1.5, 0) {$n$};
	\end{pgfonlayer}
	\begin{pgfonlayer}{edgelayer}
		\draw [bend left=15] (1.center) to (0);
		\draw [bend right=15] (3.center) to (0);
		\draw [bend left=15] (0) to (4.center);
		\draw [bend right=15] (0) to (6.center);
	\end{pgfonlayer}
\end{tikzpicture}

%% file: figures/red-spider.tikz
\begin{tikzpicture}
	\begin{pgfonlayer}{nodelayer}
		\node [style=X phase dot] (0) at (0, 0) {$a,b$};
		\node [style=none] (1) at (-1.25, 0.75) {};
		\node [style=none] (2) at (-1, 0.25) {$\vdots$};
		\node [style=none] (3) at (-1.25, -0.75) {};
		\node [style=none] (4) at (1.25, 0.75) {};
		\node [style=none] (5) at (1, 0.25) {$\vdots$};
		\node [style=none] (6) at (1.25, -0.75) {};
		\node [style=none] (7) at (-1.5, 0) {$m$};
		\node [style=none] (8) at (1.5, 0) {$n$};
	\end{pgfonlayer}
	\begin{pgfonlayer}{edgelayer}
		\draw [bend left=15] (1.center) to (0);
		\draw [bend right=15] (3.center) to (0);
		\draw [bend left=15] (0) to (4.center);
		\draw [bend right=15] (0) to (6.center);
	\end{pgfonlayer}
\end{tikzpicture}

%% file: figures/hadamard.tikz
\begin{tikzpicture}
	\begin{pgfonlayer}{nodelayer}
		\node [style=hadamard] (0) at (0, 0) {};
		\node [style=none] (1) at (-1, 0) {};
		\node [style=none] (2) at (1, 0) {};
	\end{pgfonlayer}
	\begin{pgfonlayer}{edgelayer}
		\draw (1.center) to (2.center);
	\end{pgfonlayer}
\end{tikzpicture}

%% file: figures/identity.tikz
\begin{tikzpicture}
	\begin{pgfonlayer}{nodelayer}
		\node [style=none] (1) at (-1, 0) {};
		\node [style=none] (2) at (1, 0) {};
	\end{pgfonlayer}
	\begin{pgfonlayer}{edgelayer}
		\draw (1.center) to (2.center);
	\end{pgfonlayer}
\end{tikzpicture}

%% file: figures/swap.tikz
\begin{tikzpicture}
	\begin{pgfonlayer}{nodelayer}
		\node [style=none] (1) at (-1, -0.5) {};
		\node [style=none] (2) at (1, 0.5) {};
		\node [style=none] (3) at (-1, 0.5) {};
		\node [style=none] (4) at (1, -0.5) {};
	\end{pgfonlayer}
	\begin{pgfonlayer}{edgelayer}
		\draw [in=180, out=0] (1.center) to (2.center);
		\draw [in=-180, out=0] (3.center) to (4.center);
	\end{pgfonlayer}
\end{tikzpicture}

%% file: figures/h-minus.tikz
\begin{tikzpicture}
	\begin{pgfonlayer}{nodelayer}
		\node [style=hadamard] (0) at (-1.5, 0) {-};
		\node [style=none] (4) at (-2.25, 0) {};
		\node [style=none] (5) at (-0.75, 0) {};
	\end{pgfonlayer}
	\begin{pgfonlayer}{edgelayer}
		\draw (4.center) to (5.center);
	\end{pgfonlayer}
\end{tikzpicture}

%% file: figures/hada-n.tikz
\begin{tikzpicture}
	\begin{pgfonlayer}{nodelayer}
		\node [style=hadamard] (0) at (0, 0) {$a$};
		\node [style=none] (1) at (-1, 0) {};
		\node [style=none] (2) at (1, 0) {};
	\end{pgfonlayer}
	\begin{pgfonlayer}{edgelayer}
		\draw (1.center) to (2.center);
	\end{pgfonlayer}
\end{tikzpicture}

%% file: figures/hada-n-def.tikz
\begin{tikzpicture}
	\begin{pgfonlayer}{nodelayer}
		\node [style=Z dot] (0) at (-1.25, 0) {};
		\node [style=Z dot] (1) at (1.25, 0) {};
		\node [style=none] (2) at (-2, 0) {};
		\node [style=none] (3) at (2, 0) {};
		\node [style=hadamard] (4) at (0, 0.75) {};
		\node [style=hadamard] (5) at (0, -0.75) {};
		\node [style=none] (6) at (-0.25, 0.25) {$\vdots$};
		\node [style=none] (7) at (0.25, 0) {$a$};
	\end{pgfonlayer}
	\begin{pgfonlayer}{edgelayer}
		\draw (2.center) to (0);
		\draw [bend left=15] (0) to (4);
		\draw [bend left=15] (4) to (1);
		\draw (1) to (3.center);
		\draw [bend right=15] (0) to (5);
		\draw [bend right=15] (5) to (1);
	\end{pgfonlayer}
\end{tikzpicture}

%% file: figures/red-phase.tikz
\begin{tikzpicture}
	\begin{pgfonlayer}{nodelayer}
		\node [style=none] (1) at (-1, 0) {};
		\node [style=none] (2) at (2, 0) {};
		\node [style=X phase dot] (3) at (0, 0) {$0,\gamma$};
		\node [style=X dot] (4) at (1.25, 0) {};
	\end{pgfonlayer}
	\begin{pgfonlayer}{edgelayer}
		\draw (1.center) to (2.center);
	\end{pgfonlayer}
\end{tikzpicture}

%% file: figures/phase-G.tikz
\begin{tikzpicture}
	\begin{pgfonlayer}{nodelayer}
		\node [style=none] (1) at (-1.75, 0) {};
		\node [style=none] (2) at (1.75, 0) {};
		\node [style=Z phase dot] (3) at (0, 0) {$0,-\gamma G_{v,w}^2$};
	\end{pgfonlayer}
	\begin{pgfonlayer}{edgelayer}
		\draw (1.center) to (2.center);
	\end{pgfonlayer}
\end{tikzpicture}

%% file: figures/lemma-h-box0.tikz
\begin{tikzpicture}
	\begin{pgfonlayer}{nodelayer}
		\node [style=hadamard] (0) at (-1.5, 0) {$0$};
		\node [style=none] (2) at (-2.25, 0) {};
		\node [style=none] (3) at (-0.75, 0) {};
		\node [style=none] (4) at (0, 0) {$=$};
		\node [style=none] (5) at (0.75, 0) {};
		\node [style=none] (6) at (3, 0) {};
		\node [style=Z dot] (7) at (1.5, 0) {};
		\node [style=Z dot] (8) at (2.25, 0) {};
	\end{pgfonlayer}
	\begin{pgfonlayer}{edgelayer}
		\draw (2.center) to (3.center);
		\draw (5.center) to (7);
		\draw (8) to (6.center);
	\end{pgfonlayer}
\end{tikzpicture}

%% file: figures/lemma-h-box1.tikz
\begin{tikzpicture}
	\begin{pgfonlayer}{nodelayer}
		\node [style=hadamard] (0) at (-1.5, 0) {$1$};
		\node [style=hadamard] (1) at (1.5, 0) {};
		\node [style=none] (2) at (-2.25, 0) {};
		\node [style=none] (3) at (-0.75, 0) {};
		\node [style=none] (4) at (0, 0) {$=$};
		\node [style=none] (5) at (0.75, 0) {};
		\node [style=none] (6) at (2.25, 0) {};
	\end{pgfonlayer}
	\begin{pgfonlayer}{edgelayer}
		\draw (2.center) to (3.center);
		\draw (5.center) to (6.center);
	\end{pgfonlayer}
\end{tikzpicture}

%% file: figures/pivot-helper0.tikz
\begin{tikzpicture}
	\begin{pgfonlayer}{nodelayer}
		\node [style=Z dot] (0) at (-1.5, 0) {};
		\node [style=Z dot] (1) at (1.5, 0) {};
		\node [style=Z dot] (2) at (-4.5, 1) {};
		\node [style=Z dot] (3) at (-4.5, -1) {};
		\node [style=Z dot] (4) at (4.5, 1) {};
		\node [style=Z dot] (5) at (4.5, -1) {};
		\node [style=none] (6) at (-5.5, 1) {};
		\node [style=none] (7) at (-5.5, -1) {};
		\node [style=none] (8) at (5.5, -1) {};
		\node [style=none] (9) at (5.5, 1) {};
		\node [style=hadamard] (10) at (-3, 0.5) {\scriptsize$e_1$};
		\node [style=hadamard] (11) at (-3, -0.5) {\scriptsize$e_k$};
		\node [style=hadamard] (12) at (0, 0) {\scriptsize$\epsilon$};
		\node [style=hadamard] (13) at (3, 0.5) {\scriptsize$f_1$};
		\node [style=hadamard] (14) at (3, -0.5) {\scriptsize$f_j$};
		\node [style=none] (15) at (-1.5, 1.5) {};
		\node [style=none] (16) at (1.5, 1.5) {};
		\node [style=none] (17) at (-4.5, 0) {$\vdots$};
		\node [style=none] (18) at (4.5, 0) {$\vdots$};
	\end{pgfonlayer}
	\begin{pgfonlayer}{edgelayer}
		\draw (7.center) to (3);
		\draw (3) to (11);
		\draw (11) to (0);
		\draw (0) to (12);
		\draw (12) to (1);
		\draw (1) to (14);
		\draw (14) to (5);
		\draw (5) to (8.center);
		\draw (9.center) to (4);
		\draw (4) to (13);
		\draw (13) to (1);
		\draw (6.center) to (2);
		\draw (2) to (10);
		\draw (10) to (0);
		\draw (15.center) to (0);
		\draw (16.center) to (1);
	\end{pgfonlayer}
\end{tikzpicture}

%% file: figures/pivot-helper-final.tikz
\begin{tikzpicture}
	\begin{pgfonlayer}{nodelayer}
		\node [style=Z dot] (2) at (-4.75, 1) {};
		\node [style=Z dot] (3) at (-4.75, -1) {};
		\node [style=none] (6) at (-5.75, 1) {};
		\node [style=none] (7) at (-5.75, -1) {};
		\node [style=none] (8) at (5.75, -1) {};
		\node [style=none] (9) at (5.75, 1) {};
		\node [style=hadamard] (10) at (-3.25, 1.5) {\scriptsize $e_1$};
		\node [style=hadamard] (11) at (-3.25, -0.5) {\scriptsize $e_k$};
		\node [style=none] (15) at (1.25, 3) {};
		\node [style=none] (16) at (-1.25, 3) {};
		\node [style=none] (17) at (-4.75, 0.25) {$\vdots$};
		\node [style=none] (18) at (4.75, 0.25) {$\vdots$};
		\node [style=hadamard] (20) at (1.25, 2.75) {\scriptsize $\epsilon$};
		\node [style=Z dot] (23) at (1.25, 2) {};
		\node [style=Z dot] (24) at (-1.25, 2) {};
		\node [style=Z dot] (25) at (4.75, 1) {};
		\node [style=Z dot] (26) at (4.75, -1) {};
		\node [style=hadamard] (29) at (2.5, 1.75) {\scriptsize $f_1$};
		\node [style=hadamard] (32) at (2.5, 0) {\scriptsize $f_j$};
		\node [style=hadamard] (39) at (0, 1) {\scriptsize $w_{1,1}$};
		\node [style=hadamard] (40) at (0, -0.5) {\scriptsize $w_{1,j}$};
		\node [style=hadamard] (41) at (0, -2) {\scriptsize $w_{k,1}$};
		\node [style=hadamard] (42) at (0, -3.5) {\scriptsize $w_{k,j}$};
		\node [style=hadamard] (44) at (0, 2) {\scriptsize $-\epsilon$};
		\node [style=hadamard] (45) at (-1.25, 2.75) {\scriptsize $\epsilon$};
		\node [style=none] (46) at (0, 0.5) {$\vdots$};
		\node [style=none] (47) at (0, -1) {$\vdots$};
		\node [style=none] (48) at (0, -2.5) {$\vdots$};
		\node [style=none] (49) at (-1.25, 4.5) {};
		\node [style=none] (50) at (1.25, 4.5) {};
	\end{pgfonlayer}
	\begin{pgfonlayer}{edgelayer}
		\draw (7.center) to (3);
		\draw (3) to (11);
		\draw (6.center) to (2);
		\draw (2) to (10);
		\draw (15.center) to (23);
		\draw (23) to (24);
		\draw (9.center) to (25);
		\draw (26) to (8.center);
		\draw (23) to (29);
		\draw (29) to (25);
		\draw (26) to (32);
		\draw [in=-75, out=150] (32) to (23);
		\draw (2) to (39);
		\draw (2) to (40);
		\draw [in=180, out=-15] (3) to (41);
		\draw [in=180, out=-45] (3) to (42);
		\draw [in=180, out=0] (39) to (25);
		\draw (40) to (26);
		\draw [in=-135, out=0] (41) to (25);
		\draw [in=-135, out=0] (42) to (26);
		\draw (10) to (24);
		\draw [in=15, out=-105] (24) to (11);
		\draw (16.center) to (45);
		\draw (45) to (24);
		\draw [in=90, out=-90] (49.center) to (15.center);
		\draw [in=-90, out=90] (16.center) to (50.center);
	\end{pgfonlayer}
\end{tikzpicture}

%% file: figures/pivot-helper1.tikz
\begin{tikzpicture}
	\begin{pgfonlayer}{nodelayer}
		\node [style=X dot] (0) at (-0.5, 0) {};
		\node [style=Z dot] (1) at (1.5, 0) {};
		\node [style=Z dot] (2) at (-4.5, 1) {};
		\node [style=Z dot] (3) at (-4.5, -1) {};
		\node [style=Z dot] (4) at (4.5, 1) {};
		\node [style=Z dot] (5) at (4.5, -1) {};
		\node [style=none] (6) at (-5.25, 1) {};
		\node [style=none] (7) at (-5.25, -1) {};
		\node [style=none] (8) at (5.25, -1) {};
		\node [style=none] (9) at (5.25, 1) {};
		\node [style=hadamard] (10) at (-3.5, 0.75) {$e_1$};
		\node [style=hadamard] (11) at (-3.5, -0.75) {$e_k$};
		\node [style=hadamard] (12) at (-2, 0.5) {$-\epsilon$};
		\node [style=hadamard] (13) at (3, 0.5) {$f_1$};
		\node [style=hadamard] (14) at (3, -0.5) {$f_j$};
		\node [style=none] (15) at (-0.5, 1.5) {};
		\node [style=none] (16) at (1.5, 1.5) {};
		\node [style=none] (17) at (-4.5, 0) {$\vdots$};
		\node [style=none] (18) at (4.5, 0) {$\vdots$};
		\node [style=hadamard] (19) at (-2, -0.5) {$-\epsilon$};
		\node [style=hadamard] (20) at (-0.5, 0.75) {$-\epsilon$};
	\end{pgfonlayer}
	\begin{pgfonlayer}{edgelayer}
		\draw (7.center) to (3);
		\draw (3) to (11);
		\draw (1) to (14);
		\draw (14) to (5);
		\draw (5) to (8.center);
		\draw (9.center) to (4);
		\draw (4) to (13);
		\draw (13) to (1);
		\draw (6.center) to (2);
		\draw (2) to (10);
		\draw (15.center) to (0);
		\draw (16.center) to (1);
		\draw (10) to (12);
		\draw (12) to (0);
		\draw (11) to (19);
		\draw (19) to (0);
		\draw (0) to (1);
	\end{pgfonlayer}
\end{tikzpicture}

%% file: figures/pivot-helper2.tikz
\begin{tikzpicture}
	\begin{pgfonlayer}{nodelayer}
		\node [style=X dot] (0) at (-1, 0) {};
		\node [style=Z dot] (1) at (1, 0) {};
		\node [style=Z dot] (2) at (-5.5, 1) {};
		\node [style=Z dot] (3) at (-5.5, -1) {};
		\node [style=Z dot] (4) at (4, 1) {};
		\node [style=Z dot] (5) at (4, -1) {};
		\node [style=none] (6) at (-6.25, 1) {};
		\node [style=none] (7) at (-6.25, -1) {};
		\node [style=none] (8) at (4.75, -1) {};
		\node [style=none] (9) at (4.75, 1) {};
		\node [style=hadamard] (10) at (-3.75, 0.75) {$-\epsilon^{-1}e_1$};
		\node [style=hadamard] (11) at (-3.75, -0.75) {$-\epsilon^{-1}e_k$};
		\node [style=hadamard] (12) at (-2, 0.5) {};
		\node [style=hadamard] (13) at (2.5, 0.5) {$f_1$};
		\node [style=hadamard] (14) at (2.5, -0.5) {$f_j$};
		\node [style=none] (15) at (-1, 1.5) {};
		\node [style=none] (16) at (1, 1.5) {};
		\node [style=none] (17) at (-5.5, 0) {$\vdots$};
		\node [style=none] (18) at (4, 0) {$\vdots$};
		\node [style=hadamard] (19) at (-2, -0.5) {};
		\node [style=hadamard] (20) at (-1, 0.75) {$-\epsilon$};
	\end{pgfonlayer}
	\begin{pgfonlayer}{edgelayer}
		\draw (7.center) to (3);
		\draw (3) to (11);
		\draw (1) to (14);
		\draw (14) to (5);
		\draw (5) to (8.center);
		\draw (9.center) to (4);
		\draw (4) to (13);
		\draw (13) to (1);
		\draw (6.center) to (2);
		\draw (2) to (10);
		\draw (15.center) to (0);
		\draw (16.center) to (1);
		\draw (10) to (12);
		\draw (12) to (0);
		\draw (11) to (19);
		\draw (19) to (0);
		\draw (0) to (1);
	\end{pgfonlayer}
\end{tikzpicture}

%% file: figures/pivot-helper3.tikz
\begin{tikzpicture}
	\begin{pgfonlayer}{nodelayer}
		\node [style=Z dot] (2) at (-5.5, 1) {};
		\node [style=Z dot] (3) at (-5.5, -1) {};
		\node [style=Z dot] (4) at (4, 1) {};
		\node [style=Z dot] (5) at (4, -1) {};
		\node [style=none] (6) at (-6.25, 1) {};
		\node [style=none] (7) at (-6.25, -1) {};
		\node [style=none] (8) at (4.75, -1) {};
		\node [style=none] (9) at (4.75, 1) {};
		\node [style=hadamard] (10) at (-3.75, 1) {$-\epsilon^{-1}e_1$};
		\node [style=hadamard] (11) at (-3.75, -1) {$-\epsilon^{-1}e_k$};
		\node [style=hadamard] (12) at (-2, 1) {};
		\node [style=hadamard] (13) at (2.5, 1) {$f_1$};
		\node [style=hadamard] (14) at (2.5, -1) {$f_j$};
		\node [style=none] (15) at (-1, 3.25) {};
		\node [style=none] (16) at (1, 3.25) {};
		\node [style=none] (17) at (-5.5, 0) {$\vdots$};
		\node [style=none] (18) at (4, 0) {$\vdots$};
		\node [style=hadamard] (19) at (-2, -1) {};
		\node [style=hadamard] (20) at (-1, 2.5) {$-\epsilon$};
		\node [style=Z dot] (21) at (-1, 1) {};
		\node [style=Z dot] (22) at (-1, -1) {};
		\node [style=Z dot] (23) at (-1, 1.75) {};
		\node [style=X dot] (24) at (1, 1.75) {};
		\node [style=X dot] (25) at (1, 1) {};
		\node [style=X dot] (26) at (1, -1) {};
		\node [style=none] (27) at (-1, 0) {$\vdots$};
		\node [style=none] (28) at (1, 0) {$\vdots$};
	\end{pgfonlayer}
	\begin{pgfonlayer}{edgelayer}
		\draw (7.center) to (3);
		\draw (3) to (11);
		\draw (14) to (5);
		\draw (5) to (8.center);
		\draw (9.center) to (4);
		\draw (4) to (13);
		\draw (6.center) to (2);
		\draw (2) to (10);
		\draw (10) to (12);
		\draw (11) to (19);
		\draw (15.center) to (23);
		\draw (16.center) to (24);
		\draw (12) to (21);
		\draw (19) to (22);
		\draw (26) to (14);
		\draw (25) to (13);
		\draw (23) to (26);
		\draw (23) to (25);
		\draw (23) to (24);
		\draw (24) to (21);
		\draw (21) to (25);
		\draw (25) to (22);
		\draw (22) to (24);
		\draw (22) to (26);
		\draw (21) to (26);
	\end{pgfonlayer}
\end{tikzpicture}

%% file: figures/pivot-helper4.tikz
\begin{tikzpicture}
	\begin{pgfonlayer}{nodelayer}
		\node [style=Z dot] (2) at (-5.5, 1) {};
		\node [style=Z dot] (3) at (-5.5, -1.5) {};
		\node [style=none] (6) at (-6.25, 1) {};
		\node [style=none] (7) at (-6.25, -1.5) {};
		\node [style=none] (8) at (5.25, -1.5) {};
		\node [style=none] (9) at (5.25, 1) {};
		\node [style=hadamard] (10) at (-3.75, 1) {\scriptsize $-\epsilon^{-1}e_1$};
		\node [style=hadamard] (11) at (-3.75, -1.5) {\scriptsize $-\epsilon^{-1}e_k$};
		\node [style=hadamard] (12) at (-2, 1) {};
		\node [style=hadamard] (13) at (0.5, 1) {\scriptsize $-f_1$};
		\node [style=hadamard] (14) at (0.75, -1.5) {\scriptsize $-f_j$};
		\node [style=none] (15) at (-0.75, 4.5) {};
		\node [style=none] (16) at (4, 4.5) {};
		\node [style=none] (17) at (-5.5, 0) {$\vdots$};
		\node [style=none] (18) at (4.5, 0) {$\vdots$};
		\node [style=hadamard] (19) at (-2, -1.5) {};
		\node [style=hadamard] (20) at (-0.75, 3.75) {\scriptsize $-\epsilon$};
		\node [style=Z dot] (21) at (-1, 1) {};
		\node [style=Z dot] (22) at (-1, -1.5) {};
		\node [style=Z dot] (23) at (-0.75, 3) {};
		\node [style=X dot] (24) at (4, 3) {};
		\node [style=Z dot] (25) at (4.5, 1) {};
		\node [style=Z dot] (26) at (4.5, -1.5) {};
		\node [style=none] (27) at (-1, 0) {$\vdots$};
		\node [style=hadamard] (29) at (1, 2.5) {\scriptsize $-f_1$};
		\node [style=hadamard] (30) at (3.25, 0.5) {\scriptsize $-f_1$};
		\node [style=hadamard] (31) at (2, -1) {\scriptsize $-f_j$};
		\node [style=hadamard] (32) at (3.25, -0.5) {\scriptsize $-f_j$};
	\end{pgfonlayer}
	\begin{pgfonlayer}{edgelayer}
		\draw (7.center) to (3);
		\draw (3) to (11);
		\draw (6.center) to (2);
		\draw (2) to (10);
		\draw (10) to (12);
		\draw (11) to (19);
		\draw (15.center) to (23);
		\draw (16.center) to (24);
		\draw (12) to (21);
		\draw (19) to (22);
		\draw (23) to (24);
		\draw (24) to (21);
		\draw (21) to (25);
		\draw (22) to (24);
		\draw (22) to (26);
		\draw (9.center) to (25);
		\draw (26) to (8.center);
		\draw (23) to (29);
		\draw (29) to (25);
		\draw (25) to (30);
		\draw (30) to (22);
		\draw (26) to (32);
		\draw (32) to (23);
		\draw (21) to (31);
		\draw (31) to (26);
	\end{pgfonlayer}
\end{tikzpicture}

%% file: figures/pivot-helper5.tikz
\begin{tikzpicture}
	\begin{pgfonlayer}{nodelayer}
		\node [style=Z dot] (2) at (-5.5, 1) {};
		\node [style=Z dot] (3) at (-5.5, -1.5) {};
		\node [style=none] (6) at (-6.25, 1) {};
		\node [style=none] (7) at (-6.25, -1.5) {};
		\node [style=none] (8) at (5.25, -1.5) {};
		\node [style=none] (9) at (5.25, 1) {};
		\node [style=hadamard] (10) at (-3.75, 1) {\scriptsize $-\epsilon^{-1}e_1$};
		\node [style=hadamard] (11) at (-3.75, -1.5) {\scriptsize $-\epsilon^{-1}e_k$};
		\node [style=hadamard] (13) at (0.5, 1) {\scriptsize $-f_1$};
		\node [style=hadamard] (14) at (0.75, -1.5) {\scriptsize $-f_j$};
		\node [style=none] (15) at (-0.75, 4.5) {};
		\node [style=none] (16) at (4, 4.5) {};
		\node [style=none] (17) at (-5.5, 0) {$\vdots$};
		\node [style=none] (18) at (4.5, 0) {$\vdots$};
		\node [style=hadamard] (20) at (-0.75, 3.75) {\scriptsize $-\epsilon$};
		\node [style=X dot] (21) at (-2, 1) {};
		\node [style=X dot] (22) at (-2, -1.5) {};
		\node [style=Z dot] (23) at (-0.75, 3) {};
		\node [style=X dot] (24) at (4, 3) {};
		\node [style=Z dot] (25) at (4.5, 1) {};
		\node [style=Z dot] (26) at (4.5, -1.5) {};
		\node [style=none] (27) at (-2, 0) {$\vdots$};
		\node [style=hadamard] (29) at (1, 2.5) {\scriptsize $-f_1$};
		\node [style=hadamard] (30) at (3.25, 0.5) {\scriptsize $-f_1$};
		\node [style=hadamard] (31) at (2, -1) {\scriptsize $-f_j$};
		\node [style=hadamard] (32) at (3.25, -0.5) {\scriptsize $-f_j$};
		\node [style=hadamard] (33) at (-0.75, 1) {-};
		\node [style=hadamard] (34) at (-0.75, 1.5) {-};
		\node [style=hadamard] (35) at (-0.75, 0.5) {-};
		\node [style=hadamard] (36) at (-0.5, -0.5) {-};
		\node [style=hadamard] (37) at (-0.5, -1) {-};
		\node [style=hadamard] (38) at (-0.5, -1.5) {-};
	\end{pgfonlayer}
	\begin{pgfonlayer}{edgelayer}
		\draw (7.center) to (3);
		\draw (3) to (11);
		\draw (6.center) to (2);
		\draw (2) to (10);
		\draw (15.center) to (23);
		\draw (16.center) to (24);
		\draw (23) to (24);
		\draw (21) to (25);
		\draw (22) to (26);
		\draw (9.center) to (25);
		\draw (26) to (8.center);
		\draw (23) to (29);
		\draw (29) to (25);
		\draw (25) to (30);
		\draw (30) to (22);
		\draw (26) to (32);
		\draw (32) to (23);
		\draw (31) to (26);
		\draw (10) to (21);
		\draw (11) to (22);
		\draw (22) to (36);
		\draw [bend right=15] (36) to (24);
		\draw (21) to (35);
		\draw (35) to (31);
		\draw (21) to (34);
		\draw [in=210, out=15] (34) to (24);
	\end{pgfonlayer}
\end{tikzpicture}

%% file: figures/pivot-helper6.tikz
\begin{tikzpicture}
	\begin{pgfonlayer}{nodelayer}
		\node [style=none] (43) at (-3.5, 2.25) {};
		\node [style=none] (44) at (-3.5, 1.25) {};
		\node [style=none] (45) at (-3.5, 0.25) {};
		\node [style=none] (46) at (-3.5, -0.75) {};
		\node [style=none] (47) at (-3.5, -1.75) {};
		\node [style=none] (48) at (-3.5, -2.75) {};
		\node [style=Z dot] (2) at (-5.5, 1) {};
		\node [style=Z dot] (3) at (-5.5, -1.5) {};
		\node [style=none] (6) at (-6.25, 1) {};
		\node [style=none] (7) at (-6.25, -1.5) {};
		\node [style=none] (8) at (5.75, -1.5) {};
		\node [style=none] (9) at (5.75, 1) {};
		\node [style=hadamard] (10) at (-2.75, 2.25) {\scriptsize $\epsilon^{-1}e_1$};
		\node [style=hadamard] (11) at (-2.75, -0.75) {\scriptsize $\epsilon^{-1}e_k$};
		\node [style=hadamard] (13) at (3.5, 1.25) {\scriptsize $-f_1$};
		\node [style=hadamard] (14) at (3.5, -2.75) {\scriptsize $-f_j$};
		\node [style=none] (15) at (0, 4.5) {};
		\node [style=none] (16) at (2.5, 4.5) {};
		\node [style=none] (17) at (-5.5, 0) {$\vdots$};
		\node [style=none] (18) at (5, 0) {$\vdots$};
		\node [style=hadamard] (20) at (0, 3.75) {\scriptsize $-\epsilon$};
		\node [style=Z dot] (23) at (0, 3) {};
		\node [style=X dot] (24) at (2.5, 3) {};
		\node [style=Z dot] (25) at (5, 1) {};
		\node [style=Z dot] (26) at (5, -1.5) {};
		\node [style=hadamard] (29) at (3.5, 2.25) {\scriptsize $-f_1$};
		\node [style=hadamard] (30) at (3.5, 0.25) {\scriptsize $-f_1$};
		\node [style=hadamard] (31) at (3.5, -1.75) {\scriptsize $-f_j$};
		\node [style=hadamard] (32) at (3.5, -0.75) {\scriptsize $-f_j$};
		\node [style=hadamard] (33) at (-1, 2.25) {\scriptsize -};
		\node [style=hadamard] (34) at (-1, 1.25) {\scriptsize -};
		\node [style=hadamard] (35) at (-1, 0.25) {\scriptsize -};
		\node [style=hadamard] (36) at (-1, -0.75) {\scriptsize -};
		\node [style=hadamard] (37) at (-1, -1.75) {\scriptsize -};
		\node [style=hadamard] (38) at (-1, -2.75) {\scriptsize -};
		\node [style=hadamard] (39) at (-2.75, 1.25) {\scriptsize $\epsilon^{-1}e_1$};
		\node [style=hadamard] (40) at (-2.75, 0.25) {\scriptsize $\epsilon^{-1}e_1$};
		\node [style=hadamard] (41) at (-2.75, -1.75) {\scriptsize $\epsilon^{-1}e_k$};
		\node [style=hadamard] (42) at (-2.75, -2.75) {\scriptsize $\epsilon^{-1}e_k$};
		\node [style=none] (49) at (4, 2.25) {};
		\node [style=none] (50) at (4, 1.25) {};
		\node [style=none] (51) at (4, 0.25) {};
		\node [style=none] (52) at (4, -0.75) {};
		\node [style=none] (53) at (4, -1.75) {};
		\node [style=none] (54) at (4, -2.75) {};
		\node [style=none] (55) at (3, -0.75) {};
		\node [style=none] (56) at (3, 0.25) {};
	\end{pgfonlayer}
	\begin{pgfonlayer}{edgelayer}
		\draw (7.center) to (3);
		\draw (6.center) to (2);
		\draw (15.center) to (23);
		\draw (16.center) to (24);
		\draw (23) to (24);
		\draw (9.center) to (25);
		\draw (26) to (8.center);
		\draw (23) to (29);
		\draw (33) to (24);
		\draw (35) to (31);
		\draw (36) to (24);
		\draw (10) to (33);
		\draw (39) to (34);
		\draw (40) to (35);
		\draw (42) to (38);
		\draw (38) to (14);
		\draw (34) to (13);
		\draw (41) to (37);
		\draw (11) to (36);
		\draw (2) to (43.center);
		\draw (44.center) to (2);
		\draw (2) to (45.center);
		\draw (46.center) to (3);
		\draw (3) to (47.center);
		\draw (48.center) to (3);
		\draw (49.center) to (25);
		\draw (25) to (50.center);
		\draw (25) to (51.center);
		\draw (52.center) to (26);
		\draw (26) to (53.center);
		\draw (26) to (54.center);
		\draw (55.center) to (23);
		\draw (37) to (56.center);
	\end{pgfonlayer}
\end{tikzpicture}

%% file: figures/pivot-helper7.tikz
\begin{tikzpicture}
	\begin{pgfonlayer}{nodelayer}
		\node [style=Z dot] (2) at (-5.5, 1) {};
		\node [style=Z dot] (3) at (-5.5, -1.5) {};
		\node [style=none] (6) at (-6.5, 1) {};
		\node [style=none] (7) at (-6.5, -1.5) {};
		\node [style=none] (8) at (5.5, -1.5) {};
		\node [style=none] (9) at (5.5, 1) {};
		\node [style=hadamard] (10) at (-2.5, 2.25) {\scriptsize $\epsilon^{-1}e_1$};
		\node [style=hadamard] (11) at (-2.5, -0.75) {\scriptsize $\epsilon^{-1}e_k$};
		\node [style=none] (15) at (-0.5, 4.5) {};
		\node [style=none] (16) at (3, 4.5) {};
		\node [style=none] (17) at (-5.5, 0) {$\vdots$};
		\node [style=none] (18) at (4.5, 0) {$\vdots$};
		\node [style=hadamard] (20) at (-0.5, 3.75) {\scriptsize $-\epsilon$};
		\node [style=Z dot] (23) at (-0.5, 3) {};
		\node [style=X dot] (24) at (3, 3) {};
		\node [style=Z dot] (25) at (4.5, 1) {};
		\node [style=Z dot] (26) at (4.5, -1.5) {};
		\node [style=hadamard] (29) at (3, 1.75) {\scriptsize $-f_1$};
		\node [style=hadamard] (32) at (3.5, -0.5) {\scriptsize $-f_j$};
		\node [style=hadamard] (33) at (-1, 2.25) {\scriptsize -};
		\node [style=hadamard] (36) at (-1, -0.75) {\scriptsize -};
		\node [style=hadamard] (39) at (-2.25, 1.25) {\scriptsize $-\epsilon^{-1} e_1 f_1$};
		\node [style=hadamard] (40) at (-2.25, 0.25) {\scriptsize $-\epsilon^{-1} e_1 f_j$};
		\node [style=hadamard] (41) at (-2.25, -1.75) {\scriptsize $-\epsilon^{-1} e_k f_1$};
		\node [style=hadamard] (42) at (-2.25, -2.75) {\scriptsize $-\epsilon^{-1} e_k f_j$};
		\node [style=none] (43) at (-3.25, 2.25) {};
		\node [style=none] (44) at (-3.25, 1.25) {};
		\node [style=none] (45) at (-3.25, 0.25) {};
		\node [style=none] (46) at (-3.25, -0.75) {};
		\node [style=none] (47) at (-3.25, -1.75) {};
		\node [style=none] (48) at (-3.25, -2.75) {};
		\node [style=none] (49) at (-1.75, 2.25) {};
		\node [style=none] (50) at (-1, 1.25) {};
		\node [style=none] (51) at (-1, 0.25) {};
		\node [style=none] (52) at (-1.75, -0.75) {};
		\node [style=none] (53) at (-1, -1.75) {};
		\node [style=none] (54) at (-1, -2.75) {};
	\end{pgfonlayer}
	\begin{pgfonlayer}{edgelayer}
		\draw (7.center) to (3);
		\draw (6.center) to (2);
		\draw (15.center) to (23);
		\draw (16.center) to (24);
		\draw (23) to (24);
		\draw (9.center) to (25);
		\draw (26) to (8.center);
		\draw (23) to (29);
		\draw (29) to (25);
		\draw (26) to (32);
		\draw (32) to (23);
		\draw (33) to (24);
		\draw [in=-150, out=15] (36) to (24);
		\draw (2) to (43.center);
		\draw (44.center) to (2);
		\draw (2) to (45.center);
		\draw (46.center) to (3);
		\draw (3) to (47.center);
		\draw (48.center) to (3);
		\draw (49.center) to (33);
		\draw [in=-180, out=0] (50.center) to (25);
		\draw [in=180, out=0] (51.center) to (26);
		\draw (52.center) to (36);
		\draw (25) to (53.center);
		\draw (54.center) to (26);
	\end{pgfonlayer}
\end{tikzpicture}

%% file: figures/pivot-helper8.tikz
\begin{tikzpicture}
	\begin{pgfonlayer}{nodelayer}
		\node [style=Z dot] (2) at (-5.5, 1) {};
		\node [style=Z dot] (3) at (-5.5, -1.5) {};
		\node [style=none] (6) at (-6.5, 1) {};
		\node [style=none] (7) at (-6.5, -1.5) {};
		\node [style=none] (8) at (5.5, -1.5) {};
		\node [style=none] (9) at (5.5, 1) {};
		\node [style=hadamard] (10) at (-2.5, 2.25) {\scriptsize $\epsilon^{-1}e_1$};
		\node [style=hadamard] (11) at (-2.5, -0.75) {\scriptsize $\epsilon^{-1}e_k$};
		\node [style=none] (15) at (-0.5, 4.5) {};
		\node [style=none] (17) at (-5.5, 0) {$\vdots$};
		\node [style=none] (18) at (4.5, 0) {$\vdots$};
		\node [style=hadamard] (20) at (-0.5, 3.75) {\scriptsize $-\epsilon$};
		\node [style=Z dot] (23) at (-0.5, 3) {};
		\node [style=X dot] (24) at (1.75, 3) {};
		\node [style=Z dot] (25) at (4.5, 1) {};
		\node [style=Z dot] (26) at (4.5, -1.5) {};
		\node [style=hadamard] (29) at (3, 1.75) {\scriptsize $-f_1$};
		\node [style=hadamard] (32) at (3.5, -0.5) {\scriptsize $-f_j$};
		\node [style=hadamard] (33) at (-1, 2.25) {\scriptsize -};
		\node [style=hadamard] (36) at (-1, -0.75) {\scriptsize -};
		\node [style=hadamard] (39) at (-2.25, 1.25) {\scriptsize $-\epsilon^{-1} e_1 f_1$};
		\node [style=hadamard] (40) at (-2.25, 0.25) {\scriptsize $-\epsilon^{-1} e_1 f_j$};
		\node [style=hadamard] (41) at (-2.25, -1.75) {\scriptsize $-\epsilon^{-1} e_k f_1$};
		\node [style=hadamard] (42) at (-2.25, -2.75) {\scriptsize $-\epsilon^{-1} e_k f_j$};
		\node [style=none] (43) at (-3.25, 2.25) {};
		\node [style=none] (44) at (-3.25, 1.25) {};
		\node [style=none] (45) at (-3.25, 0.25) {};
		\node [style=none] (46) at (-3.25, -0.75) {};
		\node [style=none] (47) at (-3.25, -1.75) {};
		\node [style=none] (48) at (-3.25, -2.75) {};
		\node [style=none] (49) at (-1.75, 2.25) {};
		\node [style=none] (50) at (-1, 1.25) {};
		\node [style=none] (51) at (-1, 0.25) {};
		\node [style=none] (52) at (-1.75, -0.75) {};
		\node [style=none] (53) at (-1, -1.75) {};
		\node [style=none] (54) at (-1, -2.75) {};
		\node [style=hadamard] (55) at (3, 3) {$-\epsilon$};
		\node [style=hadamard] (57) at (4.25, 3) {$\epsilon$};
		\node [style=none] (59) at (4.25, 4.5) {};
	\end{pgfonlayer}
	\begin{pgfonlayer}{edgelayer}
		\draw (7.center) to (3);
		\draw (6.center) to (2);
		\draw (15.center) to (23);
		\draw (23) to (24);
		\draw (9.center) to (25);
		\draw (26) to (8.center);
		\draw (23) to (29);
		\draw (29) to (25);
		\draw (26) to (32);
		\draw (32) to (23);
		\draw [in=-150, out=0] (33) to (24);
		\draw [in=-120, out=15, looseness=0.75] (36) to (24);
		\draw (2) to (43.center);
		\draw (44.center) to (2);
		\draw (2) to (45.center);
		\draw (46.center) to (3);
		\draw (3) to (47.center);
		\draw (48.center) to (3);
		\draw (49.center) to (33);
		\draw [in=-180, out=0] (50.center) to (25);
		\draw [in=180, out=0] (51.center) to (26);
		\draw (52.center) to (36);
		\draw (25) to (53.center);
		\draw (54.center) to (26);
		\draw (59.center) to (57);
		\draw (24) to (57);
	\end{pgfonlayer}
\end{tikzpicture}

%% file: figures/pivot-helper9.tikz
\begin{tikzpicture}
	\begin{pgfonlayer}{nodelayer}
		\node [style=Z dot] (2) at (-5.5, 1) {};
		\node [style=Z dot] (3) at (-5.5, -1.5) {};
		\node [style=none] (6) at (-6.5, 1) {};
		\node [style=none] (7) at (-6.5, -1.5) {};
		\node [style=none] (8) at (5.5, -1.5) {};
		\node [style=none] (9) at (5.5, 1) {};
		\node [style=hadamard] (10) at (-2.75, 2.25) {\scriptsize $\epsilon^{-1}e_1$};
		\node [style=hadamard] (11) at (-2.5, -0.75) {\scriptsize $\epsilon^{-1}e_k$};
		\node [style=none] (15) at (-0.25, 4.5) {};
		\node [style=none] (17) at (-5.5, 0) {$\vdots$};
		\node [style=none] (18) at (4.5, 0) {$\vdots$};
		\node [style=hadamard] (20) at (-0.25, 3.75) {\scriptsize $-\epsilon$};
		\node [style=Z dot] (23) at (-0.25, 3) {};
		\node [style=Z dot] (24) at (2.75, 3) {};
		\node [style=Z dot] (25) at (4.5, 1) {};
		\node [style=Z dot] (26) at (4.5, -1.5) {};
		\node [style=hadamard] (29) at (3, 1.75) {\scriptsize $-f_1$};
		\node [style=hadamard] (32) at (3.5, -0.5) {\scriptsize $-f_j$};
		\node [style=hadamard] (33) at (-1.5, 2.25) {\scriptsize -};
		\node [style=hadamard] (36) at (-1, -0.75) {\scriptsize -};
		\node [style=hadamard] (39) at (-2.25, 1.25) {\scriptsize $-\epsilon^{-1} e_1 f_1$};
		\node [style=hadamard] (40) at (-2.25, 0.25) {\scriptsize $-\epsilon^{-1} e_1 f_j$};
		\node [style=hadamard] (41) at (-2.25, -1.75) {\scriptsize $-\epsilon^{-1} e_k f_1$};
		\node [style=hadamard] (42) at (-2.25, -2.75) {\scriptsize $-\epsilon^{-1} e_k f_j$};
		\node [style=none] (43) at (-3.25, 2.25) {};
		\node [style=none] (44) at (-3.25, 1.25) {};
		\node [style=none] (45) at (-3.25, 0.25) {};
		\node [style=none] (46) at (-3.25, -0.75) {};
		\node [style=none] (47) at (-3.25, -1.75) {};
		\node [style=none] (48) at (-3.25, -2.75) {};
		\node [style=none] (49) at (-2, 2.25) {};
		\node [style=none] (50) at (-1, 1.25) {};
		\node [style=none] (51) at (-1, 0.25) {};
		\node [style=none] (52) at (-1.75, -0.75) {};
		\node [style=none] (53) at (-1, -1.75) {};
		\node [style=none] (54) at (-1, -2.75) {};
		\node [style=hadamard] (55) at (1.5, 3) {$\epsilon$};
		\node [style=hadamard] (57) at (2.75, 3.75) {$\epsilon$};
		\node [style=none] (59) at (2.75, 4.5) {};
		\node [style=hadamard] (60) at (-0.5, 2.25) {$\epsilon$};
		\node [style=hadamard] (61) at (1, 0.5) {$\epsilon$};
	\end{pgfonlayer}
	\begin{pgfonlayer}{edgelayer}
		\draw (7.center) to (3);
		\draw (6.center) to (2);
		\draw (15.center) to (23);
		\draw (23) to (24);
		\draw (9.center) to (25);
		\draw (26) to (8.center);
		\draw (23) to (29);
		\draw (29) to (25);
		\draw (26) to (32);
		\draw (32) to (23);
		\draw (2) to (43.center);
		\draw (44.center) to (2);
		\draw (2) to (45.center);
		\draw (46.center) to (3);
		\draw (3) to (47.center);
		\draw (48.center) to (3);
		\draw (49.center) to (33);
		\draw (50.center) to (25);
		\draw [in=180, out=-15] (51.center) to (26);
		\draw (52.center) to (36);
		\draw (25) to (53.center);
		\draw (54.center) to (26);
		\draw (59.center) to (57);
		\draw (24) to (57);
		\draw (33) to (60);
		\draw [in=-150, out=0] (60) to (24);
		\draw (24) to (61);
		\draw [bend left=15] (61) to (36);
	\end{pgfonlayer}
\end{tikzpicture}

%% file: figures/pivot-helper10.tikz
\begin{tikzpicture}
	\begin{pgfonlayer}{nodelayer}
		\node [style=Z dot] (2) at (-5.5, 1) {};
		\node [style=Z dot] (3) at (-5.5, -1.5) {};
		\node [style=none] (6) at (-6.5, 1) {};
		\node [style=none] (7) at (-6.5, -1.5) {};
		\node [style=none] (8) at (5, -1.5) {};
		\node [style=none] (9) at (5, 1) {};
		\node [style=hadamard] (10) at (-1.5, 2.25) {\scriptsize $e_1$};
		\node [style=hadamard] (11) at (-1.5, -0.75) {\scriptsize $e_k$};
		\node [style=none] (15) at (-3, 4.5) {};
		\node [style=none] (17) at (-5.5, 0) {$\vdots$};
		\node [style=none] (18) at (4, 0) {$\vdots$};
		\node [style=hadamard] (20) at (-3, 3.75) {\scriptsize $\epsilon$};
		\node [style=Z dot] (23) at (-1, 3) {};
		\node [style=Z dot] (24) at (2, 3) {};
		\node [style=Z dot] (25) at (4, 1) {};
		\node [style=Z dot] (26) at (4, -1.5) {};
		\node [style=hadamard] (29) at (2.75, 1.75) {\scriptsize $-f_1$};
		\node [style=hadamard] (32) at (3, -0.5) {\scriptsize $-f_j$};
		\node [style=hadamard] (39) at (-2.25, 1.25) {\scriptsize $-\epsilon^{-1} e_1 f_1$};
		\node [style=hadamard] (40) at (-2.25, 0.25) {\scriptsize $-\epsilon^{-1} e_1 f_j$};
		\node [style=hadamard] (41) at (-2.25, -1.75) {\scriptsize $-\epsilon^{-1} e_k f_1$};
		\node [style=hadamard] (42) at (-2.25, -2.75) {\scriptsize $-\epsilon^{-1} e_k f_j$};
		\node [style=none] (44) at (-3.25, 1.25) {};
		\node [style=none] (45) at (-3.25, 0.25) {};
		\node [style=none] (47) at (-3.25, -1.75) {};
		\node [style=none] (48) at (-3.25, -2.75) {};
		\node [style=none] (50) at (-1, 1.25) {};
		\node [style=none] (51) at (-1, 0.25) {};
		\node [style=none] (53) at (-1, -1.75) {};
		\node [style=none] (54) at (-1, -2.75) {};
		\node [style=hadamard] (55) at (0.5, 3) {$\epsilon$};
		\node [style=none] (59) at (2, 4.5) {};
		\node [style=hadamard] (60) at (2, 3.75) {$\epsilon$};
		\node [style=hadamard] (61) at (-3, 3) {-};
		\node [style=hadamard] (62) at (-2, 3) {$-1$};
	\end{pgfonlayer}
	\begin{pgfonlayer}{edgelayer}
		\draw (7.center) to (3);
		\draw (6.center) to (2);
		\draw (23) to (24);
		\draw (9.center) to (25);
		\draw (26) to (8.center);
		\draw (23) to (29);
		\draw (29) to (25);
		\draw (26) to (32);
		\draw (32) to (23);
		\draw (44.center) to (2);
		\draw (2) to (45.center);
		\draw (3) to (47.center);
		\draw (48.center) to (3);
		\draw (50.center) to (25);
		\draw [in=180, out=0] (51.center) to (26);
		\draw (25) to (53.center);
		\draw (54.center) to (26);
		\draw [in=180, out=30] (2) to (10);
		\draw [in=210, out=0] (10) to (24);
		\draw [in=15, out=-180] (11) to (3);
		\draw [in=0, out=-105] (24) to (11);
		\draw (59.center) to (24);
		\draw (15.center) to (61);
		\draw (61) to (23);
	\end{pgfonlayer}
\end{tikzpicture}

%% file: figures/pivot-helper11.tikz
\begin{tikzpicture}
	\begin{pgfonlayer}{nodelayer}
		\node [style=Z dot] (2) at (-5.5, 1) {};
		\node [style=Z dot] (3) at (-5.5, -1.5) {};
		\node [style=none] (6) at (-6.5, 1) {};
		\node [style=none] (7) at (-6.5, -1.5) {};
		\node [style=none] (8) at (5, -1.5) {};
		\node [style=none] (9) at (5, 1) {};
		\node [style=hadamard] (10) at (-1.5, 2.25) {\scriptsize $e_1$};
		\node [style=hadamard] (11) at (-1.5, -0.75) {\scriptsize $e_k$};
		\node [style=none] (15) at (-1, 4.5) {};
		\node [style=none] (17) at (-5.5, 0) {$\vdots$};
		\node [style=none] (18) at (4, 0) {$\vdots$};
		\node [style=hadamard] (20) at (-1, 3.75) {\scriptsize $\epsilon$};
		\node [style=Z dot] (23) at (-1, 3) {};
		\node [style=Z dot] (24) at (2, 3) {};
		\node [style=Z dot] (25) at (4, 1) {};
		\node [style=Z dot] (26) at (4, -1.5) {};
		\node [style=hadamard] (29) at (2.75, 1.75) {\scriptsize $-f_1$};
		\node [style=hadamard] (32) at (3, -0.5) {\scriptsize $-f_j$};
		\node [style=hadamard] (39) at (-2.25, 1.25) {\scriptsize $-\epsilon^{-1} e_1 f_1$};
		\node [style=hadamard] (40) at (-2.25, 0.25) {\scriptsize $-\epsilon^{-1} e_1 f_j$};
		\node [style=hadamard] (41) at (-2.25, -1.75) {\scriptsize $-\epsilon^{-1} e_k f_1$};
		\node [style=hadamard] (42) at (-2.25, -2.75) {\scriptsize $-\epsilon^{-1} e_k f_j$};
		\node [style=none] (44) at (-3.25, 1.25) {};
		\node [style=none] (45) at (-3.25, 0.25) {};
		\node [style=none] (47) at (-3.25, -1.75) {};
		\node [style=none] (48) at (-3.25, -2.75) {};
		\node [style=none] (50) at (-1, 1.25) {};
		\node [style=none] (51) at (-1, 0.25) {};
		\node [style=none] (53) at (-1, -1.75) {};
		\node [style=none] (54) at (-1, -2.75) {};
		\node [style=hadamard] (55) at (1, 3) {$\epsilon$};
		\node [style=none] (59) at (2, 4.5) {};
		\node [style=hadamard] (60) at (2, 3.75) {$\epsilon$};
		\node [style=hadamard] (62) at (1.75, 0.5) {};
		\node [style=hadamard] (63) at (1.25, 2) {};
		\node [style=hadamard] (64) at (0.25, 3) {};
		\node [style=hadamard] (65) at (0, 1.75) {};
		\node [style=hadamard] (66) at (-0.25, 3) {};
		\node [style=hadamard] (67) at (0.75, 2) {};
	\end{pgfonlayer}
	\begin{pgfonlayer}{edgelayer}
		\draw (7.center) to (3);
		\draw (6.center) to (2);
		\draw (23) to (24);
		\draw (9.center) to (25);
		\draw (26) to (8.center);
		\draw (29) to (25);
		\draw (26) to (32);
		\draw (44.center) to (2);
		\draw (2) to (45.center);
		\draw (3) to (47.center);
		\draw (48.center) to (3);
		\draw (50.center) to (25);
		\draw [in=180, out=0] (51.center) to (26);
		\draw (25) to (53.center);
		\draw (54.center) to (26);
		\draw [in=180, out=30] (2) to (10);
		\draw [in=210, out=0] (10) to (24);
		\draw [in=15, out=-180] (11) to (3);
		\draw [in=0, out=-105] (24) to (11);
		\draw (59.center) to (24);
		\draw (62) to (32);
		\draw (63) to (29);
		\draw (15.center) to (23);
		\draw (23) to (67);
		\draw (67) to (63);
		\draw (23) to (65);
		\draw (65) to (62);
	\end{pgfonlayer}
\end{tikzpicture}

%% file: figures/pivot-helper12.tikz
\begin{tikzpicture}
	\begin{pgfonlayer}{nodelayer}
		\node [style=Z dot] (2) at (-5.5, 1) {};
		\node [style=Z dot] (3) at (-5.5, -1.5) {};
		\node [style=none] (6) at (-6.5, 1) {};
		\node [style=none] (7) at (-6.5, -1.5) {};
		\node [style=none] (8) at (5, -1.5) {};
		\node [style=none] (9) at (5, 1) {};
		\node [style=hadamard] (10) at (-1.5, 2.25) {\scriptsize $e_1$};
		\node [style=hadamard] (11) at (-1.5, -0.75) {\scriptsize $e_k$};
		\node [style=none] (15) at (-1, 4.5) {};
		\node [style=none] (17) at (-5.5, 0) {$\vdots$};
		\node [style=none] (18) at (4, 0) {$\vdots$};
		\node [style=hadamard] (20) at (-1, 3.75) {\scriptsize $\epsilon$};
		\node [style=Z dot] (23) at (-1, 3) {};
		\node [style=Z dot] (24) at (2, 3) {};
		\node [style=Z dot] (25) at (4, 1) {};
		\node [style=Z dot] (26) at (4, -1.5) {};
		\node [style=hadamard] (29) at (2.75, 1.75) {\scriptsize $f_1$};
		\node [style=hadamard] (32) at (3, -0.5) {\scriptsize $f_j$};
		\node [style=hadamard] (39) at (-2.25, 1.25) {\scriptsize $-\epsilon^{-1} e_1 f_1$};
		\node [style=hadamard] (40) at (-2.25, 0.25) {\scriptsize $-\epsilon^{-1} e_1 f_j$};
		\node [style=hadamard] (41) at (-2.25, -1.75) {\scriptsize $-\epsilon^{-1} e_k f_1$};
		\node [style=hadamard] (42) at (-2.25, -2.75) {\scriptsize $-\epsilon^{-1} e_k f_j$};
		\node [style=none] (44) at (-3.25, 1.25) {};
		\node [style=none] (45) at (-3.25, 0.25) {};
		\node [style=none] (47) at (-3.25, -1.75) {};
		\node [style=none] (48) at (-3.25, -2.75) {};
		\node [style=none] (50) at (-1, 1.25) {};
		\node [style=none] (51) at (-1, 0.25) {};
		\node [style=none] (53) at (-1, -1.75) {};
		\node [style=none] (54) at (-1, -2.75) {};
		\node [style=hadamard] (55) at (0.5, 3) {$-\epsilon$};
		\node [style=none] (59) at (2, 4.5) {};
		\node [style=hadamard] (60) at (2, 3.75) {$\epsilon$};
	\end{pgfonlayer}
	\begin{pgfonlayer}{edgelayer}
		\draw (7.center) to (3);
		\draw (6.center) to (2);
		\draw (15.center) to (23);
		\draw (23) to (24);
		\draw (9.center) to (25);
		\draw (26) to (8.center);
		\draw (23) to (29);
		\draw (29) to (25);
		\draw (26) to (32);
		\draw (32) to (23);
		\draw (44.center) to (2);
		\draw (2) to (45.center);
		\draw (3) to (47.center);
		\draw (48.center) to (3);
		\draw (50.center) to (25);
		\draw [in=180, out=0] (51.center) to (26);
		\draw (25) to (53.center);
		\draw (54.center) to (26);
		\draw [in=180, out=30] (2) to (10);
		\draw [in=210, out=0] (10) to (24);
		\draw [in=15, out=-180] (11) to (3);
		\draw [in=0, out=-105] (24) to (11);
		\draw (59.center) to (24);
	\end{pgfonlayer}
\end{tikzpicture}

%% file: figures/pivot-start.tikz
\begin{tikzpicture}
	\begin{pgfonlayer}{nodelayer}
		\node [style=Z dot] (0) at (0, 1.5) {};
		\node [style=Z dot] (1) at (0, -1.5) {};
		\node [style=Z dot] (2) at (-3, 2.5) {};
		\node [style=Z dot] (3) at (-3, 0.5) {};
		\node [style=Z dot] (4) at (-3, -2.5) {};
		\node [style=none] (6) at (-4, 3) {};
		\node [style=none] (7) at (-4, 1) {};
		\node [style=none] (9) at (-4, -2) {};
		\node [style=hadamard] (10) at (-2.25, 2.5) {\scriptsize$e_1$};
		\node [style=hadamard] (11) at (-2.25, 0.75) {\scriptsize$e_2$};
		\node [style=hadamard] (12) at (0, 0) {\scriptsize$\epsilon$};
		\node [style=hadamard] (13) at (-1, -0.25) {\scriptsize$f_1$};
		\node [style=hadamard] (14) at (-1, -1.25) {\scriptsize$f_2$};
		\node [style=none] (15) at (1, 1.5) {};
		\node [style=none] (16) at (1, -1.5) {};
		\node [style=none] (18) at (-3, -0.75) {$\vdots$};
		\node [style=hadamard] (19) at (-2.25, -1.25) {\scriptsize$e_k$};
		\node [style=hadamard] (20) at (-1, -2.25) {\scriptsize$f_k$};
		\node [style=none] (21) at (-4, 2) {};
		\node [style=none] (22) at (-4, 0) {};
		\node [style=none] (23) at (-4, -3) {};
		\node [style=none] (24) at (-3.75, 2.75) {$\vdots$};
		\node [style=none] (25) at (-3.75, 0.75) {$\vdots$};
		\node [style=none] (26) at (-3.75, -2.25) {$\vdots$};
	\end{pgfonlayer}
	\begin{pgfonlayer}{edgelayer}
		\draw [bend left=15] (7.center) to (3);
		\draw (3) to (11);
		\draw (11) to (0);
		\draw (0) to (12);
		\draw (12) to (1);
		\draw (1) to (14);
		\draw [bend left=15] (9.center) to (4);
		\draw (13) to (1);
		\draw [bend left=15] (6.center) to (2);
		\draw (2) to (10);
		\draw (10) to (0);
		\draw (15.center) to (0);
		\draw (16.center) to (1);
		\draw [in=-60, out=105] (13) to (2);
		\draw (14) to (3);
		\draw [bend right=15] (0) to (19);
		\draw (19) to (4);
		\draw (1) to (20);
		\draw (20) to (4);
		\draw [bend left=15] (4) to (23.center);
		\draw [bend left=15] (3) to (22.center);
		\draw [bend left=15] (2) to (21.center);
	\end{pgfonlayer}
\end{tikzpicture}

%% file: figures/pivot-end.tikz
\begin{tikzpicture}
	\begin{pgfonlayer}{nodelayer}
		\node [style=Z dot] (0) at (0, 1.25) {};
		\node [style=Z dot] (1) at (0, -1.25) {};
		\node [style=Z phase dot] (2) at (-3, 3.5) {$0,\alpha_1$};
		\node [style=Z phase dot] (3) at (-4, 0.5) {$0,\alpha_2$};
		\node [style=Z phase dot] (4) at (-3, -3.5) {$0,\alpha_k$};
		\node [style=none] (6) at (-4.25, 4) {};
		\node [style=none] (7) at (-5.25, 1) {};
		\node [style=none] (9) at (-4.25, -3) {};
		\node [style=hadamard] (10) at (-1.5, 2.25) {\scriptsize$e_1$};
		\node [style=hadamard] (11) at (-2.25, 0.75) {\scriptsize$e_2$};
		\node [style=hadamard] (12) at (0, 0) {\scriptsize$-\epsilon$};
		\node [style=hadamard] (13) at (-1, -0.25) {\scriptsize$f_1$};
		\node [style=hadamard] (14) at (-1, -1.25) {\scriptsize$f_2$};
		\node [style=none] (15) at (1, 1.25) {};
		\node [style=none] (16) at (1, -1.25) {};
		\node [style=none] (18) at (-3.5, -0.75) {$\vdots$};
		\node [style=hadamard] (19) at (-2.25, -1.5) {\scriptsize$e_k$};
		\node [style=hadamard] (20) at (-1, -2.25) {\scriptsize$f_k$};
		\node [style=none] (21) at (-4.25, 3) {};
		\node [style=none] (22) at (-5.25, 0) {};
		\node [style=none] (23) at (-4.25, -4) {};
		\node [style=none] (24) at (-4, 3.75) {$\vdots$};
		\node [style=none] (25) at (-5, 0.75) {$\vdots$};
		\node [style=none] (26) at (-4, -3.25) {$\vdots$};
		\node [style=hadamard] (27) at (-3.75, 2.25) {\scriptsize$w_{1,2}$};
		\node [style=hadamard] (28) at (-3, 1.5) {\scriptsize$w_{1,k}$};
		\node [style=hadamard] (29) at (-4, -2) {\scriptsize$w_{2,k}$};
		\node [style=hadamard] (30) at (0.75, 1.25) {\scriptsize$\epsilon$};
		\node [style=hadamard] (31) at (0.75, -1.25) {\scriptsize$\epsilon$};
		\node [style=none] (32) at (2.25, 1.25) {};
		\node [style=none] (33) at (2.25, -1.25) {};
	\end{pgfonlayer}
	\begin{pgfonlayer}{edgelayer}
		\draw [bend left=15] (7.center) to (3);
		\draw (3) to (11);
		\draw (11) to (0);
		\draw (0) to (12);
		\draw (12) to (1);
		\draw (1) to (14);
		\draw [bend left=15] (9.center) to (4);
		\draw (13) to (1);
		\draw [bend left=15] (6.center) to (2);
		\draw (2) to (10);
		\draw (10) to (0);
		\draw (15.center) to (0);
		\draw (16.center) to (1);
		\draw [in=-60, out=120] (13) to (2);
		\draw (14) to (3);
		\draw [bend right] (0) to (19);
		\draw (19) to (4);
		\draw (1) to (20);
		\draw (20) to (4);
		\draw [bend left=15] (4) to (23.center);
		\draw [bend left=15] (3) to (22.center);
		\draw [bend left=15] (2) to (21.center);
		\draw (2) to (27);
		\draw (27) to (3);
		\draw (2) to (28);
		\draw (28) to (4);
		\draw (3) to (29);
		\draw (29) to (4);
		\draw [in=180, out=0] (15.center) to (33.center);
		\draw [in=0, out=180] (32.center) to (16.center);
	\end{pgfonlayer}
\end{tikzpicture}

%% file: figures/h-box-epsilon.tikz
\begin{tikzpicture}
	\begin{pgfonlayer}{nodelayer}
		\node [style=hadamard] (0) at (0, 0) {$\epsilon$};
		\node [style=none] (1) at (-0.75, 0) {};
		\node [style=none] (2) at (0.75, 0) {};
	\end{pgfonlayer}
	\begin{pgfonlayer}{edgelayer}
		\draw (1.center) to (2.center);
	\end{pgfonlayer}
\end{tikzpicture}

%% file: figures/phase.tikz
\begin{tikzpicture}
	\begin{pgfonlayer}{nodelayer}
		\node [style=none] (1) at (-1, 0) {};
		\node [style=none] (2) at (1, 0) {};
		\node [style=Z phase dot] (3) at (0, 0) {$0,\alpha_v$};
	\end{pgfonlayer}
	\begin{pgfonlayer}{edgelayer}
		\draw (1.center) to (2.center);
	\end{pgfonlayer}
\end{tikzpicture}

%% file: main.bbl
\begin{thebibliography}{10}
\providecommand{\bibitemdeclare}[2]{}
\providecommand{\surnamestart}{}
\providecommand{\surnameend}{}
\providecommand{\urlprefix}{Available at }
\providecommand{\url}[1]{\texttt{#1}}
\providecommand{\href}[2]{\texttt{#2}}
\providecommand{\urlalt}[2]{\href{#1}{#2}}
\providecommand{\doi}[1]{doi:\urlalt{https://doi.org/#1}{#1}}
\providecommand{\eprint}[1]{arXiv:\urlalt{https://arxiv.org/abs/#1}{#1}}
\providecommand{\bibinfo}[2]{#2}

\bibitemdeclare{article}{backensZXcalculusCompleteStabilizer2014}
\bibitem{backensZXcalculusCompleteStabilizer2014}
\bibinfo{author}{Miriam \surnamestart Backens\surnameend}
  (\bibinfo{year}{2014}): \emph{\bibinfo{title}{The {{ZX-calculus}} Is Complete
  for Stabilizer Quantum Mechanics}}.
\newblock {\slshape \bibinfo{journal}{New Journal of Physics}}
  \bibinfo{volume}{16}(\bibinfo{number}{9}), p. \bibinfo{pages}{093021},
  \doi{10.1088/1367-2630/16/9/093021}.

\bibitemdeclare{misc}{backensGenerating2026}
\bibitem{backensGenerating2026}
\bibinfo{author}{Miriam \surnamestart Backens\surnameend}
  (\bibinfo{year}{2026}): \emph{\bibinfo{title}{Generating One-Way Computations
  with Flow: Flow-Preserving Rewriting That Ignores the Interpretation}},
  \doi{10.48550/arXiv.2607.03250}.

\bibitemdeclare{article}{backensThereBackAgain2021}
\bibitem{backensThereBackAgain2021}
\bibinfo{author}{Miriam \surnamestart Backens\surnameend},
  \bibinfo{author}{Hector \surnamestart {Miller-Bakewell}\surnameend},
  \bibinfo{author}{Giovanni \surnamestart {de Felice}\surnameend},
  \bibinfo{author}{Leo \surnamestart Lobski\surnameend} \&
  \bibinfo{author}{John \surnamestart {van de Wetering}\surnameend}
  (\bibinfo{year}{2021}): \emph{\bibinfo{title}{There and Back Again: {{A}}
  Circuit Extraction Tale}}.
\newblock {\slshape \bibinfo{journal}{Quantum}} \bibinfo{volume}{5}, p.
  \bibinfo{pages}{421}, \doi{10.22331/q-2021-03-25-421}.

\bibitemdeclare{misc}{backensCompleteness2026}
\bibitem{backensCompleteness2026}
\bibinfo{author}{Miriam \surnamestart Backens\surnameend} \&
  \bibinfo{author}{Simon \surnamestart Perdrix\surnameend}
  (\bibinfo{year}{2026}): \emph{\bibinfo{title}{Completeness for
  flow-preserving rewrite rules}}.
\newblock \bibinfo{note}{To appear}.

\bibitemdeclare{article}{perezBackens2025}
\bibitem{perezBackens2025}
\bibinfo{author}{Miriam \surnamestart Backens\surnameend} \&
  \bibinfo{author}{Thomas \surnamestart Perez\surnameend}
  (\bibinfo{year}{2025}): \emph{\bibinfo{title}{Inserting Planar-Measured
  Qubits into MBQC Patterns while Preserving Flow}}.
\newblock {\slshape \bibinfo{journal}{Electronic Proceedings in Theoretical
  Computer Science}} \bibinfo{volume}{426}, p. \bibinfo{pages}{100–126},
  \doi{10.4204/eptcs.426.4}.

\bibitemdeclare{article}{bahramgiriGraph2006}
\bibitem{bahramgiriGraph2006}
\bibinfo{author}{Mohsen \surnamestart Bahramgiri\surnameend} \&
  \bibinfo{author}{Salman \surnamestart Beigi\surnameend}
  (\bibinfo{year}{2006}): \emph{\bibinfo{title}{Graph {{States Under}} the
  {{Action}} of {{Local Clifford Group}} in {{Non-Binary Case}}}}.
\newblock \doi{10.48550/arXiv.quant-ph/0610267}.

\bibitemdeclare{inproceedings}{boothCompleteZXcalculiStabiliser2022}
\bibitem{boothCompleteZXcalculiStabiliser2022}
\bibinfo{author}{Robert~I. \surnamestart Booth\surnameend} \&
  \bibinfo{author}{Titouan \surnamestart Carette\surnameend}
  (\bibinfo{year}{2022}): \emph{\bibinfo{title}{Complete {{ZX-Calculi}} for the
  {{Stabiliser Fragment}} in {{Odd Prime Dimensions}}}}.
\newblock In \bibinfo{editor}{Stefan \surnamestart Szeider\surnameend},
  \bibinfo{editor}{Robert \surnamestart Ganian\surnameend} \&
  \bibinfo{editor}{Alexandra \surnamestart Silva\surnameend}, editors:
  {\slshape \bibinfo{booktitle}{47th {{International Symposium}} on
  {{Mathematical Foundations}} of {{Computer Science}} ({{MFCS}} 2022)}},
  {\slshape \bibinfo{series}{Leibniz {{International Proceedings}} in
  {{Informatics}} ({{LIPIcs}})}} \bibinfo{volume}{241},
  \bibinfo{publisher}{Schloss Dagstuhl -- Leibniz-Zentrum f{\"u}r Informatik},
  \bibinfo{address}{Dagstuhl, Germany}, pp. \bibinfo{pages}{24:1--24:15},
  \doi{10.4230/LIPIcs.MFCS.2022.24}.
\newblock \bibinfo{note}{Extended version at
  arXiv:\href{http://arxiv.org/abs/2204.12531}{2204.12531}}.

\bibitemdeclare{article}{boothOutcomeDeterminismMeasurementbased2023}
\bibitem{boothOutcomeDeterminismMeasurementbased2023}
\bibinfo{author}{Robert~I. \surnamestart Booth\surnameend},
  \bibinfo{author}{Aleks \surnamestart Kissinger\surnameend},
  \bibinfo{author}{Damian \surnamestart Markham\surnameend},
  \bibinfo{author}{Cl{\'e}ment \surnamestart Meignant\surnameend} \&
  \bibinfo{author}{Simon \surnamestart Perdrix\surnameend}
  (\bibinfo{year}{2023}): \emph{\bibinfo{title}{Outcome Determinism in
  Measurement-Based Quantum Computation with Qudits}}.
\newblock {\slshape \bibinfo{journal}{Journal of Physics A: Mathematical and
  Theoretical}} \bibinfo{volume}{56}(\bibinfo{number}{11}), p.
  \bibinfo{pages}{115303}, \doi{10.1088/1751-8121/acbace}.

\bibitemdeclare{article}{boothCVflow2023}
\bibitem{boothCVflow2023}
\bibinfo{author}{Robert~I. \surnamestart Booth\surnameend} \&
  \bibinfo{author}{Damian \surnamestart Markham\surnameend}
  (\bibinfo{year}{2023}): \emph{\bibinfo{title}{Flow conditions for continuous
  variable measurement-based quantum computing}}.
\newblock {\slshape \bibinfo{journal}{Quantum}} \bibinfo{volume}{7}, p.
  \bibinfo{pages}{1146}, \doi{10.22331/q-2023-10-19-1146}.

\bibitemdeclare{phdthesis}{boothMeasurementbasedQuantumComputation2022}
\bibitem{boothMeasurementbasedQuantumComputation2022}
\bibinfo{author}{Robert~Ivan \surnamestart Booth\surnameend}
  (\bibinfo{year}{2022}): \emph{\bibinfo{title}{Measurement-Based Quantum
  Computation beyond Qubits}}.
\newblock Ph.D. thesis, \bibinfo{school}{Sorbonne Universit\'e}.
\newblock \urlprefix\url{https://theses.hal.science/tel-03867179}.

\bibitemdeclare{article}{browneGeneralizedFlowDeterminism2007}
\bibitem{browneGeneralizedFlowDeterminism2007}
\bibinfo{author}{Daniel~E. \surnamestart Browne\surnameend},
  \bibinfo{author}{Elham \surnamestart Kashefi\surnameend},
  \bibinfo{author}{Mehdi \surnamestart Mhalla\surnameend} \&
  \bibinfo{author}{Simon \surnamestart Perdrix\surnameend}
  (\bibinfo{year}{2007}): \emph{\bibinfo{title}{Generalized Flow and
  Determinism in Measurement-Based Quantum Computation}}.
\newblock {\slshape \bibinfo{journal}{New Journal of Physics}}
  \bibinfo{volume}{9}(\bibinfo{number}{8}), p. \bibinfo{pages}{250},
  \doi{10.1088/1367-2630/9/8/250}.

\bibitemdeclare{article}{caoMulti-agent2023}
\bibitem{caoMulti-agent2023}
\bibinfo{author}{Shuxiang \surnamestart Cao\surnameend} (\bibinfo{year}{2023}):
  \emph{\bibinfo{title}{Multi-Agent Blind Quantum Computation without Universal
  Cluster States}}.
\newblock {\slshape \bibinfo{journal}{New Journal of Physics}}
  \bibinfo{volume}{25}(\bibinfo{number}{10}), p. \bibinfo{pages}{103028},
  \doi{10.1088/1367-2630/acfab6}.

\bibitemdeclare{book}{cormenIntroductionAlgorithmsFourth2022}
\bibitem{cormenIntroductionAlgorithmsFourth2022}
\bibinfo{author}{Thomas~H. \surnamestart Cormen\surnameend},
  \bibinfo{author}{Charles~E. \surnamestart Leiserson\surnameend},
  \bibinfo{author}{Ronald~L. \surnamestart Rivest\surnameend} \&
  \bibinfo{author}{Clifford \surnamestart Stein\surnameend}
  (\bibinfo{year}{2022}): \emph{\bibinfo{title}{Introduction to {{Algorithms}},
  Fourth Edition}}.
\newblock \bibinfo{publisher}{MIT Press}.

\bibitemdeclare{article}{danosDeterminismOnewayModel2006}
\bibitem{danosDeterminismOnewayModel2006}
\bibinfo{author}{Vincent \surnamestart Danos\surnameend} \&
  \bibinfo{author}{Elham \surnamestart Kashefi\surnameend}
  (\bibinfo{year}{2006}): \emph{\bibinfo{title}{Determinism in the One-Way
  Model}}.
\newblock {\slshape \bibinfo{journal}{Physical Review A}}
  \bibinfo{volume}{74}(\bibinfo{number}{5}), p. \bibinfo{pages}{052310},
  \doi{10.1103/PhysRevA.74.052310}.

\bibitemdeclare{misc}{debeaudrapCompleteAlgorithmFind2007}
\bibitem{debeaudrapCompleteAlgorithmFind2007}
\bibinfo{author}{N.~\surnamestart {de Beaudrap}\surnameend}
  (\bibinfo{year}{2007}): \emph{\bibinfo{title}{A Complete Algorithm to Find
  Flows in the One-Way Measurement Model}},
  \doi{10.48550/arXiv.quant-ph/0603072}.

\bibitemdeclare{article}{deBeaudrapFinding2008}
\bibitem{deBeaudrapFinding2008}
\bibinfo{author}{N.~\surnamestart {de Beaudrap}\surnameend}
  (\bibinfo{year}{2008}): \emph{\bibinfo{title}{Finding Flows in the One-Way
  Measurement Model}}.
\newblock {\slshape \bibinfo{journal}{Physical Review A}}
  \bibinfo{volume}{77}(\bibinfo{number}{2}), p. \bibinfo{pages}{022328},
  \doi{10.1103/PhysRevA.77.022328}.

\bibitemdeclare{article}{duncanGraphtheoreticSimplificationQuantum2020}
\bibitem{duncanGraphtheoreticSimplificationQuantum2020}
\bibinfo{author}{Ross \surnamestart Duncan\surnameend}, \bibinfo{author}{Aleks
  \surnamestart Kissinger\surnameend}, \bibinfo{author}{Simon \surnamestart
  Perdrix\surnameend} \& \bibinfo{author}{John \surnamestart {van de
  Wetering}\surnameend} (\bibinfo{year}{2020}):
  \emph{\bibinfo{title}{Graph-Theoretic {{Simplification}} of {{Quantum
  Circuits}} with the {{ZX-calculus}}}}.
\newblock {\slshape \bibinfo{journal}{Quantum}} \bibinfo{volume}{4}, p.
  \bibinfo{pages}{279}, \doi{10.22331/q-2020-06-04-279}.

\bibitemdeclare{article}{ewenApplication2025}
\bibitem{ewenApplication2025}
\bibinfo{author}{Tom \surnamestart Ewen\surnameend}, \bibinfo{author}{Ivica
  \surnamestart Turkalj\surnameend}, \bibinfo{author}{Patrick \surnamestart
  Holzer\surnameend} \& \bibinfo{author}{Mark-Oliver \surnamestart
  Wolf\surnameend} (\bibinfo{year}{2025}): \emph{\bibinfo{title}{Application of
  {ZX}-calculus to quantum architecture search}}.
\newblock {\slshape \bibinfo{journal}{Quantum Machine Intelligence}}
  \bibinfo{volume}{7}(\bibinfo{number}{1}), p.~\bibinfo{pages}{34},
  \doi{10.1007/s42484-025-00264-6}.

\bibitemdeclare{article}{gashkovComplexityComputationFinite2013}
\bibitem{gashkovComplexityComputationFinite2013}
\bibinfo{author}{Sergey~B. \surnamestart Gashkov\surnameend} \&
  \bibinfo{author}{Igor~S. \surnamestart Sergeev\surnameend}
  (\bibinfo{year}{2013}): \emph{\bibinfo{title}{Complexity of Computation in
  Finite Fields}}.
\newblock {\slshape \bibinfo{journal}{Journal of Mathematical Sciences}}
  \bibinfo{volume}{191}(\bibinfo{number}{5}), pp. \bibinfo{pages}{661--685},
  \doi{10.1007/s10958-013-1350-5}.

\bibitemdeclare{article}{gheorghiuStandardFormQudit2014}
\bibitem{gheorghiuStandardFormQudit2014}
\bibinfo{author}{Vlad \surnamestart Gheorghiu\surnameend}
  (\bibinfo{year}{2014}): \emph{\bibinfo{title}{Standard form of qudit
  stabilizer groups}}.
\newblock {\slshape \bibinfo{journal}{Physics Letters A}}
  \bibinfo{volume}{378}(\bibinfo{number}{5}), pp. \bibinfo{pages}{505--509},
  \doi{10.1016/j.physleta.2013.12.009}.

\bibitemdeclare{article}{harveyIntegerMultiplicationTime2021}
\bibitem{harveyIntegerMultiplicationTime2021}
\bibinfo{author}{David \surnamestart Harvey\surnameend} \&
  \bibinfo{author}{Joris \surnamestart van~der Hoeven\surnameend}
  (\bibinfo{year}{2021}): \emph{\bibinfo{title}{Integer Multiplication in Time
  $O(n\log n)$}}.
\newblock {\slshape \bibinfo{journal}{Annals of Mathematics}}
  \bibinfo{volume}{193}(\bibinfo{number}{2}), pp. \bibinfo{pages}{563--617},
  \doi{10.4007/annals.2021.193.2.4}.

\bibitemdeclare{article}{harveyPolyMultFiniteFields2022}
\bibitem{harveyPolyMultFiniteFields2022}
\bibinfo{author}{David \surnamestart Harvey\surnameend} \&
  \bibinfo{author}{Joris \surnamestart van~der Hoeven\surnameend}
  (\bibinfo{year}{2022}): \emph{\bibinfo{title}{Polynomial Multiplication over
  Finite Fields in Time \( O(n \log n) \)}}.
\newblock {\slshape \bibinfo{journal}{Journal of the ACM}}
  \bibinfo{volume}{69}(\bibinfo{number}{2}), \doi{10.1145/3505584}.

\bibitemdeclare{misc}{heyfronQuantum2019}
\bibitem{heyfronQuantum2019}
\bibinfo{author}{Luke~E. \surnamestart Heyfron\surnameend} \&
  \bibinfo{author}{Earl \surnamestart Campbell\surnameend}
  (\bibinfo{year}{2019}): \emph{\bibinfo{title}{A quantum compiler for qudits
  of prime dimension greater than 3}}, \doi{10.48550/arXiv.1902.05634}.

\bibitemdeclare{misc}{holkerCausal2024}
\bibitem{holkerCausal2024}
\bibinfo{author}{Calum \surnamestart Holker\surnameend} (\bibinfo{year}{2023}):
  \emph{\bibinfo{title}{Causal Flow Preserving Optimisation of Quantum Circuits
  in the {{ZX-calculus}}}}, \doi{10.48550/arXiv.2312.02793}.

\bibitemdeclare{article}{hostensStabilizer2005}
\bibitem{hostensStabilizer2005}
\bibinfo{author}{Erik \surnamestart Hostens\surnameend},
  \bibinfo{author}{Jeroen \surnamestart Dehaene\surnameend} \&
  \bibinfo{author}{Bart \surnamestart De~Moor\surnameend}
  (\bibinfo{year}{2005}): \emph{\bibinfo{title}{Stabilizer states and
  {Clifford} operations for systems of arbitrary dimensions and modular
  arithmetic}}.
\newblock {\slshape \bibinfo{journal}{Physical Review A}}
  \bibinfo{volume}{71}(\bibinfo{number}{4}), p. \bibinfo{pages}{042315},
  \doi{10.1103/PhysRevA.71.042315}.

\bibitemdeclare{article}{kissingerReducingTcountZXcalculus2020}
\bibitem{kissingerReducingTcountZXcalculus2020}
\bibinfo{author}{Aleks \surnamestart Kissinger\surnameend} \&
  \bibinfo{author}{John \surnamestart {van de Wetering}\surnameend}
  (\bibinfo{year}{2020}): \emph{\bibinfo{title}{Reducing {{T-count}} with the
  {{ZX-calculus}}}}.
\newblock {\slshape \bibinfo{journal}{Physical Review A}}
  \bibinfo{volume}{102}(\bibinfo{number}{2}), p. \bibinfo{pages}{022406},
  \doi{10.1103/PhysRevA.102.022406}.
\newblock \eprint{1903.10477}.

\bibitemdeclare{article}{mcelvanneyCompleteFlowPreservingRewrite2023}
\bibitem{mcelvanneyCompleteFlowPreservingRewrite2023}
\bibinfo{author}{Tommy \surnamestart McElvanney\surnameend} \&
  \bibinfo{author}{Miriam \surnamestart Backens\surnameend}
  (\bibinfo{year}{2023}): \emph{\bibinfo{title}{Complete {{Flow-Preserving
  Rewrite Rules}} for {{MBQC Patterns}} with {{Pauli Measurements}}}}.
\newblock {\slshape \bibinfo{journal}{Electronic Proceedings in Theoretical
  Computer Science}} \bibinfo{volume}{394}, pp. \bibinfo{pages}{66--82},
  \doi{10.4204/EPTCS.394.5}.

\bibitemdeclare{article}{mcelvanneyFlowpreservingZXcalculusRewrite2023a}
\bibitem{mcelvanneyFlowpreservingZXcalculusRewrite2023a}
\bibinfo{author}{Tommy \surnamestart McElvanney\surnameend} \&
  \bibinfo{author}{Miriam \surnamestart Backens\surnameend}
  (\bibinfo{year}{2023}): \emph{\bibinfo{title}{Flow-Preserving {{ZX-calculus
  Rewrite Rules}} for {{Optimisation}} and {{Obfuscation}}}}.
\newblock {\slshape \bibinfo{journal}{Electronic Proceedings in Theoretical
  Computer Science}} \bibinfo{volume}{384}, pp. \bibinfo{pages}{203--219},
  \doi{10.4204/EPTCS.384.12}.

\bibitemdeclare{inproceedings}{mhallaWhichGraphStates2014a}
\bibitem{mhallaWhichGraphStates2014a}
\bibinfo{author}{Mehdi \surnamestart Mhalla\surnameend}, \bibinfo{author}{Mio
  \surnamestart Murao\surnameend}, \bibinfo{author}{Simon \surnamestart
  Perdrix\surnameend}, \bibinfo{author}{Masato \surnamestart Someya\surnameend}
  \& \bibinfo{author}{Peter~S. \surnamestart Turner\surnameend}
  (\bibinfo{year}{2014}): \emph{\bibinfo{title}{Which Graph States Are Useful
  for Quantum Information Processing?}}
\newblock In \bibinfo{editor}{Dave \surnamestart Bacon\surnameend},
  \bibinfo{editor}{Miguel \surnamestart {Martin-Delgado}\surnameend} \&
  \bibinfo{editor}{Martin \surnamestart Roetteler\surnameend}, editors:
  {\slshape \bibinfo{booktitle}{Theory of Quantum Computation, Communication,
  and Cryptography}}, \bibinfo{publisher}{Springer Berlin Heidelberg},
  \bibinfo{address}{Berlin, Heidelberg}, pp. \bibinfo{pages}{174--187},
  \doi{10.1007/978-3-642-54429-3_12}.

\bibitemdeclare{article}{mhallaShadowPauliFlow2025}
\bibitem{mhallaShadowPauliFlow2025}
\bibinfo{author}{Mehdi \surnamestart Mhalla\surnameend}, \bibinfo{author}{Simon
  \surnamestart Perdrix\surnameend} \& \bibinfo{author}{Luc \surnamestart
  Sanselme\surnameend} (\bibinfo{year}{2026}): \emph{\bibinfo{title}{Shadow
  {{Pauli Flow}}: {{Characterising Determinism}} in {{MBQCs}} Involving {{Pauli
  Measurements}}}}.
\newblock {\slshape \bibinfo{journal}{Quantum}} \bibinfo{volume}{10}, p.
  \bibinfo{pages}{2163}, \doi{10.22331/q-2026-07-16-2163}.

\bibitemdeclare{article}{mitosekPauliFlowOpen2024}
\bibitem{mitosekPauliFlowOpen2024}
\bibinfo{author}{Piotr \surnamestart Mitosek\surnameend}
  (\bibinfo{year}{2024}): \emph{\bibinfo{title}{Pauli {{Flow}} on {{Open
  Graphs}} with {{Unknown Measurement Labels}}}}.
\newblock {\slshape \bibinfo{journal}{Electronic Proceedings in Theoretical
  Computer Science}} \bibinfo{volume}{406}, pp. \bibinfo{pages}{117--136},
  \doi{10.4204/EPTCS.406.6}.

\bibitemdeclare{article}{mitosekAlgebraicInterpretationPauli2026}
\bibitem{mitosekAlgebraicInterpretationPauli2026}
\bibinfo{author}{Piotr \surnamestart Mitosek\surnameend} \&
  \bibinfo{author}{Miriam \surnamestart Backens\surnameend}
  (\bibinfo{year}{2026}): \emph{\bibinfo{title}{An algebraic formulation of
  Pauli flow, leading to faster flow-finding algorithms}}.
\newblock {\slshape \bibinfo{journal}{Journal of Physics A: Mathematical and
  Theoretical}} \bibinfo{volume}{59}(\bibinfo{number}{3}), p.
  \bibinfo{pages}{035301}, \doi{10.1088/1751-8121/ae2999}.

\bibitemdeclare{article}{paesaniScheme2021}
\bibitem{paesaniScheme2021}
\bibinfo{author}{Stefano \surnamestart Paesani\surnameend},
  \bibinfo{author}{Jacob F.~F. \surnamestart Bulmer\surnameend},
  \bibinfo{author}{Alex~E. \surnamestart Jones\surnameend},
  \bibinfo{author}{Raffaele \surnamestart Santagati\surnameend} \&
  \bibinfo{author}{Anthony \surnamestart Laing\surnameend}
  (\bibinfo{year}{2021}): \emph{\bibinfo{title}{Scheme for {Universal}
  {High}-{Dimensional} {Quantum} {Computation} with {Linear} {Optics}}}.
\newblock {\slshape \bibinfo{journal}{Physical Review Letters}}
  \bibinfo{volume}{126}(\bibinfo{number}{23}), p. \bibinfo{pages}{230504},
  \doi{10.1103/PhysRevLett.126.230504}.

\bibitemdeclare{article}{poorQupitStabiliserZXtravaganza2023}
\bibitem{poorQupitStabiliserZXtravaganza2023}
\bibinfo{author}{Boldizsár \surnamestart Poór\surnameend},
  \bibinfo{author}{Robert~I. \surnamestart Booth\surnameend},
  \bibinfo{author}{Titouan \surnamestart Carette\surnameend},
  \bibinfo{author}{John \surnamestart Van De~Wetering\surnameend} \&
  \bibinfo{author}{Lia \surnamestart Yeh\surnameend} (\bibinfo{year}{2023}):
  \emph{\bibinfo{title}{The {Qupit} {Stabiliser} {ZX}-travaganza: {Simplified}
  {Axioms}, {Normal} {Forms} and {Graph}-{Theoretic} {Simplification}}}.
\newblock {\slshape \bibinfo{journal}{Electronic Proceedings in Theoretical
  Computer Science}} \bibinfo{volume}{384}, pp. \bibinfo{pages}{220--264},
  \doi{10.4204/EPTCS.384.13}.

\bibitemdeclare{article}{randallEfficientGenerationRandom1993}
\bibitem{randallEfficientGenerationRandom1993}
\bibinfo{author}{Dana \surnamestart Randall\surnameend} (\bibinfo{year}{1993}):
  \emph{\bibinfo{title}{Efficient Generation of Random Nonsingular Matrices}}.
\newblock {\slshape \bibinfo{journal}{Random Structures \& Algorithms}}
  \bibinfo{volume}{4}(\bibinfo{number}{1}), pp. \bibinfo{pages}{111--118},
  \doi{10.1002/rsa.3240040108}.

\bibitemdeclare{article}{raussendorfOneWayQuantumComputer2001}
\bibitem{raussendorfOneWayQuantumComputer2001}
\bibinfo{author}{Robert \surnamestart Raussendorf\surnameend} \&
  \bibinfo{author}{Hans~J. \surnamestart Briegel\surnameend}
  (\bibinfo{year}{2001}): \emph{\bibinfo{title}{A {{One-Way Quantum
  Computer}}}}.
\newblock {\slshape \bibinfo{journal}{Physical Review Letters}}
  \bibinfo{volume}{86}(\bibinfo{number}{22}), pp. \bibinfo{pages}{5188--5191},
  \doi{10.1103/PhysRevLett.86.5188}.

\bibitemdeclare{article}{romanovaMeasurement-based2026}
\bibitem{romanovaMeasurement-based2026}
\bibinfo{author}{Alena \surnamestart Romanova\surnameend} \&
  \bibinfo{author}{Wolfgang \surnamestart Dür\surnameend}
  (\bibinfo{year}{2026}): \emph{\bibinfo{title}{Measurement-based quantum
  computing with qudit stabilizer states}}.
\newblock {\slshape \bibinfo{journal}{Quantum Science and Technology}}
  \bibinfo{volume}{11}(\bibinfo{number}{1}), p. \bibinfo{pages}{015054},
  \doi{10.1088/2058-9565/ae3b6f}.

\bibitemdeclare{article}{scirihaCharacterizationSingularGraphs2007}
\bibitem{scirihaCharacterizationSingularGraphs2007}
\bibinfo{author}{Irene \surnamestart Sciriha\surnameend}
  (\bibinfo{year}{2007}): \emph{\bibinfo{title}{A Characterization of Singular
  Graphs}}.
\newblock {\slshape \bibinfo{journal}{The Electronic Journal of Linear
  Algebra}} \bibinfo{volume}{16}, pp. \bibinfo{pages}{451--462},
  \doi{10.13001/1081-3810.1215}.

\bibitemdeclare{article}{simmonsRelatingMeasurementPatterns2021}
\bibitem{simmonsRelatingMeasurementPatterns2021}
\bibinfo{author}{Will \surnamestart Simmons\surnameend} (\bibinfo{year}{2021}):
  \emph{\bibinfo{title}{Relating {{Measurement Patterns}} to {{Circuits}} via
  {{Pauli Flow}}}}.
\newblock {\slshape \bibinfo{journal}{Electronic Proceedings in Theoretical
  Computer Science}} \bibinfo{volume}{343}, pp. \bibinfo{pages}{50--101},
  \doi{10.4204/EPTCS.343.4}.

\bibitemdeclare{article}{staudacherReducing2QuBitGate2023}
\bibitem{staudacherReducing2QuBitGate2023}
\bibinfo{author}{Korbinian \surnamestart Staudacher\surnameend},
  \bibinfo{author}{Tobias \surnamestart Guggemos\surnameend},
  \bibinfo{author}{Sophia \surnamestart {Grundner-Culemann}\surnameend} \&
  \bibinfo{author}{Wolfgang \surnamestart Gehrke\surnameend}
  (\bibinfo{year}{2023}): \emph{\bibinfo{title}{Reducing 2-{{QuBit Gate Count}}
  for {{ZX-Calculus}} Based {{Quantum Circuit Optimization}}}}.
\newblock {\slshape \bibinfo{journal}{Electronic Proceedings in Theoretical
  Computer Science}} \bibinfo{volume}{394}, pp. \bibinfo{pages}{29--45},
  \doi{10.4204/EPTCS.394.3}.

\bibitemdeclare{misc}{sutherlandLectureNotes2025}
\bibitem{sutherlandLectureNotes2025}
\bibinfo{author}{Andrew \surnamestart Sutherland\surnameend}
  (\bibinfo{year}{2025}): \emph{\bibinfo{title}{Lecture {Notes}: Mathematics
  18.783 -- {Elliptic} {Curves}, {Lecture} \#3}}.
\newblock
  \urlprefix\url{https://math.mit.edu/classes/18.783/2025/LectureNotes3.pdf}.
\newblock \bibinfo{note}{Accessed: 28th February 2025}.

\bibitemdeclare{misc}{townsend-teagueSimplification2022}
\bibitem{townsend-teagueSimplification2022}
\bibinfo{author}{Alex \surnamestart {Townsend-Teague}\surnameend} \&
  \bibinfo{author}{Konstantinos \surnamestart Meichanetzidis\surnameend}
  (\bibinfo{year}{2022}): \emph{\bibinfo{title}{Simplification {{Strategies}}
  for the {{Qutrit ZX-Calculus}}}}, \doi{10.48550/arXiv.2103.06914}.

\bibitemdeclare{misc}{yangQuantum2025}
\bibitem{yangQuantum2025}
\bibinfo{author}{Shuai \surnamestart Yang\surnameend}, \bibinfo{author}{Lihao
  \surnamestart Xu\surnameend}, \bibinfo{author}{Guojing \surnamestart
  Tian\surnameend} \& \bibinfo{author}{Xiaoming \surnamestart Sun\surnameend}
  (\bibinfo{year}{2025}): \emph{\bibinfo{title}{Quantum circuit synthesis with
  qudit phase gadget method}}, \doi{10.48550/arXiv.2504.12710}.

\end{thebibliography}
